\documentclass[onecolumn]{IEEEtran}

\usepackage[utf8]{inputenc}
\usepackage{amsmath}
\usepackage{amsthm}
\usepackage{amsfonts}
\usepackage{bbm}
\usepackage{hyperref}
\hypersetup{pdfborderstyle={/W .5},pdflinkmargin={.5pt}}
\usepackage{url}
\usepackage{stackengine}
\usepackage{tikz}
\usetikzlibrary{spy,calc}
\usepackage{pgfplots}
\pgfplotsset{compat=1.13}
\usepackage[nonumberlist,acronym]{glossaries}
\glsdisablehyper
\usepackage[noadjust]{cite}
\usepackage{aligned-overset}
\usepackage{todonotes}
\usepackage{graphicx}
\usepackage{xcolor}
\usepackage{xifthen}
\usepackage[shortlabels]{enumitem}
\usepackage{subcaption}
\usepackage{siunitx}

\newcommand\blankfootnote[1]{%
  \begin{NoHyper}
  \let\svthefootnote\thefootnote%
  \let\thefootnote\relax\footnotetext{#1}%
  \let\thefootnote\svthefootnote%
  \end{NoHyper}
}

\newcommand{\inputAlphabet}{\mathcal{X}}
\newcommand{\inputRV}{X}

\newcommand{\inputDistribution}{P}
\newcommand{\mutualInformation}[2]{\mathbf{I}_{#1,#2}}
\newcommand{\outputAlphabet}{\mathcal{Y}}
\newcommand{\outputRV}{Y}
\newcommand{\noiseRV}{N}
\newcommand{\noisestd}{\sigma}
\newcommand{\disturbancestd}{\sigma_B}

\newcommand{\normaldistribution}[2]{\mathcal{N}(#1,#2)}

\newcommand{\channel}{\mathcal{W}}
\newcommand{\channelAWGN}{\mathcal{W}_\mathrm{AWGN}}
\newcommand{\channelFading}{\mathcal{W}_\mathrm{fading}}
\newcommand{\txNum}{K}
\newcommand{\txIndex}{k}
\newcommand{\objectiveFunction}{f}
\newcommand{\objectiveFunctionEstimate}{\hat{f}}
\newcommand{\objectiveFunctionDomain}{\mathcal{S}}

\newcommand{\stateSpaceElement}{s}
\newcommand{\stateSpace}{\mathcal{S}}
\newcommand{\messageRealization}{m}
\newcommand{\messageRV}{\mathcal{M}}
\newcommand{\blocklength}{n}
\newcommand{\blIndex}{i}
\newcommand{\preproc}{E}
\newcommand{\postproc}{D}
\newcommand{\reals}{\mathbb{R}}
\newcommand{\complexNumbers}{\mathbb{C}}
\newcommand{\generalReal}{x}
\newcommand{\generalRealTwo}{y}
\newcommand{\powerconstraint}{\mathcal{P}}
\newcommand{\tail}{\varepsilon}
\newcommand{\errorprob}{\delta}
\newcommand{\Probability}{\mathbb{P}}
\newcommand{\naturals}{\mathbb{N}}
\newcommand{\generalNatural}{n}
\newcommand{\generalNaturalTwo}{k}

\newcommand{\generalNaturalFour}{p}
\newcommand{\absolutevalue}[1]{\left\lvert #1 \right\rvert}
\newcommand{\positivePart}[1]{\left[ #1 \right]^+}
\newcommand{\codebook}{\preproc}
\newcommand{\decoder}{\postproc}
\newcommand{\codebookSize}{M}
\newcommand{\cardinality}[1]{\left\lvert #1 \right\rvert}
\newcommand{\fading}{h}

\newcommand{\euclidNorm}[1]{\lVert #1 \rVert}
\newcommand{\maximumNorm}[1]{\lVert #1 \rVert_\infty}
\newcommand{\generalRV}{A}
\newcommand{\generalRVTwo}{B}

\newcommand{\Expectation}{{\mathbb{E}}}
\newcommand{\Variance}{\mathrm{Var}}
\newcommand{\txSubset}{J}
\newcommand{\generalSubset}{J}
\newcommand{\analogTxSubset}{J_a}
\newcommand{\digitalTxSubset}{J_d}
\newcommand{\analogTxSet}{\mathcal{K}_a}
\newcommand{\digitalTxSet}{\mathcal{K}_d}
\newcommand{\complement}[1]{#1^c}
\newcommand{\analogTxSubsetComplement}{J_a^c}
\newcommand{\digitalTxSubsetComplement}{J_d^c}
\newcommand{\transpose}[1]{#1^T}
\newcommand{\planeMap}{U}
\newcommand{\planeMapElement}{u}
\newcommand{\planeMapInv}{\transpose{\planeMap}}
\newcommand{\identityMapping}[1]{\mathrm{id}_{#1}}
\newcommand{\basisOne}{v}
\newcommand{\basisTwo}{w}
\newcommand{\innerproduct}[2]{\left\langle{#1},{#2}\right\rangle}
\newcommand{\onevector}{\mathbbm{1}}
\newcommand{\plane}{E}
\newcommand{\generalIndex}{i}
\newcommand{\imaginaryUnit}{\mathbf{j}}
\newcommand{\numAnalogTransmissions}{L}
\newcommand{\indexAnalogTransmissions}{\ell}

\newcommand{\preprocHybrid}{\mathcal{E}}
\newcommand{\postprocHybrid}{\mathcal{D}}

\newcommand{\analograte}{\beta}

\newcommand{\analogTxNum}{K_a}
\newcommand{\digitalTxNum}{K_d}

\newcommand{\unconstrainedGaussianCapacity}{C}
\newcommand{\unconstrainedMacCapacity}[1]{C_{#1}}
\newcommand{\constrainedMacCapacity}[1]{\bar{C}_{#1}}

\newcommand{\constrainedMacCapacityAchievable}[1]{\hat{C}_{#1}}
\newcommand{\constrainedMacCapacityConverse}[1]{\check{C}_{#1}}
\newcommand{\decodingError}{\varepsilon_\mathrm{dec}}
\newcommand{\closure}{\mathrm{cl}}
\newcommand{\amplitudeConstraint}{\mathcal{A}}
\newcommand{\analogAmplitudeConstraint}{\amplitudeConstraint_a}
\newcommand{\optimalInputDistribution}{\hat{P}}
\newcommand{\timeSharingRV}{Q}
\newcommand{\mutualInformationRV}[2]{\mathbf{I}\left(#1;#2\right)}
\newcommand{\mutualInformationRVCond}[3]{\mathbf{I}\left(#1;#2|#3\right)}
\newcommand{\entropyRV}[1]{H(#1)}
\newcommand{\entropyRVCond}[2]{H(#1|#2)}
\newcommand{\diffEntropyRVCond}[2]{h(#1|#2)}
\newcommand{\diffEntropyRV}[1]{h(#1)}
\newcommand{\analogTransmissionsIndexSet}{\naturalsInitialSection{\numAnalogTransmissions}}
\newcommand{\compoundChannel}[1]{\mathcal{W}_{#1}}

\newcommand{\codebookRate}{R}
\newcommand{\converseAnalogErrorVariance}{V}
\newcommand{\nomPre}[2]{\ifthenelse{\isempty{#2}}{f_{#1}}{f_{#1}^{(#2)}}}
\newcommand{\nomPost}[1]{\ifthenelse{\isempty{#1}}{F}{F^{(#1)}}}
\newcommand{\nomInc}[1]{\ifthenelse{\isempty{#1}}{\Phi}{\Phi^{(#1)}}}
\newcommand{\nomMin}[2]{\ifthenelse{\isempty{#2}}{\phi_{\min,#1}}{\phi_{\min,#1}^{(#2)}}}
\newcommand{\nomMax}[2]{\ifthenelse{\isempty{#2}}{\phi_{\max,#1}}{\phi_{\max,#1}^{(#2)}}}

\newcommand{\nomMaxSpread}[1]{\Delta(#1)}
\newcommand{\effectiveNoise}{N_\mathrm{eff}}
\newcommand{\Fmon}{\mathfrak{F}_\mathrm{mon}}
\newcommand{\nomIntermediateDomain}{\mathcal{D}}
\newcommand{\hadsCapacity}[2]{\mathcal{H}_{#1,#2}}
\newcommand{\hadsConverse}[2]{\check{\mathcal{H}}_{#1,#2}}
\newcommand{\hadsAchievable}[2]{\hat{\mathcal{H}}_{#1,#2}}
\newcommand{\eulersNumber}{e}
\newcommand{\naturalsInitialSection}[1]{[#1]}

\newcommand{\flRounds}{t}
\newcommand{\landauO}{\mathcal{O}}
\newcommand{\mean}{\mu}

\definecolor{revcolor}{rgb}{0,0,0.5}
\newtheorem{theorem}{Theorem}

\newtheorem{lemma}{Lemma}
\newtheorem{cor}{Corollary}

\newtheorem{remark}{Remark}
\newtheorem{definition}{Definition}

\newcommand{\drawaxes}[2]{
  \coordinate (origin) at (0,0);
  \coordinate (xend)   at (#1,0);
  \coordinate (yend)   at (0,#2);

  \draw[->] (origin) -- (xend) node[at end, below right] {$\codebookRate_1$};
  \draw[->] (origin) -- (yend) node[at end, above left] {$\codebookRate_2$};
}

\newcommand{\drawxtick}[2]{
  \draw (#1,0) -- (#1,-.03) node[at end,below] {#2};
}

\newcommand{\drawytick}[2]{
  \draw (0,#1) -- (-.03,#1) node[at end, left] {#2};
}

\newacronym{awgn}{AWGN}{additive white Gaussian noise}
\newacronym{ota}{OTA}{Over-the-Air}
\newacronym{otac}{OTA-C}{Over-the-Air Computation}
\newacronym{mac}{MAC}{multiple-access channel}
\newacronym{mse}{MSE}{mean squared error}
\newacronym{socc}{SOCC}{Simultaneous Over-the-Air Computation and Communication}
\newacronym{ldpc}{LDPC}{low density parity check}
\newacronym{ml}{ML}{machine learning}
\newacronym{mimo}{MIMO}{multiple-input multiple-output}
\newacronym{sinr}{SINR}{signal-to-interference-and-noise ratio}
\newacronym{ris}{RIS}{reflective intelligent surface}
\newacronym{5g}{5G}{fifth-generation technology standard for cellular networks}
\newacronym{ber}{BER}{bit error rate}

\makenoidxglossaries
\glsdisablehyper
\newglossaryentry{outputRV}{
  name={\ensuremath{\outputRV}},
  description={Channel output at the receiver}
}
\newglossaryentry{analogTxNum}{
  name={\ensuremath{\analogTxNum}},
  description={Number of analog transmitters}
}
\newglossaryentry{digitalTxNum}{
  name={\ensuremath{\digitalTxNum}},
  description={Number of digital transmitters}
}
\newglossaryentry{totalTxNum}{
  name={\ensuremath{\txNum = \analogTxNum + \digitalTxNum}},
  description={Overall number of transmitters}
}
\newglossaryentry{channelInputs}{
  name={\ensuremath{\inputRV_1, \dots, \inputRV_\txNum}},
  description={Channel inputs}
}
\newglossaryentry{normalDistribution}{
  name={\ensuremath{\normaldistribution{\mean}{\noisestd^2}}},
  description={Normal distribution with mean $\mean$ and variance $\noisestd^2$}
}
\newglossaryentry{channelAWGN}{
  name={\ensuremath{\channelAWGN}},
  description={\gls{awgn} Channel}
}
\newglossaryentry{blocklength}{
  name={\ensuremath{\blocklength}},
  description={Number of channel uses (block length)}
}
\newglossaryentry{analogPreProcessors}{
  name={\ensuremath{\preprocHybrid_1, \dots, \preprocHybrid_{\analogTxNum}}},
  description={\acrshort{socc} pre-processors at the analog transmitters}
}
\newglossaryentry{digitalPreProcessors}{
  name={\ensuremath{\preprocHybrid_{\analogTxNum+1}, \dots, \preprocHybrid_{\txNum}}},
  description={\acrshort{socc} pre-processors at the digital transmitters}
}
\newglossaryentry{postProcessor}{
  name={\ensuremath{\postprocHybrid}},
  description={\acrshort{socc} post-processor at the receiver}
}
\newglossaryentry{analogValues}{
  name={\ensuremath{\stateSpaceElement_{\txIndex,1} \in \stateSpace_{\txIndex,1}, \dots, \stateSpaceElement_{\txIndex,\numAnalogTransmissions} \in \stateSpace_{\txIndex,\numAnalogTransmissions}}},
  description={Analog values held at transmitter $\txIndex$}
}
\newglossaryentry{naturalsInitialSection}{
  name={\ensuremath{\naturalsInitialSection{\generalNatural} := \{1, \dots, \generalNatural\}}},
  text={\ensuremath{\naturalsInitialSection{\analogTxNum} := \{1, \dots, \analogTxNum\}}},
  description={Set of first $\generalNatural$ natural numbers}
}
\newglossaryentry{analogAmplitudeConstraint}{
  name={\ensuremath{\analogAmplitudeConstraint}},
  description={Amplitude constraint at the analog transmitters}
}
\newglossaryentry{maximumNorm}{
  name={\ensuremath{\maximumNorm{\cdot}}},
  description={Maximum norm}
}
\newglossaryentry{euclidNorm}{
  name={\ensuremath{\euclidNorm{\cdot}}},
  description={Standard Euclidean norm}
}
\newglossaryentry{codebookSize}{
  name={\ensuremath{\codebookSize_\txIndex}},
  description={Size of codebook at the $\txIndex$th digital transmitter}
}
\newglossaryentry{amplitudeConstraint}{
  name={\ensuremath{\amplitudeConstraint_\txIndex}},
  text={\ensuremath{\amplitudeConstraint_{\analogTxNum + \txIndex}}},
  description={Amplitude constraint at transmitter $\txIndex$}
}
\newglossaryentry{powerconstraint}{
  name={\ensuremath{\powerconstraint_\txIndex}},
  text={\ensuremath{\powerconstraint_{\analogTxNum + \txIndex}}},
  description={Average power constraint at transmitter $\txIndex$}
}
\newglossaryentry{messageRV}{
  name={\ensuremath{\messageRV_\txIndex}},
  description={Message to be transmitted by the $\txIndex$-th digital transmitter (($\analogTxNum+\txIndex$)-th overall transmitter)}
}
\newglossaryentry{functionEstimates}{
  name={\ensuremath{\objectiveFunctionEstimate_1, \dots, \objectiveFunctionEstimate_\numAnalogTransmissions}},
  description={Approximate function values obtained by the receiver through \acrshort{otac}}
}
\newglossaryentry{messageEstimates}{
  name={\ensuremath{\hat{\messageRV}_1, \dots, \hat{\messageRV}_{\digitalTxNum}}},
  description={Digital messages recovered at the receiver}
}
\newglossaryentry{objectiveFunction}{
  name={\ensuremath{\objectiveFunction^{(\indexAnalogTransmissions)}}},
  text={},
  description={$\indexAnalogTransmissions$-th objective function to be computed over the air}
}
\newglossaryentry{converseAnalogErrorVariance}{
  name={\ensuremath{\converseAnalogErrorVariance}},
  description={\acrshort{mse} of the analog function estimates at the receiver}
}
\newglossaryentry{codebookRate}{
  name={\ensuremath{\codebookRate_\txIndex}},
  text={},
  description={Coding rate at the $\txIndex$-th digital transmitter (($\analogTxNum+\txIndex$)-th overall transmitter)}
}
\newglossaryentry{analograte}{
  name={\ensuremath{\analograte}},
  description={Rate of analog function computations}
}
\newglossaryentry{noiseStd}{
  name={\ensuremath{\noisestd^2}},
  description={Power of the channel noise}
}
\newglossaryentry{preproc}{
  name={\ensuremath{\preproc_1, \dots, \preproc_\txNum}},
  description={Digital encoders for a \acrshort{mac} with $\txNum$ inputs}
}
\newglossaryentry{postproc}{
  name={\ensuremath{\postproc}},
  description={Decoder which recovers digital messages}
}
\newglossaryentry{decodingError}{
  name={\ensuremath{\decodingError}},
  description={Digital decoding error}
}
\newglossaryentry{unconstrainedGaussianCapacity}{
  name={\ensuremath{\unconstrainedGaussianCapacity}},
  description={Gaussian capacity function for real channels}
}
\newglossaryentry{unconstrainedMacCapacity}{
  name={\ensuremath{\unconstrainedMacCapacity{\channel}(\powerconstraint_1, \dots, \powerconstraint_\txNum)}},
  description={Capacity region of the \acrshort{mac} $\channel$ under average power constraints $\powerconstraint_1, \dots, \powerconstraint_\txNum$}
}
\newglossaryentry{inputDistributions}{
  name={\ensuremath{\inputDistribution_1, \dots, \inputDistribution_\txNum}},
  description={Channel input distributions}
}
\newglossaryentry{constrainedMacCapacity}{
  name={\ensuremath{
    \constrainedMacCapacity{\channel}(
      \powerconstraint_1,
      \dots,
      \powerconstraint_\txNum,
      \amplitudeConstraint_1,
      \dots,
      \amplitudeConstraint_\txNum
    )
  }},
  text={},
  description={Capacity region of the \acrshort{mac} $\channel$ under average power constraints $\powerconstraint_1, \dots, \powerconstraint_\txNum$ and amplitude constraints $\amplitudeConstraint_1, \dots, \amplitudeConstraint_\txNum$}
}
\newglossaryentry{inputAlphabets}{
  name={\ensuremath{\inputAlphabet_1, \dots, \inputAlphabet_\txNum}},
  description={Channel input alphabets}
}\newglossaryentry{outputAlphabet}{
  name={\ensuremath{\outputAlphabet}},
  description={Channel output alphabet}
}
\newglossaryentry{mutualInformation}{
  name={\ensuremath{\mutualInformation{\inputDistribution_1, \dots, \inputDistribution_\txNum}{\channel}(\cdot;\cdot|\cdot)}},
  description={Mutual information between the inputs and output of channel $\channel$ under the input distributions $\inputDistribution_1, \dots, \inputDistribution_\txNum$ (subscript sometimes omitted when clear from context))}
}
\newglossaryentry{noiseRV}{
  name={\ensuremath{\noiseRV}},
  description={Channel noise}
}
\newglossaryentry{positivePart}{
  name={\ensuremath{\positivePart{\generalReal} := \max(0,x)}},
  description={Positive part of $\generalReal$}
}
\newglossaryentry{closure}{
  name={\ensuremath{\closure}},
  description={Topological closure}
}

\title{Simultaneous Computation and Communication over MAC}
\author{%
  \IEEEauthorblockN{Matthias Frey\IEEEauthorrefmark{1},
                    Igor Bjelaković\IEEEauthorrefmark{2}\IEEEauthorrefmark{3},
                    Michael C. Gastpar\IEEEauthorrefmark{4},
                    and Jingge Zhu\IEEEauthorrefmark{1}}
  \\
  \IEEEauthorblockA{\IEEEauthorrefmark{1}%
                    Department of Electrical and Electronic Engineering, The University of Melbourne, Australia}
  \\
  \IEEEauthorblockA{\IEEEauthorrefmark{2}%
                    Fraunhofer Heinrich Hertz Institute, Berlin, Germany}
  \IEEEauthorblockA{\IEEEauthorrefmark{3}%
                    Technische Universität Berlin, Germany}
  \\
  \IEEEauthorblockA{\IEEEauthorrefmark{4}%
                    École Polytechnique Fédérale de Lausanne, Switzerland}
}

\allowdisplaybreaks

\begin{document}
\maketitle

\blankfootnote{Part of this work has been presented at the 2024 IEEE International Symposium of Information Theory (ISIT), Athens, Greece.

The work of J. Zhu was supported in part by the Australian Research Council under project DE210101497 and under project DP230101493. I. Bjelakovi\'c acknowledges support by the German Research Foundation (DFG) under grant STA 864/15-1 and financial support by the Federal Ministry of Education and Research of Germany in the program of ``Souverän. Digital. Vernetzt.''. Joint project 6G-RIC, project identification numbers: 16KISK020K, 16KISK030. The work of M. C. Gastpar was supported in part by the Swiss National Science Foundation under Grant 200364. This research was also supported by The University of Melbourne’s Research Computing Services and the Petascale Campus Initiative.}
\begin{abstract}
We study communication over a Gaussian \gls{mac} with two types of transmitters: Digital transmitters hold a message from a discrete set that needs to be communicated to the receiver with vanishing error probability. Analog transmitters hold sequences of analog values. Some functions of these distributed values (but not the values themselves) need to be conveyed to the receiver, subject to a fidelity criterion such as \gls{mse} or a certain maximum error with given confidence. For the case in which the computed function for the analog transmitters is a sum of values in $[-1,1]$, we derive inner and outer bounds for the tradeoff of digital and analog rates of communication under peak and average power constraints for digital transmitters and a peak power constraint for analog transmitters. We then extend the achievability result to a class of functions that includes all linear and some non-linear functions. This extended scheme works over fading channels as long as full channel state information is available at the transmitter. The practicality of our proposed communication scheme is shown in channel simulations that use a version of the scheme based on \gls{ldpc} coding. We evaluate the system performance for different block lengths and Gaussian as well as non-Gaussian noise distributions.
\end{abstract}

\section{Introduction}
\label{sec:introduction}
\gls{otac} is a class of joint source-channel communication schemes that can be used whenever data from a large number of transmitters has to be collected in such a way that the receiver requires some function of the transmitted data, but not necessarily all of the individual data points. \gls{otac} schemes can be roughly divided into two classes: Digital (or coded) \gls{otac} and analog (or uncoded) \gls{otac}. While digital schemes are vastly more efficient in terms of energy and spectrum usage (with some notable exceptions~\cite{gastpar2008uncoded}), analog schemes have much simpler transceiver designs, are more versatile with respect to what functions they can compute, and are often more robust to practical inhibitions such as small synchronization errors. This makes them easier to use for emerging applications such as distributed sensing, distributed control, and distributed learning in large wireless networks. Because of this, there has been a growing research interest in analog \gls{otac} schemes over the last few years, despite their drawbacks in comparison to digital schemes when it comes to efficiency. One feature that many analog \gls{otac} schemes share and which ostensibly is the main reason for their low efficiency is that they combat the influence of channel noise and fading by repeating the same transmission multiple times. In this work, we take a slightly different perspective: We show that this type of repeated transmission of the same signal can be combined with digitally coded communication in a way that yields an efficient overall communication system. That is, the spectral efficiency of the resulting system is much better than that of a system which separates digitally coded communication and analog \gls{otac} in time domain. Moreover, we show that it is still possible to satisfy reasonable amplitude constraints. We propose a novel \emph{hybrid} communication scheme in which digital communications are encoded in such a way that they do not disturb and are not disturbed by analog \gls{otac} that is executed concurrently through the same channel.

\subsection{Prior Work}
\label{sec:literature}

\paragraph{\GLS{otac}}
Digital \gls{otac} was first proposed and analyzed in~\cite{nazer2007computation}. First ideas for analog \gls{otac} of sums can be traced back to~\cite{gastpar2003source}. A larger class of functions for \gls{otac} was considered in~\cite{goldenbaum2009function}. In~\cite{goldenbaum2013harnessing,goldenbaum2014nomographic}, a systematic approach for the analog \gls{otac} of a large class of functions was proposed. This is the first work to make the connection between \gls{otac} and nomographic representations of the functions to be computed. This means that a large class of functions can essentially be reduced to a summation for the sake of \gls{otac} schemes. Applications of \gls{otac} include physical layer network coding~\cite{nazer2007computation,nazer2011compute}, distributed control~\cite{cai2018modulation}, distributed sensing~\cite{goldenbaum2009function}, and distributed optimization for federated learning~\cite{amiri2020machine}. There has been a tremendous development of research in \gls{otac}, in particular over the last few years. For a comprehensive overview of the field, we direct the reader to the surveys~\cite{wang2022over,sahin2023survey}.

\paragraph{Co-existence of Digital Communication and Analog \GLS{otac}\label{par:literature-coexistence} in Wireless Networks}
In~\cite{qi2020integration}, the authors consider a cellular network with \gls{mimo} transceivers in which three different tasks need to be reconciled in the same wireless frequency band, but with a partial separation in time and space (with the use of beamforming): Power delivery from the base station to the network nodes, digital communications, and analog \gls{otac}. The authors solve an optimization problem to derive beamformers and power management which optimize the analog \gls{mse} while adhering to transmitter power constraints and guaranteeing a certain minimum \glspl{sinr} for digital communication. A similar scenario, but for the simultaneous tasks of digital communication, analog \gls{otac}, and sensing was considered in~\cite{qi2021integrated}, and only for digital communication and analog \gls{otac} in~\cite{du2021interference}. In~\cite{qi2022integrating}, the authors establish a tradeoff between the performances of the three tasks of digital communication, analog \gls{otac}, and sensing. \Glspl{ris} are used for a partial spatial separation of digital communication and analog \gls{otac} in~\cite{ni2022integrating}, and the authors use power control to achieve a tradeoff between the performance of the digital and analog parts of the system. A central assumption of this work is that the signal from the analog transmitters is weaker at the receiver than the signal from the digital transmitters, which enables the use of successive interference cancellation at the decoder. A common feature of all of these works is that they at least partially separate the digital and analog signals in the spatial domain, either with \gls{mimo} beamforming or with \glspl{ris}. Moreover, there is always a residual impact on \gls{sinr} or \gls{mse} of the digital and analog parts of the system from the respective other part of the system.

Our present study is complementary to these existing works in that we approach the problem of co-existence of digital communication systems and analog \gls{otac} from an information-theoretic angle. We separate the digital and analog signals purely by coding instead of exploiting spatial modes. Hence, we establish a different type of tradeoff between the digital and analog communication rates. This works regardless of how strong the signal of the analog transmitters is compared to the digital transmitters. We consider a simpler channel model than these prior works, but our approach requires neither the presence of \gls{ris} in the channel nor of multiple antennas at transmitting or receiving nodes.

\paragraph{Hybrid Digital/Analog Coding Schemes}
This is a class of joint source-channel coding schemes that exhibit characteristics of both digital coding and analog uncoded transmission. To the best of our knowledge, they were first proposed under this name in~\cite{mittal2002hybrid} as a means to make the transmission of a Gaussian source through a point-to-point \gls{awgn} channel more robust against deviations from the modeling assumptions. First ideas for this type of codes can be traced back to the systematic source-channel codes introduced in~\cite{shamai1998systematic} for a scenario with side information. Several later works have applied hybrid digital/analog coding schemes to different scenarios, such as relay channels~\cite{yao2009hybrid} and \gls{mac}~\cite{lapidoth2010sending}. A unified framework for hybrid digital/analog coding schemes that recovers many of these older results was proposed in~\cite{minero2015unified}. For a more extensive list of works on this topic and a more detailed discussion of these works, we refer the reader to the literature section in~\cite{minero2015unified}.

These papers employ coding schemes that have features of both digital coding and analog uncoded transmission in a joint source-channel coding problem. Our present work differs from this in that we employ analog or digital coding schemes at different transmitters based on the type of information that the particular transmitter conveys, and that we perform~\gls{otac} based on analog values. Therefore, both the communication scenario and the methodology are fundamentally different, although we do draw on the basic idea of combining digital coding techniques and analog transmissions concurrently in the same communication system.

\subsection{Contribution and Outline}
\label{sec:contribution}
The contribution of this paper can be summarized as follows:
\begin{itemize}
  \item We propose a system model for \gls{socc} over a \gls{mac}. In this system model, there are analog transmitters that perform \gls{otac} and digital transmitters that each communicate a message to the receiver. Both types of transmitters share the same channel.
  \item For the case of a channel with \gls{awgn} and no fading, we derive inner and outer bounds for the tradeoff between the rates of digital communication and the number of function values that are computed with \gls{otac}. For this, we focus on functions of sums of values in the interval $[-1,1]$. Our bounds hold under peak and average power constraints at the digital transmitters and a peak power constraint at the analog transmitters.
  \item We show that the communication schemes used for the achievability proofs can be modified in a straightforward way to compute any function in the class $\Fmon$. This is the same class of functions that the scheme proposed in~\cite{frey2021over-tsp} is able to compute. $\Fmon$ contains all linear functions such as weighted averages which are of particular interest in distributed learning applications. It also includes some nonlinear functions such as $p$-norms for $p \geq 1$. This is due to the observation~\cite{goldenbaum2013harnessing,goldenbaum2014nomographic} that summation can be seen as the prototypical function for \gls{otac} since via suitable nomographic representations, a large class of function computations can be reduced to summation. The modified scheme we propose works over fading channels, as long as full channel state information is available at the transmitter.
  \item The hybrid communication schemes we propose for the achievability part of our main result are derived from codes for standard \gls{mac} communication. For this, we use computationally inexpensive additional processing steps at the transmitters and the receiver. In the theoretical part of this work, we apply them to random \gls{mac} codes for simplicity, but they are compatible with arbitrary \gls{mac} codes. This means that the communication schemes we propose are in principle suitable for real-world communication systems when used for instance with \gls{ldpc} or polar codes. To illustrate this point, we have conducted channel simulations with our proposed \gls{socc} scheme, using digital \gls{ldpc} coding and modulation from the \acrfull{5g}. Our simulations include both short and moderate block lengths, and we have tested both Gaussian and non-Gaussian noise, showing a moderate degradation of performance as the channel noise deviates from a Gaussian distribution.
\end{itemize}
In Section~\ref{sec:system-model}, we introduce the novel system model for analog \gls{otac} and digitally coded communication through the same channel. In Section~\ref{sec:main-result}, we define the rate tradeoff region that we study in this paper and summarize our result on the bounds of this region. We also plot sum rate bounds for numerical examples. These plots show that our bounds are not fully tight, but they do establish that there is a meaningful tradeoff between the rates of digital communications and the number of \glspl{otac} that can be performed for the system model we propose. Section~\ref{sec:achievability} is devoted to the achievability part of our main result. In particular, in Section~\ref{sec:construction} we give details on how digital codes that are compatible with concurrent \gls{otac} can be constructed given \emph{any} base code for communication through a \gls{mac} (including practically relevant schemes such as \gls{ldpc} and polar codes). The converse part of the main result is proven in Section~\ref{sec:converse}. This proof is an extension of the standard technique used to derive converse results for the Gaussian \gls{mac}. In Section~\ref{sec:fmon}, we show how the achievability part of the main result can be extended to compute any function in the class $\Fmon$ introduced in~\cite{frey2021over-tsp}. Finally, we summarize the results of our numerical simulations in Section~\ref{sec:simulations}. In order to aid the reader in keeping track of the extensive notation used throughout the paper, we provide an index of all symbols that are used in multiple sections (in order of first usage) in Appendix~\ref{appendix:symbols}.

\begin{figure}
    \centering
    \begin{tikzpicture}
        \coordinate (in1analog) at (0,3);
        \coordinate (inKanalog) at (0,0);
        \coordinate (arrowturn) at (6.5,0);
        \coordinate (in1digital) at (0,-2);
        \coordinate (inKdigital) at (0,-5);

        \node[rectangle,draw,minimum width=1.5cm,minimum height=1.5cm] (enc1) at (4,3) {$\preprocHybrid_1$};
        \node (dots) at (4,1.5) {\Shortstack{. . . . . .}};
        \node[rectangle,draw,minimum width=1.5cm,minimum height=1.5cm] (encKa) at (4,0) {$\preprocHybrid_{\analogTxNum}$};

        \node[rectangle,draw,minimum width=1.5cm,minimum height=1.5cm] (encKa1) at (4,-2) {$\preprocHybrid_{\analogTxNum+1}$};
        \node (dots) at (4,-3.5) {\Shortstack{. . . . . .}};
        \node[rectangle,draw,minimum width=1.5cm,minimum height=1.5cm] (encKaKd) at (4,-5) {$\preprocHybrid_{\txNum}$};

        \node[circle,draw] (plus) at (9,-1) {$+$};
        \node (noise) at (9,2) {$\noiseRV^\blocklength \sim \text{i.i.d. } \normaldistribution{0}{\noisestd^2}$};

        \node[rectangle,draw,minimum width=1.5cm,minimum height=1.5cm] (dec) at (12,-1) {$\postprocHybrid$};

        \coordinate (outdigital) at (16,-1.25);
        \coordinate (outanalog) at (16,-.75);

        \draw[->] (in1analog) -- (enc1) node[midway,above] {$\stateSpaceElement_{1,1}, \dots, \stateSpaceElement_{1,\numAnalogTransmissions}$};
        \draw[->] (inKanalog) -- (encKa) node[midway,above] {$\stateSpaceElement_{\analogTxNum,1}, \dots, \stateSpaceElement_{\analogTxNum,\numAnalogTransmissions}$};
        \draw[->] (in1digital) -- (encKa1) node[midway,above] {$\messageRV_{1}$};
        \draw[->] (inKdigital) -- (encKaKd) node[midway,above] {$\messageRV_{\digitalTxNum}$};

        \draw[->] (enc1) -- (enc1-|arrowturn) node[midway,above] {$\inputRV_1^{\blocklength}$} -- (plus);
        \draw[->] (encKa) -- (encKa-|arrowturn) node[midway,above] {$\inputRV_{\analogTxNum}^{\blocklength}$} -- (plus);
        \draw[->] (encKa1) -- (encKa1-|arrowturn) node[midway,above] {$\inputRV_{\analogTxNum+1}^{\blocklength}$} -- (plus);
        \draw[->] (encKaKd) -- (encKaKd-|arrowturn) node[midway,above] {$\inputRV_{\txNum}^{\blocklength}$} -- (plus);
        \draw[->] (noise) -- (plus);

        \draw[->] (plus) -- (dec) node[midway,above] {$\outputRV^{\blocklength}$};

        \draw[->] (dec.east|-outdigital) -- (outdigital) node[midway,below] {$\hat{\messageRV}_1, \dots, \hat{\messageRV}_{\digitalTxNum}$};
        \draw[->] (dec.east|-outanalog) -- (outanalog) node[midway,above] {$\hat{\objectiveFunction}_1, \dots, \hat{\objectiveFunction}_{\numAnalogTransmissions}$};
    \end{tikzpicture}
    \caption{System model for \acrfull{socc}.}
    \label{fig:digital-mac-with-ota-c}
\end{figure}

\section{System Model}
\label{sec:system-model}

\subsection{Channel and Communication Model}
In this work, we consider a combination of analog~\gls{otac} and digital communication over a shared \gls{mac} with \gls{awgn}. The representation of this channel in terms of random variables is
\begin{equation}
\label{eq:channel-awgn}
\outputRV = \inputRV_1 + \dots + \inputRV_\txNum + \noiseRV,
\end{equation}
where $\gls{channelInputs}$ are the channel inputs, $\gls{outputRV}$ is the channel output, and $\gls{noiseRV} \sim \normaldistribution{0}{\gls{noiseStd}}$ where $\gls{normalDistribution}$ denotes the normal distribution with mean $\mean$ and variance $\noisestd^2$. For simplicity, we focus on the case where each transmitter conveys either a digital message or analog values for \gls{otac}: There are $\gls{analogTxNum}$ analog transmitters, indexed $1, \dots, \analogTxNum$, and $\gls{digitalTxNum}$ digital transmitters, indexed $\analogTxNum + 1, \dots, \analogTxNum + \digitalTxNum$, where the total number of transmitters is $\gls{totalTxNum}$. Transmitters and receiver are connected via the channel $\gls{channelAWGN}$ defined in (\ref{eq:channel-awgn}), which they can use a total of $\gls{blocklength}$ times. A \gls{socc} scheme, depicted in Fig.~\ref{fig:digital-mac-with-ota-c}, consists of analog pre-processors \gls{analogPreProcessors}, digital pre-processors \gls{digitalPreProcessors}, and a post-processor \gls{postProcessor}.

Each analog transmitter $\txIndex\in\gls{naturalsInitialSection}$ holds analog values $\gls{analogValues}$ which are passed through $\preprocHybrid_\txIndex$ to  create a channel input sequence $\inputRV_\txIndex^\blocklength$. The pre-processor output is subject to a peak amplitude constraint of $\gls{analogAmplitudeConstraint} > 0$. This means that for every $\txIndex \in \naturalsInitialSection{\analogTxNum}$, the pre-processor $\preprocHybrid_\txIndex$ has to satisfy
\begin{equation}
\label{eq:analog-amplitude-constraint}
\forall \stateSpaceElement_{\txIndex,1} \in \stateSpace_{\txIndex,1}, \dots, \stateSpaceElement_{\txIndex,\numAnalogTransmissions} \in \stateSpace_{\txIndex,\numAnalogTransmissions}:~
\maximumNorm{
  \preprocHybrid_\txIndex(
    \stateSpaceElement_{\txIndex,1}, \dots, \stateSpaceElement_{\txIndex,\numAnalogTransmissions}
  )
}
\leq
\analogAmplitudeConstraint,
\end{equation}
where $\gls{maximumNorm}$ denotes the maximum norm on real Euclidean spaces. This amplitude constraint automatically implies an average power constraint
\begin{equation}
\label{eq:analog-power-constraint}
\forall \stateSpaceElement_{\txIndex,1} \in \stateSpace_{\txIndex,1}, \dots, \stateSpaceElement_{\txIndex,\numAnalogTransmissions} \in \stateSpace_{\txIndex,\numAnalogTransmissions}:~
\frac{1}{\blocklength}
\euclidNorm{
  \preprocHybrid_\txIndex(
    \stateSpaceElement_{\txIndex,1}, \dots, \stateSpaceElement_{\txIndex,\numAnalogTransmissions}
  )
}^2
\leq
\analogAmplitudeConstraint^2,
\end{equation}
where \gls{euclidNorm} denotes the standard Euclidean norm.

Each digital transmitter $\analogTxNum + \txIndex$, $\txIndex \in \naturalsInitialSection{\digitalTxNum}$, holds a message $\gls{messageRV} \in \naturalsInitialSection{\gls{codebookSize}}$, where $\codebookSize_\txIndex$ is the number of different digital messages transmitter $\analogTxNum + \txIndex$ can convey. The message is passed through $\preprocHybrid_{\analogTxNum + \txIndex}$ which outputs a channel input sequence $\inputRV_{\analogTxNum + \txIndex}^\blocklength$. For the digital encoder outputs, we impose a maximum amplitude constraint
\begin{equation}
\label{eq:digital-amplitude-constraint}
\forall \messageRealization \in \naturalsInitialSection{\codebookSize_\txIndex}:~
\maximumNorm{
  \preprocHybrid_{\analogTxNum+\txIndex}(
    \messageRealization
  )
}
\leq
\gls{amplitudeConstraint}
\end{equation}
and a separate average power constraint
\begin{equation}
\label{eq:digital-power-constraint}
\forall \messageRealization \in \naturalsInitialSection{\codebookSize_\txIndex}:~
\frac{1}{\blocklength}
\euclidNorm{
  \preprocHybrid_{\analogTxNum+\txIndex}(
    \messageRealization
  )
}^2
\leq
\gls{powerconstraint}.
\end{equation}

The channel input sequences $\inputRV_1^\blocklength, \dots, \inputRV_{\txNum}^\blocklength$ are then passed through the $\blocklength$-fold product channel $\channelAWGN^\blocklength$. The receiver passes the channel output sequence $\outputRV^\blocklength$ through $\postprocHybrid$ which outputs analog \gls{ota} computed function estimates \gls{functionEstimates}, along with estimates \gls{messageEstimates} of the digital messages.

We assume that there is a tuple
\[
\gls{objectiveFunction}
\objectiveFunction^{(1)}:
  \stateSpace_{1,1} \times \dots \times \stateSpace_{\analogTxNum,1}
  \rightarrow
  \reals,
\dots,
\objectiveFunction^{(\numAnalogTransmissions)}:
  \stateSpace_{1,\numAnalogTransmissions} \times \dots \times \stateSpace_{\analogTxNum,\numAnalogTransmissions}
  \rightarrow
  \reals
\]
of functions (known to all transmitters and the receiver) that are to be \gls{ota} computed. The function estimates $\objectiveFunctionEstimate_1, \dots, \objectiveFunctionEstimate_\numAnalogTransmissions$ should closely approximate the functions $\objectiveFunction^{(1)}, \dots, \objectiveFunction^{(\indexAnalogTransmissions)}$.

\subsection{Error Criteria for \gls{otac}}
We next introduce criteria for the approximation quality of the \gls{otac} results at the receiver that we consider in this paper. All of these criteria derive from widely known and used mathematical approximation criteria.
\begin{definition}
\label{def:analog-approximation-criteria}
Let $\objectiveFunction: \stateSpace_1 \times \dots \times \stateSpace_\txNum \rightarrow \reals$ be a function that is to be \gls{ota} computed between $\txNum$ transmitters and a receiver, let, for all $\txIndex \in \naturalsInitialSection{\txNum}$, $\stateSpaceElement_\txIndex \in \stateSpace_\txIndex$ be the function argument that transmitter $\txIndex$ holds,\footnote{We treat each $\stateSpaceElement_\txIndex$ as a deterministic quantity that can take any value in $\stateSpace_\txNum$. For stochastic scenarios, this means in particular that direct results specialize to every probability distribution supported on $\stateSpace_\txNum$, and even to the case that there are stochastic dependencies between different $\stateSpaceElement_\txIndex$.} and let $\objectiveFunctionEstimate$ be the estimate of $\objectiveFunction(\stateSpaceElement_1, \dots, \stateSpaceElement_\txNum)$ which the receiver obtains as a result of carrying out the \gls{otac} scheme. We then define the following approximation criteria:
\begin{enumerate}
  \item \label{item:analog-approximation-criteria-tail}
  We say that $\objectiveFunctionEstimate$ $(\tail,\errorprob)$-approximates $\objectiveFunction$ if uniformly for all $\stateSpaceElement_1 \in \stateSpace_1, \dots, \stateSpaceElement_\txNum \in \stateSpace_\txNum$,
    \begin{equation*}
    \Probability\left(
      \absolutevalue{
        \objectiveFunction(\stateSpaceElement_1, \dots, \stateSpaceElement_\txNum)
        -
        \objectiveFunctionEstimate
      }
      >
      \tail
    \right)
    \leq
    \errorprob.
    \end{equation*}
  \item \label{item:analog-approximation-criteria-mse}
  We say that $\objectiveFunctionEstimate$ approximates $\objectiveFunction$ with \gls{mse} \gls{converseAnalogErrorVariance} if
    \begin{equation*}
    \Expectation\left(
      \left(
        \objectiveFunction(\stateSpaceElement_1, \dots, \stateSpaceElement_\txNum)
        -
        \objectiveFunctionEstimate
      \right)^2
    \right)
    \leq
    \converseAnalogErrorVariance
    \end{equation*}
  for all $\stateSpaceElement_1 \in \stateSpace_1, \dots, \stateSpaceElement_\txNum \in \stateSpace_\txNum$.
  \item \label{item:analog-approximation-criteria-gauss}
  We say that $\objectiveFunctionEstimate$ is a Gaussian approximation of $\objectiveFunction$ with variance $\converseAnalogErrorVariance$ if
    \begin{equation*}
    \objectiveFunctionEstimate
    \sim
    \normaldistribution{\objectiveFunction(\stateSpaceElement_1, \dots, \stateSpaceElement_\txNum)}{\converseAnalogErrorVariance}
    \end{equation*}
  for all $\stateSpaceElement_1 \in \stateSpace_1, \dots, \stateSpaceElement_\txNum \in \stateSpace_\txNum$.
\end{enumerate}
\end{definition}
In the following lemma, we summarize some known and immediate implications between these approximation criteria.
\begin{lemma}
\label{lemma:analog-approximation-implications}
Let $\objectiveFunction$ be the function that is \gls{ota} computed, and let $\objectiveFunctionEstimate$ be the estimate of $\objectiveFunction(\stateSpaceElement_1, \dots, \stateSpaceElement_\txNum)$ at the receiver. We then have the following relations between the criteria in Definition~\ref{def:analog-approximation-criteria}:
\begin{enumerate}
  \item \label{item:analog-approximation-implication-gauss}
  If $\objectiveFunctionEstimate$ is a Gaussian approximation of $\objectiveFunction$ with variance $\converseAnalogErrorVariance$, then
  \begin{enumerate}
  \item \label{item:analog-approximation-implication-gauss-mse}
  $\objectiveFunctionEstimate$ approximates $\objectiveFunction$ with \gls{mse} $\converseAnalogErrorVariance$, and
  \item \label{item:analog-approximation-implication-gauss-tail}
  for every $\tail \in (0,\infty)$, $\objectiveFunctionEstimate$ $(\tail,\errorprob)$-approximates $\objectiveFunction$ with
    \begin{equation*}
    \errorprob
    =
    \frac{1}{\tail}
    \sqrt{\frac{2}{\pi}}
    \exp\left(-\frac{\tail^2}{2}\right).
    \end{equation*}
  \end{enumerate}
  \item \label{item:analog-approximation-implication-tail}
  If for every $\tail \in (0,\infty)$, we have $\errorprob(\tail) \in [0,1]$ such that $\objectiveFunctionEstimate$ $(\tail,\errorprob(\tail))$-approximates $\objectiveFunction$ and $\errorprob$ is a measurable function of $\tail$, then $\objectiveFunctionEstimate$ approximates $\objectiveFunction$ with \gls{mse}
    \begin{equation*}
    \converseAnalogErrorVariance
    =
    \int_0^{\infty}
      \errorprob(\sqrt{\tail})
      d\tail.
    \end{equation*}
  \item \label{item:analog-approximation-implication-mse}
  If $\objectiveFunctionEstimate$ approximates $\objectiveFunction$ with \gls{mse} $\converseAnalogErrorVariance$, then, for every $\tail \in \reals$, $\objectiveFunctionEstimate$ $(\tail, \converseAnalogErrorVariance/\tail^2)$-approximates $\objectiveFunction$.
\end{enumerate}
\end{lemma}
\begin{proof}
Item \ref{item:analog-approximation-implication-gauss-mse} is immediate. For item \ref{item:analog-approximation-implication-gauss-tail}, we apply \cite[Proposition 2.1.2]{vershynin2018high} to derive
\[
\Probability\left(
  \absolutevalue{
    \objectiveFunction(\stateSpaceElement_1, \dots, \stateSpaceElement_\txNum)
    -
    \objectiveFunctionEstimate
  }
  \geq
  \tail
\right)
=
\Probability\left(
  \objectiveFunction(\stateSpaceElement_1, \dots, \stateSpaceElement_\txNum)
  -
  \objectiveFunctionEstimate
  \geq
  \tail
\right)
+
\Probability\left(
  \objectiveFunctionEstimate
  -
  \objectiveFunction(\stateSpaceElement_1, \dots, \stateSpaceElement_\txNum)
  \geq
  \tail
\right)
=
2
\cdot
\frac{1}{\tail}
\frac{1}{\sqrt{2\pi}}
\exp\left(-\frac{\tail^2}{2}\right).
\]
For item \ref{item:analog-approximation-implication-tail}, we use~\cite[eq. 21.9]{billingsley1995probability} to obtain
\[
\Expectation\left(
  \left(
    \objectiveFunctionEstimate
    -
    \objectiveFunction(\stateSpaceElement_1, \dots, \stateSpaceElement_\txNum)
  \right)^2
\right)
=
\int_0^\infty
  \Probability\left(
    \left(
      \objectiveFunction(\stateSpaceElement_1, \dots, \stateSpaceElement_\txNum)
      -
      \objectiveFunctionEstimate
    \right)^2
    \geq
    \tail
  \right)
  d\tail
=
\int_0^\infty
  \Probability\left(
    \absolutevalue{
      \objectiveFunction(\stateSpaceElement_1, \dots, \stateSpaceElement_\txNum)
      -
      \objectiveFunctionEstimate
    }
    \geq
    \sqrt{\tail}
  \right)
  d\tail.
\]
For item \ref{item:analog-approximation-implication-mse}, we proceed very closely to the proof of Chebyshev's inequality and get
\[
\Probability\left(
  \absolutevalue{
    \objectiveFunction(\stateSpaceElement_1, \dots, \stateSpaceElement_\txNum)
    -
    \objectiveFunctionEstimate
  }
  \geq
  \tail
\right)
=
\Probability\left(
  \left(
    \objectiveFunction(\stateSpaceElement_1, \dots, \stateSpaceElement_\txNum)
    -
    \objectiveFunctionEstimate
  \right)^2
  \geq
  \tail^2
\right)
\leq
\frac{
  \Expectation\left(
    \left(
      \objectiveFunction(\stateSpaceElement_1, \dots, \stateSpaceElement_\txNum)
      -
      \objectiveFunctionEstimate
    \right)^2
  \right)
}{
  \tail^2
}.
\qedhere
\]
\end{proof}
So Gaussian approximation is the most stringent criterion, implying both \gls{mse} approximation and exponentially good $(\tail,\errorprob)$-approximation, with the other two notions implying each other.

\subsection{Digital Decoding Error}
For the digital messages, we consider the reconstruction error probability
\begin{equation}
\label{eq:reconstruction-error}
\decodingError
:=
\Probability\left(
  \messageRV_1 \neq \hat{\messageRV}_1
  \text{ or }
  \dots
  \text{ or }
  \messageRV_{\digitalTxNum} \neq \hat{\messageRV}_{\digitalTxNum}
\right),
\end{equation}
under uniform distributions of $\messageRV_1, \dots, \messageRV_{\digitalTxNum}$ (average error criterion).

\subsection{Framework for \gls{socc}}

For any given \gls{socc} scheme, we define its digital rates as $(\log(\codebookSize_1)/\blocklength, \dots, \log(\codebookSize_{\digitalTxNum})/\blocklength)$. Moreover, we define its analog rate as $\numAnalogTransmissions/\blocklength$.

\begin{definition}
\label{def:hads-achievable}
For $\gls{codebookRate}\codebookRate_1, \dots, \codebookRate_{\digitalTxNum}, \converseAnalogErrorVariance \in [0,\infty), \analograte \in [0,1]$, we say that \gls{socc} with digital rates $(\codebookRate_1, \dots, \codebookRate_{\digitalTxNum})$, analog rate \gls{analograte}, and analog error $\converseAnalogErrorVariance$ is \emph{achievable} over a channel $\channel$ if there is a sequence of \gls{socc} schemes (i.e., one scheme for each block length $\blocklength \in \naturals$) such that
\begin{enumerate}
  \item \label{item:hads-achievable-digital-rates} for all schemes in the sequence, the digital rates are (component-wise) at least $(\codebookRate_1, \dots, \codebookRate_{\digitalTxNum})$;
  \item \label{item:hads-achievable-analog-rate} for all of them, the analog rate is at least $\analograte$;
  \item \label{item:hads-achievable-constraints} all schemes in the sequence satisfy the power and amplitude constraints (\ref{eq:analog-amplitude-constraint}), (\ref{eq:digital-amplitude-constraint}), (\ref{eq:digital-power-constraint});
  \item \label{item:hads-achievable-digital-error} we have $\lim_{\blocklength \rightarrow \infty} \decodingError \rightarrow 0$ with $\decodingError$ defined in (\ref{eq:reconstruction-error});
  \item \label{item:hads-achievable-analog-error} all except finitely many schemes in the sequence have the property that $\objectiveFunctionEstimate_1, \dots, \objectiveFunctionEstimate_{\analogTxNum}$ approximate, respectively, $\objectiveFunction^{(1)}, \dots, \objectiveFunction^{(\numAnalogTransmissions)}$ with \gls{mse} $\converseAnalogErrorVariance$.
\end{enumerate}
\end{definition}
In this work, we make a first step towards determining the region of digital and analog rates with which \gls{socc} is achievable, at least for special cases of the framework defined above. In the following, we assume that all \gls{ota} computed functions are sum functions, i.e.,
\begin{equation}
\label{eq:sum-function-restriction}
\objectiveFunction^{(\indexAnalogTransmissions)}: (\stateSpaceElement_1, \dots, \stateSpaceElement_{\analogTxNum}) \mapsto \stateSpaceElement_1 + \dots + \stateSpaceElement_{\analogTxNum},
\text{ with domains }
\stateSpace_{1,1} = \dots = \stateSpace_{1,\numAnalogTransmissions} = \dots = \stateSpace_{\analogTxNum,1} = \dots = \stateSpace_{\analogTxNum,\numAnalogTransmissions} = [-1,1].
\end{equation}
In particular for the converse bounds it is clear that they cannot generalize to arbitrary functions.\footnote{Consider, for instance, the extreme case of \gls{ota} computed functions that map every point in their domain to the same value. Clearly, these functions can be computed at an arbitrarily high analog rate without impacting on the digital transmissions at all by letting the analog transmitters input $0$ to the channel.} However, slight variations of both the converse and achievable bounds we propose in this paper hold for significantly larger classes of functions as well (see Remark~\ref{remark:nonsum} for details).

\section{Main Result}
\label{sec:main-result}

For $\converseAnalogErrorVariance \in [0,\infty)$ and $\analograte \in [0,1]$, we define
\begin{multline*}
\hadsCapacity{\analograte}{\converseAnalogErrorVariance}
:=
\closure
\{
  (\codebookRate_1, \dots, \codebookRate_{\digitalTxNum})
  :~
  \text{\gls{socc} with digital rates }
  \codebookRate_1, \dots, \codebookRate_{\digitalTxNum},
  \text{ analog rate }
  \analograte
  \text{ and analog error }
  \converseAnalogErrorVariance
  \\
  \text{ is achievable over the channel $\channelAWGN$ for }
  \objectiveFunction^{(1)}, \dots, \objectiveFunction^{(\numAnalogTransmissions)}
  \text{ given in (\ref{eq:sum-function-restriction})}
\},
\end{multline*}
where \gls{closure} denotes topological closure. Before we state our main result regarding $\hadsCapacity{\analograte}{\converseAnalogErrorVariance}$, we define the rate regions
\begin{align}
\label{eq:inner-bound}
\hadsAchievable{\analograte}{\analograte'}
&:=
\left\{
  (\codebookRate_1, \dots, \codebookRate_{\digitalTxNum})
  ~:~
  \left(
    \frac{\codebookRate_1}{1-\analograte'},
    \dots,
    \frac{\codebookRate_{\digitalTxNum}}{1-\analograte'}
  \right)
  \in
  \gls{constrainedMacCapacity}
  \constrainedMacCapacity{\channelAWGN}\left(
    \frac{\powerconstraint_{\analogTxNum+1}}{1-\analograte},
    \dots,
    \frac{\powerconstraint_{\txNum}}{1-\analograte},
    \frac{\amplitudeConstraint_{\analogTxNum+1}}{\frac{\sqrt{2}}{\sqrt{2}-1}},
    \dots,
    \frac{\amplitudeConstraint_{\txNum}}{\frac{\sqrt{2}}{\sqrt{2}-1}}
  \right)
\right\}
\\
\label{eq:outer-bound}
\hadsConverse{\analograte}{\converseAnalogErrorVariance}
&:=
\left\{
  (\codebookRate_1, \dots, \codebookRate_{\digitalTxNum})
  ~:~
  \forall \digitalTxSubset \subseteq \naturalsInitialSection{\digitalTxNum}~
  \sum_{\txIndex \in \digitalTxSubset} \codebookRate_{\txIndex}
  \leq
  \min_{\txIndex \in \{0, \dots, \analogTxNum\}}\left(
    \unconstrainedGaussianCapacity\left(
      \frac{\sum_{\txIndex \in \digitalTxSubset} \powerconstraint_{\analogTxNum+\txIndex} + \txIndex^2\analogAmplitudeConstraint^2}
          {\noisestd^2}
    \right)
    -
    \positivePart{
      \frac{\analograte}{2}\log\frac{2\txIndex^2}{\pi\eulersNumber\converseAnalogErrorVariance}
    }
  \right)
\right\},
\end{align}
where for any $\generalReal \in \reals$, we define \gls{positivePart}, $\unconstrainedGaussianCapacity$ denotes the Gaussian capacity function
\[
\gls{unconstrainedGaussianCapacity}:
[0,\infty) \rightarrow [0,\infty),~
\generalReal
\mapsto
\frac{1}{2}
\log(1+\generalReal),
\]
and $\constrainedMacCapacity{\channel}$ denotes the (digital communication only) capacity region of a channel $\channel$ under average power constraint and amplitude constraint which is defined as
\begin{multline}
\label{eq:constrained-mac-capacity}
\constrainedMacCapacity{\channel}(
  \powerconstraint_1, \dots, \powerconstraint_\txNum,
  \amplitudeConstraint_1, \dots, \amplitudeConstraint_\txNum
)
=
\closure\big\{
  (\codebookRate_1, \dots, \codebookRate_\txNum)
  :~
  (\codebookRate_1, \dots, \codebookRate_\txNum)
  \text{ is achievable for $\channel$}
  \\
  \text{under average power constraints }
  \powerconstraint_1, \dots, \powerconstraint_\txNum
  \text{ and amplitude constraints }
  \amplitudeConstraint_1, \dots, \amplitudeConstraint_\txNum
\big\}.
\end{multline}
As noted by the authors, some of the methods proposed in~\cite{mamandipoor2014capacity} can be applied to determine this capacity region. For simplicity, we will instead be using inner bounds based on the capacity achieving distributions for power and amplitude constrained channels in the single user case~\cite{smith1971information} and outer bounds based on the Gaussian capacity function (i.e., the region of rates that are achievable if the amplitude constraints are disregarded). Particularly in the case in which the amplitude constraints are reasonably generous compared to the power constraints, this yields bounds which are not overly loose as we discuss in Section~\ref{sec:gaussian-mac}.

\begin{figure}
\centering
\begin{subfigure}[t]{.45\textwidth}
\begin{tikzpicture}
\begin{axis}[
  xlabel={$\digitalTxNum$},
  ylabel={sum rate in nats},
  legend pos=south east
]
  \addplot[black,dotted,mark=*,mark options=solid] table[x=k, y=capacity_nats] {gaussian_capacity_function_multiple_8dB.csv};
  \addplot[black,dashed,mark=o,mark options=solid] table[x=num_digital_tx, y=sum_rate_bound_nats, col sep=comma] {converse_power_digital_8dB_analog_5dB_beta_1_over_10.csv};
  \addplot[black,mark=square*] table[x=num_tx, y=achievable_sum_rate_nats, col sep=comma] {achievable_power_digital_8dB_analog_5dB_beta_1_over_10.csv};
  \legend{trivial converse, converse, achievable};
\end{axis}
\end{tikzpicture}
\caption{Analog rate fixed at $\analograte$ just below $1/10$, number $\digitalTxNum$ of digital transmitters variable.}
\label{fig:capacity-bounds-num-tx-variable}
\end{subfigure}
\hspace{.05\textwidth}
\begin{subfigure}[t]{.45\textwidth}
\begin{tikzpicture}
\begin{axis}[
  xlabel={$\analograte$},
  ylabel={sum rate in nats},
  legend pos=south west
]
  \addplot[black,dotted,domain=0:0.5] {1.9051835220871938};
  \addplot[black,dashed] table[x=analog_rate,y=sum_rate_bound_nats, col sep=comma] {converse_power_digital_8dB_analog_5dB_num_digital_tx_7.csv};
  \addplot[black] table[x=analog_rate,y=achievable_sum_rate_nats, col sep=comma] {achievable_power_digital_8dB_analog_5dB_num_digital_tx_7.csv};
  \legend{trivial converse, converse, achievable};
\end{axis}
\end{tikzpicture}
\caption{Number of digital transmitters fixed at $\digitalTxNum=7$, analog rate $\analograte$ variable. The achievable sum rate changes whenever $\analograte$ crosses a threshold of the form $1/\generalNaturalFour, \generalNaturalFour \in \naturals$. The achievable and converse curves are not monotonous since the analog approximation \gls{mse} changes in dependence of $\analograte$.}
\label{fig:capacity-bounds-analograte-variable}
\end{subfigure}
\caption{Theorem~\ref{theorem:capacity} sum rate bounds for $\noisestd^2 = \qty{0}{\dB}, \analogAmplitudeConstraint = \qty{2.5}{\dB}, \powerconstraint_{\analogTxNum+1} = \dots = \powerconstraint_{\txNum} = \qty{8}{\dB}, \amplitudeConstraint_{\analogTxNum+1} = \cdots = \amplitudeConstraint_\txNum = 2\sqrt{2\powerconstraint_\txIndex}$. The analog approximation is with \gls{mse} $\converseAnalogErrorVariance := \analograte' \noisestd^2/\analogAmplitudeConstraint^2$. \emph{Trivial converse} is the known sum rate bound for power-constrained Gaussian \gls{mac} without analog transmitters or \gls{otac}, \emph{converse} is the second inclusion in Theorem~\ref{theorem:capacity}, and \emph{achievable} is the first inclusion in Theorem~\ref{theorem:capacity}, where $\constrainedMacCapacityAchievable{\channelAWGN}$ defined in Section~\ref{sec:gaussian-mac} is used as a bound for $\constrainedMacCapacity{\channelAWGN}$. There need to be at least four analog transmitters in the system for all converse bounds shown in these figures to hold.}
\label{fig:capacity-bounds}
\end{figure}

In the course of this study, we prove the following bounds for $\hadsCapacity{\cdot}{\cdot}$:
\begin{theorem}
\label{theorem:capacity}
Suppose $\analograte \in (0,1)$ and define $\analograte' := \min\{1/\generalNaturalFour :~ \generalNaturalFour \in \naturals,~ 1/\generalNaturalFour > \analograte\}$. Then we have
\[
\hadsAchievable{\analograte}{\analograte'} \subseteq \hadsCapacity{\analograte}{\analograte'\noisestd^2/\analogAmplitudeConstraint^2} \subseteq \hadsConverse{\analograte}{\analograte'\noisestd^2/\analogAmplitudeConstraint^2},
\]
where $\hadsAchievable{\analograte}{\analograte'}$ is defined in (\ref{eq:inner-bound}) and $\hadsConverse{\analograte}{\analograte'\noisestd^2/\analogAmplitudeConstraint^2}$ is defined in (\ref{eq:outer-bound}).
\end{theorem}
\begin{proof}
The first inclusion (achievability) is proved in Lemma~\ref{lemma:digital-mac-with-ota-c}, and the second inclusion (converse) in Corollary~\ref{cor:converse}. For exact details on how Lemma~\ref{lemma:digital-mac-with-ota-c} implies the achievability part of Theorem~\ref{theorem:capacity}, see the end of Section~\ref{sec:random-hybrid}.
\end{proof}

In Fig.~\ref{fig:capacity-bounds}, we show a numerical example of the converse and achievable bounds given in Theorem~\ref{theorem:capacity}. The result establishes that there is a meaningful tradeoff between analog computations and digital communications on the same channel resources. We are also able to give a first impression of how this tradeoff looks numerically, however, questions regarding inner and outer bounds that are tighter and valid for arbitrary parameter choices remain open for further research. One open question in particular is the decoupling between $\analograte$ and $\converseAnalogErrorVariance$. While the converse result states a bound (\ref{eq:outer-bound}) on the digital rates for every choice of $\analograte$ and $\converseAnalogErrorVariance$, the achievability part of Theorem~\ref{theorem:capacity} gives a result only for
\begin{equation}
\label{eq:analograte-mse-relationship}
\converseAnalogErrorVariance = \analograte'\noisestd^2/\analogAmplitudeConstraint^2,
\end{equation}
with $\analograte'$ defined in the statement of Theorem~\ref{theorem:capacity}. This functional relationship means that as $\analograte$ varies, so does $\converseAnalogErrorVariance$. In Fig.~\ref{fig:capacity-bounds-analograte-variable}, we therefore plot both the inner and outer bound on the achievable sum rate for values of $\converseAnalogErrorVariance$ that vary according to the functional relationship (\ref{eq:analograte-mse-relationship}). Both $\converseAnalogErrorVariance$ (upon which the converse bound depends) and the achievable digital sum rate have a dependence on $\analograte'$. Since $\analograte'$ changes only when $\analograte$ crosses a threshold of the form $1/\generalNaturalFour$ with $\generalNaturalFour \in \naturals$, there is a visible jump in both the achievable and converse curves in Fig.~\ref{fig:capacity-bounds-analograte-variable} resulting in a zigzag pattern. Since the jumps in the graph occur only at multiplicative inverses of integers, they grow further apart as $\analograte$ increases. It is a highly relevant question for further research if achievable schemes can be proposed for \emph{any} combination of $\analograte$ and $\converseAnalogErrorVariance$. This includes both the question of what digital rates are achievable for pairs of $\analograte$, $\converseAnalogErrorVariance$ that do not satisfy (\ref{eq:analograte-mse-relationship}), and the question if performance can be improved in the cases where $\analograte$ is not just below the multiplicative inverse of an integer (i.e., $\analograte \not\approx \analograte'$). It should be noted, however, that despite this restriction, the \gls{socc} scheme proposed in this paper is already suitable for many practical communication systems that involve \gls{otac}. We give a few examples for this and discuss possible directions for future research to remove or at least mitigate this restriction in Section~\ref{sec:analograte-mse-relationship}.

It should also be noted that the gap between inner and outer bounds can be arbitrarily large in some cases, such as when the amplitude constraints approach zero. However, our numerical evaluations indicate that the behavior we show in Fig.~\ref{fig:capacity-bounds} is fairly typical for most practically relevant parameter regimes. The behavior of the plots shows that the gap is present, but does not grow in an extreme way.

\begin{remark}
\label{remark:time-sharing-amplitude-constraint}
\emph{(Time sharing comparison with amplitude constraints).}
If the same analog amplitude constraints are to be observed as in an equivalent \gls{socc} scheme and the same analog computation rate and \gls{mse} is to be achieved, then a time sharing scheme would need to spend \emph{all} channel uses for the analog computations. Therefore, such a time sharing scheme would always achieve a digital sum rate of $0$. If, on the other hand, the same digital sum rate needed to be achieved, the relationship \eqref{eq:analograte-mse-relationship} between analog rate and \gls{mse} means that only one single channel use is available for each analog computation which would result in a significantly higher \gls{mse}, especially for small values of $\analograte$. For the parameters used in Fig.~\ref{fig:capacity-bounds-analograte-variable}, the \gls{mse} of the equivalent time sharing scheme would be approximately $0.32$ for all values of $\analograte$ while the \gls{mse} of the \gls{socc} scheme would vanish for $\analograte \rightarrow 0$ and be approximately $0.16$ at $\analograte = 0.5$.
\end{remark}

\begin{remark}
\label{remark:time-sharing}
\emph{(\gls{socc} without amplitude constraints).}
If we drop the amplitude constraints and retain only the average power constraints (\ref{eq:analog-power-constraint}) and (\ref{eq:digital-power-constraint}), there is a straightforward \gls{socc} scheme based on time sharing that achieves the inner bound of Theorem~\ref{theorem:capacity}:

Let $\analograte$ be an analog rate, let $\analograte'$ be defined as in Theorem~\ref{theorem:capacity}, and let $(\codebookRate_1, \dots, \codebookRate_{\digitalTxNum})$ be a tuple of digital rates such that $(\codebookRate_1/(1-\analograte'), \dots, \codebookRate_{\digitalTxNum}/(1-\analograte'))$ is an inner point of the capacity region of the Gaussian \gls{mac} with noise power $\noisestd^2$ under average power constraints $(\powerconstraint_{\analogTxNum+1}/(1-\analograte'), \dots, \powerconstraint_{\txNum}/(1-\analograte'))$. Then there is $\analograte'' > \analograte'$ such that $(\codebookRate_1/(1-\analograte''), \dots, \codebookRate_{\digitalTxNum}/(1-\analograte''))$ is achievable for the Gaussian \gls{mac} as well. Fix a sequence of codes that achieves these rates, and define a \gls{socc} scheme for each block length $\blocklength$ as follows:
\begin{itemize}
  \item Let the number of analog \gls{ota} computations be $\numAnalogTransmissions := \lfloor \blocklength\analograte' \rfloor$.
  \item Define the analog pre-processors (i.e., $\txIndex \in \naturalsInitialSection{\analogTxNum}$) as follows:
  \[
    \preprocHybrid_\txIndex(\stateSpaceElement_{\txIndex,1}, \dots, \stateSpaceElement_{\txIndex,\numAnalogTransmissions})
    :=
    \left(
      \stateSpaceElement_{\txIndex,1} \cdot \frac{\analogAmplitudeConstraint}{\sqrt{\analograte'}},
      \dots,
      \stateSpaceElement_{\txIndex,\numAnalogTransmissions} \cdot \frac{\analogAmplitudeConstraint}{\sqrt{\analograte'}},
      0,\dots,0
    \right)
  \]
  \item Transmit the digital messages during the last $\lfloor \blocklength(1-\analograte') \rfloor$ channel uses of the transmission block (being encoded with a code of appropriate block length) and transmit $0$ during all other channel uses.
  \item Since the analog and digital transmissions are completely separated in time domain, we can use the Gaussian \gls{mac} decoder on the last $\lfloor \blocklength(1-\analograte') \rfloor$ channel outputs to decode the digital messages, and for $\indexAnalogTransmissions \in \naturalsInitialSection{\numAnalogTransmissions}$, we compute the analog function estimate as
  \[
    \objectiveFunctionEstimate_\indexAnalogTransmissions
    :=
    \outputRV_\indexAnalogTransmissions \cdot \frac{\sqrt{\analograte'}}{\analogAmplitudeConstraint}.
  \]
\end{itemize}
The total power output of the analog pre-processors is at most $\numAnalogTransmissions \cdot \analogAmplitudeConstraint^2/\analograte' \leq \analogAmplitudeConstraint^2 \blocklength$, satisfying (\ref{eq:analog-power-constraint}), and due to the average power constraint for the digital Gaussian \gls{mac} code, the total power of the output of digital transmitter $\txIndex$ is at most $\lfloor \blocklength(1-\analograte') \rfloor \powerconstraint_{\analogTxNum+\txIndex}/(1-\analograte') \leq \blocklength\powerconstraint_{\analogTxNum+\txIndex}$, satisfying (\ref{eq:digital-power-constraint}). The digital rate of transmitter $\txIndex$ is $\codebookRate_\txIndex/(1-\analograte'') \cdot \lfloor \blocklength(1-\analograte') \rfloor / \blocklength \geq \codebookRate_\txIndex$ for sufficiently large $\blocklength$. As for the number of analog transmissions, for sufficiently large $\blocklength$, we have $\numAnalogTransmissions = \lfloor \blocklength \analograte' \rfloor \geq \blocklength \analograte$. The analog function estimates can be written as
\[
  \objectiveFunctionEstimate_\indexAnalogTransmissions
  =
  (
    \inputRV_{1,\indexAnalogTransmissions} + \dots + \inputRV_{\analogTxNum,\indexAnalogTransmissions} + \noiseRV_\indexAnalogTransmissions
  )
  \cdot
  \frac{\sqrt{\analograte'}}{\analogAmplitudeConstraint}
  =
  \left(
    \stateSpaceElement_{1,\indexAnalogTransmissions}\cdot\frac{\analogAmplitudeConstraint}{\sqrt{\analograte'}}
    +
    \dots
    +
    \stateSpaceElement_{\analogTxNum,\indexAnalogTransmissions}\cdot\frac{\analogAmplitudeConstraint}{\sqrt{\analograte'}}
    +
    \noiseRV_\indexAnalogTransmissions
  \right)
  \cdot
  \frac{\sqrt{\analograte'}}{\analogAmplitudeConstraint}
  =
  \stateSpaceElement_{1,\indexAnalogTransmissions}
  +
  \dots
  +
  \stateSpaceElement_{\analogTxNum,\indexAnalogTransmissions}
  +
  \noiseRV_\indexAnalogTransmissions
  \cdot
  \frac{\sqrt{\analograte'}}{\analogAmplitudeConstraint},
\]
and hence, the \gls{mse} is $\analograte'\noisestd^2/\analogAmplitudeConstraint^2$.

Therefore, this rather straightforward \gls{socc} scheme shows the achievability bound of Theorem~\ref{theorem:capacity}. However, the peak power consumption of the analog transmitters is $\analogAmplitudeConstraint^2/\analograte'$, which can make the scheme highly unpractical for small values of the analog rate $\analograte'$. In particular, in the case in which the analog rate is $0$ in the sense that $\numAnalogTransmissions$ grows only sub-linearly in $\blocklength$, the peak power consumption at the analog transmitters tends to infinity with the block length.\footnote{Theorem~\ref{theorem:capacity} does not strictly speaking cover this case, however, the achievability result of Lemma~\ref{lemma:digital-mac-with-ota-c} does apply, as we discuss in Remark~\ref{remark:zero-rate}.} Since the absence of amplitude constraints in our scenario permits this straightforward scheme of very limited practical significance, we argue that it is particularly relevant to analyze this communication system subject not only to average power, but also peak amplitude constraints. However, we remark that in case of large $\analograte$, this very simple transmission scheme could also be considered for implementation and the scheme we propose in Section~\ref{sec:achievability} may not be necessary.
\end{remark}

\begin{remark}
\label{remark:nonsum}
\emph{(Theorem~\ref{theorem:capacity} with functions other than sums of values in $[-1,1]$).} Theorem~\ref{theorem:capacity} does not hold for the \gls{otac} of arbitrary functions, however, it does hold for some functions that are not sums, albeit with slightly more complicated expressions for the bounds. In this paper, we give a full formal result statement only for achievability over a slightly different channel and a slightly different analog error criterion (see Corollary~\ref{cor:nom}), but in the following, we briefly sketch how inner and outer bounds akin to the ones in Theorem~\ref{theorem:capacity} can be derived for some functions other than sums.

\emph{Achievability} can be shown for functions in $\Fmon$ (see Definition~\ref{def:Fmon}) for the non-fading real channels treated in Theorem~\ref{theorem:capacity} with a minor modification of the argument used in the proof of Corollary~\ref{cor:nom}. In doing so, $(\tail,\errorprob)$-approximation is proved instead of approximation with \gls{mse} $\converseAnalogErrorVariance$. So Lemma~\ref{lemma:analog-approximation-implications}-\ref{item:analog-approximation-implication-mse} has to be used to obtain a \gls{mse} analog approximation guarantee fully matching the framework of Theorem~\ref{theorem:capacity}.

The \emph{converse} bound also does not at its core rely on the computed function being a sum function. The proof of Lemma~\ref{lemma:converse} can be checked to confirm that the only property $\objectiveFunction$ has to satisfy is that for every $\analogTxSubset \subseteq \naturalsInitialSection{\analogTxNum}$, there is a distribution of $(\stateSpaceElement_\txIndex)_{\txIndex \in \naturalsInitialSection{\analogTxNum}}$ such that if $\stateSpaceElement_\txIndex = 0$ for all $\txIndex \notin \analogTxSubset$, we have that $\objectiveFunction(\stateSpaceElement_1, \dots, \stateSpaceElement_{\analogTxNum})$ follows a uniform distribution in $[-\cardinality{\analogTxSubset},\cardinality{\analogTxSubset}]$. It is only a minor modification to prove this also for a different range of uniform distribution, but this will slightly change the expression for the converse bound.

In conclusion, a statement of the style of Theorem~\ref{theorem:capacity} can be made for every function which satisfies both of the criteria described in the previous two paragraphs. We expect that many practically relevant functions will satisfy both; one example of a function that does is the $2$-norm which we show in our simulations in Section~\ref{sec:simulations} can also be computed with \gls{socc}.
\end{remark}

\section{Achievability}
\label{sec:achievability}
In this section, we prove the first inclusion of Theorem~\ref{theorem:capacity}. To this end, we show in Subsection~\ref{sec:construction} that any code for digital communication over the Gaussian \gls{mac} can be transformed into a code for \gls{socc} simply by composing the encoding and decoding operations with suitable additional linear processing. In Subsection~\ref{sec:random-hybrid}, we apply this construction to random codes for the amplitude-constrained Gaussian \gls{mac} to prove the first inclusion of Theorem~\ref{theorem:capacity}. We remark that although this involvement of random codes makes the present work predominantly theoretical, it is straightforward to also apply the construction in Subsection~\ref{sec:construction} to practically relevant codes such as \gls{ldpc} or polar codes. Finally, in Subsection~\ref{sec:gaussian-mac}, we discuss how to numerically evaluate simple inner and outer bounds for the rate region $\constrainedMacCapacity{\channelAWGN}$ that appears in Theorem~\ref{theorem:capacity}.

\subsection{Construction of \gls{socc}-Compatible Digital Codes}
\label{sec:construction}
In this section, we propose a construction of encoding schemes for the digital messages which are compatible with analog \gls{otac} in the sense that they are (under certain assumptions) indifferent to whether an analog \gls{otac}  is happening on the same channel resources and to what value is computed. And conversely, in a sense made precise below, the encoded signals do not disturb the analog \glspl{otac}. Our construction uses an existing digital coding scheme for the \gls{mac} and composes it with a map that is carefully constructed to achieve the above goals.

The following lemma shows the existence of certain linear maps that we will then use for the construction outlined above.

\begin{lemma}
\label{lemma:orthogonal}
Let $\generalNatural \in \naturals$. Then there are linear maps $\planeMap_\generalNatural: \reals^{\generalNatural-1} \rightarrow \reals^\generalNatural$ and $\planeMapInv_\generalNatural: \reals^\generalNatural \rightarrow \reals^{\generalNatural-1}$ (the matrix transpose of $\planeMap_\generalNatural$) with the following properties:
\begin{enumerate}
 \item \label{item:orthogonal-inv} $\planeMapInv_\generalNatural \circ \planeMap_\generalNatural = \identityMapping{\reals^{\generalNatural-1}}$, where $\circ$ denotes composition and $\identityMapping{\reals^{\generalNatural-1}}$ denotes the identity mapping on $\reals^{\generalNatural-1}$.
 \item \label{item:orthogonal-plane} Let $\generalReal^{\generalNatural-1} \in \reals^{\generalNatural-1}$ be arbitrary and let $\generalRealTwo^\generalNatural := \planeMap_\generalNatural(\generalReal^{\generalNatural-1})$. Then $\generalRealTwo_1 + \dots + \generalRealTwo_\generalNatural = 0$.
 \item \label{item:orthogonal-zero} $\planeMapInv_\generalNatural(1, \dots, 1) = 0$.
 \item \label{item:orthogonal-norm} Let $\generalReal^{\generalNatural-1} \in \reals^{\generalNatural-1}$ be arbitrary. Then $\euclidNorm{\planeMap_\generalNatural(\generalReal^{\generalNatural-1})}=\euclidNorm{\generalReal^{\generalNatural-1}}$.
 \item \label{item:orthogonal-max} Let $\generalReal^{\generalNatural-1} \in \reals^{\generalNatural-1}$ be arbitrary. Then $\maximumNorm{\planeMap_\generalNatural(\generalReal^{\generalNatural-1})} < \maximumNorm{\generalReal^{\generalNatural-1}} \cdot \sqrt{2}/(\sqrt{2}-1) < \maximumNorm{\generalReal^{\generalNatural-1}} \cdot 3.42$.
 \end{enumerate}
\end{lemma}
The proof of Lemma~\ref{lemma:orthogonal} is relegated to Appendix~\ref{appendix:lemma-orthogonal}.

\begin{figure}
    \centering
    \begin{tikzpicture}[xscale=.95]
        \coordinate (in1) at (0,3);
        \coordinate (inK) at (0,0);
        
        \node[rectangle,draw,minimum width=1.5cm,minimum height=1cm] (enc1) at (2,3) {$\preproc_1$};
        \node (dots) at (2,1.5) {\Shortstack{. . . . . .}};
        \node[rectangle,draw,minimum width=1.5cm,minimum height=1cm] (encK) at (2,0) {$\preproc_\txNum$};
        
        \node[rectangle,draw,minimum width=2cm,minimum height=1cm] (planeMap1) at (5.5,3) {$\planeMap_{\blocklength_1} \oplus \dots \oplus \planeMap_{\blocklength_\numAnalogTransmissions}$};
        \node[rectangle,draw,minimum width=2cm,minimum height=1cm] (planeMapK) at (5.5,0) {$\planeMap_{\blocklength_1} \oplus \dots \oplus \planeMap_{\blocklength_\numAnalogTransmissions}$};
        
        \node[rectangle,draw,minimum width=1.5cm,minimum height=5cm,align=center] (compoundChannel) at (8.75,1.5) {$\compoundChannel{\stateSpaceElement_1}^{\blocklength_1}$ \\ $\otimes$ \\ $\vdots$ \\ $\otimes$ \\ $\compoundChannel{\stateSpaceElement_\numAnalogTransmissions}^{\blocklength_\numAnalogTransmissions}$};
        
        \node[rectangle,draw,minimum width=2cm,minimum height=1cm] (planeMapInv) at (12,1.5) {$\planeMapInv_{\blocklength_1} \oplus \dots \oplus \planeMapInv_{\blocklength_\numAnalogTransmissions}$};
        
        \node[rectangle,draw,minimum width=1.5cm,minimum height=1cm] (dec) at (15.5,1.5) {$\postproc$};
        
        \coordinate (out) at (19,1.5);
        
        \draw[->] (in1) -- (enc1) node[midway,above] {$\messageRV_1$};
        \draw[->] (inK) -- (encK) node[midway,above] {$\messageRV_\txNum$};
        \draw[->] (enc1) -- (planeMap1) node[midway,above] {$\inputRV_1^{\blocklength-\numAnalogTransmissions}$};
        \draw[->] (encK) -- (planeMapK) node[midway,above] {$\inputRV_\txNum^{\blocklength-\numAnalogTransmissions}$};
        \draw[->] (planeMap1) -- (planeMap1-|compoundChannel.west) node[midway,above] {${\inputRV_1'}^\blocklength$};
        \draw[->] (planeMapK) -- (planeMapK-|compoundChannel.west) node[midway,above] {${\inputRV_\txNum'}^\blocklength$};
        \draw[->] (compoundChannel) -- (planeMapInv) node[midway,above] {${\outputRV'}^\blocklength$};
        \draw[->] (planeMapInv) -- (dec) node[midway,above] {$\outputRV^{\blocklength-\numAnalogTransmissions}$};
        \draw[->] (dec) -- (out) node[midway,above] {$\hat{\messageRV}_1, \dots, \hat{\messageRV}_\txNum$};
        
        \draw[red] (1,3.75) -- (7.05,3.75) node[pos=.4,above] {$\preproc_1'$} -- (7.05,2.25) -- (1,2.25) -- (1,3.75);
        \draw[red] (1,0.75) -- (7.05,0.75) node[pos=.4,above] {$\preproc_\txNum'$} -- (7.05,-.75) -- (1,-.75) -- (1,0.75);
        \draw[red] (10.45,2.25) -- (16.5,2.25) node[pos=.6,above] {$\postproc'$} -- (16.5,.75) -- (10.45,.75) -- (10.45,2.25);
        \draw[blue] (3.95,4.25) -- (13.55,4.25) node[midway,above] {$\channelAWGN^{\blocklength-\numAnalogTransmissions}$} -- (13.55,-1.25) -- (3.95,-1.25) -- (3.95,4.25);
    \end{tikzpicture}
    \caption{Graphical overview of the proof idea of Lemma~\ref{lemma:zerosum}. The notation $\planeMap_{\blocklength_1} \oplus \dots \oplus \planeMap_{\blocklength_\numAnalogTransmissions}$ means that the operations $\planeMap_{\blocklength_1}, \dots, \planeMap_{\blocklength_\numAnalogTransmissions}$ are applied in parallel to consecutive, nonoverlapping blocks of lengths $\blocklength_1 - 1, \dots, \blocklength_\numAnalogTransmissions - 1$.}
    \label{figure:zerosum}
\end{figure}

We next introduce some terminology in preparation for the proofs of the following lemmas. Given a block length $\blocklength$ \gls{mac} $\channel^{(\blocklength)}$ with input alphabets \gls{inputAlphabets} and output alphabet \gls{outputAlphabet} at each channel use\footnote{That is, $\channel^{(\blocklength)}$ is a stochastic kernel from $\inputAlphabet_1^\blocklength \times \cdots \times \inputAlphabet_\txNum^\blocklength$ to $\outputAlphabet^\blocklength$.}, we call $(\gls{preproc}, \gls{postproc})$ an $\gls{decodingError}$-code for $\channel^{(\blocklength)}$ of size $(\codebookSize_1,\dots, \codebookSize_\txNum)$ if $\preproc_1, \dots, \preproc_\txNum, \postproc$ have the following properties:
\begin{itemize}
  \item $\preproc_1: \naturalsInitialSection{\codebookSize_1} \rightarrow \inputAlphabet_1^\blocklength, \dots, \preproc_\txNum: \naturalsInitialSection{\codebookSize_\txNum} \rightarrow \inputAlphabet_\txNum^\blocklength$
  \item $\postproc: \outputAlphabet^\blocklength \rightarrow \naturalsInitialSection{\codebookSize_1} \times \cdots \times \naturalsInitialSection{\codebookSize_\txNum}$
  \item The code has average decoding error at most $\decodingError$. That is, if messages $\messageRV_1 \in \naturalsInitialSection{\codebookSize_1}, \dots, \messageRV_\txNum \in \naturalsInitialSection{\codebookSize_\txNum}$ are drawn uniformly at random and the input sequences $\preproc_1(\messageRV_1), \dots, \preproc_\txNum(\messageRV_\txNum)$ are transmitted through the channel $\channel^{(\blocklength)}$ yielding the channel output sequence $\outputRV^\blocklength$, then
  \[
    \Probability\big(
      \postproc(\outputRV^\blocklength) \neq (\messageRV_1, \dots, \messageRV_\txNum)
    \big)
    \leq
    \decodingError.
  \]
\end{itemize}

The following lemma shows how to apply Lemma~\ref{lemma:orthogonal} to communication systems. The proof is based on the idea of defining modified encoders and decoders by adding additional pre- and post-processing steps to the original ones. We show that these additional steps in conjunction with the natural channel ``emulate'' a channel that allows us to draw conclusions from the decoding error guarantee of the original encoders and decoder. As a conceptual tool for our proof, we introduce a version of the Gaussian \gls{mac} with bias. For $\stateSpaceElement \in \reals$, let $\compoundChannel{\stateSpaceElement}$ be the \gls{mac} given by
\begin{equation}
\label{eq:channel-awgn-bias}
\outputRV = \inputRV_1 + \dots + \inputRV_\txNum + \stateSpaceElement + \noiseRV.
\end{equation}
The idea of the proof is illustrated in Fig.~\ref{figure:zerosum} which can also serve as an overview of the mathematical symbols we use.

\begin{lemma}
\label{lemma:zerosum}
Let $\blocklength_1, \dots, \blocklength_\numAnalogTransmissions \in \naturals$ and $\blocklength := \blocklength_1 + \dots + \blocklength_\numAnalogTransmissions$. Suppose $(\preproc_1, \dots, \preproc_\txNum, \postproc)$ is an $\decodingError$-code of size $(\codebookSize_1,\dots, \codebookSize_\txNum)$ for the channel $\channelAWGN^{\blocklength-\numAnalogTransmissions}$, the $(\blocklength-\numAnalogTransmissions)$-fold product of the channel $\channelAWGN$ given in \eqref{eq:channel-awgn}. Then there is an $\decodingError$-code $(\preproc_1', \dots, \preproc_\txNum', \postproc')$ of size  $(\codebookSize_1,\dots, \codebookSize_\txNum)$ for the channel\footnote{$\compoundChannel{\stateSpaceElement_1}^{\blocklength_1} \otimes \dots \otimes \compoundChannel{\stateSpaceElement_\numAnalogTransmissions}^{\blocklength_\numAnalogTransmissions}$ denotes the product channel obtained by using the channel $\compoundChannel{\stateSpaceElement_1}$ consecutively $\blocklength_1$ times, $\dots$, using the channel $\compoundChannel{\stateSpaceElement_\numAnalogTransmissions}$ consecutively $\blocklength_\numAnalogTransmissions$ times.} $\compoundChannel{\stateSpaceElement_1}^{\blocklength_1} \otimes \dots \otimes \compoundChannel{\stateSpaceElement_\numAnalogTransmissions}^{\blocklength_\numAnalogTransmissions}$ with the following properties:
\begin{enumerate}
 \item \label{item:zerosum-power} \emph{(Conservation of total power).}
   $
   \forall \txIndex \in \naturalsInitialSection{\txNum}~
   \forall \messageRealization \in \naturalsInitialSection{\codebookSize_\txIndex}~
   \euclidNorm{\preproc_\txIndex'(\messageRealization)}
   =
   \euclidNorm{\preproc_\txIndex(\messageRealization)}
   $.
 \item \label{item:zerosum-peakpower} \emph{(Peak amplitude bound).}
   $
     \forall \txIndex \in \naturalsInitialSection{\txNum}~
     \forall \messageRealization \in \naturalsInitialSection{\codebookSize_\txIndex}~
     \maximumNorm{\preproc_\txIndex'(\messageRealization)} \leq \sqrt{2}/(\sqrt{2}-1)\maximumNorm{\preproc_\txIndex(\messageRealization)} < 3.42\maximumNorm{\preproc_\txIndex(\messageRealization)}
   $.
 \item \label{item:zerosum} \emph{(Zero sums).}
   Let
   $
     \txIndex \in \naturalsInitialSection{\txNum},
     \messageRealization \in \naturalsInitialSection{\codebookSize_\txIndex}
   $
   and
   $
     \indexAnalogTransmissions \in \naturalsInitialSection{\numAnalogTransmissions}
   $
   be arbitrary and define ${\inputRV'}^\blocklength := \preproc_\txIndex'(\messageRealization)$. Then, $\inputRV_{\blocklength_1 + \dots + \blocklength_{\indexAnalogTransmissions-1} + 1}' + \dots + \inputRV_{\blocklength_1 + \dots + \blocklength_{\indexAnalogTransmissions}}' = 0$.
\end{enumerate}
\end{lemma}
\begin{proof}
For every $\indexAnalogTransmissions$, we use the maps $\planeMap_{\blocklength_\indexAnalogTransmissions}, \planeMapInv_{\blocklength_\indexAnalogTransmissions}$ from Lemma~\ref{lemma:orthogonal}. For the rest of this proof, fix $\messageRealization = (\messageRealization_1, \dots, \messageRealization_\txNum)$. For every $\txIndex$, denote $\inputRV_\txIndex^{\blocklength-\numAnalogTransmissions} := \preproc_\txIndex(\messageRealization_\txIndex)$ and define $\preproc_\txIndex'(\messageRealization_\txIndex) := {\inputRV_\txIndex'}^\blocklength$ by
\[
\forall \indexAnalogTransmissions \in \naturalsInitialSection{\numAnalogTransmissions}~
 (\inputRV_{\txIndex, \blocklength_1 + \dots + \blocklength_{\indexAnalogTransmissions-1} + 1}', \dots, \inputRV_{\txIndex, \blocklength_1 + \dots + \blocklength_{\indexAnalogTransmissions}}')
 :=
 \planeMap_{\blocklength_\indexAnalogTransmissions}(\inputRV_{\txIndex, \blocklength_1 + \dots + \blocklength_{\indexAnalogTransmissions-1} - \indexAnalogTransmissions + 2}, \dots, \inputRV_{\txIndex, \blocklength_1 + \dots + \blocklength_{\indexAnalogTransmissions} - \indexAnalogTransmissions}),
\]
where for $\indexAnalogTransmissions=1$, we use the convention $\blocklength_1 + \dots + \blocklength_0 = 0$.

Item \ref{item:zerosum} is immediate from Lemma~\ref{lemma:orthogonal}-\ref{item:orthogonal-plane}. To prove item \ref{item:zerosum-power}, we apply Lemma~\ref{lemma:orthogonal}-\ref{item:orthogonal-norm} and obtain
\begin{align*}
\euclidNorm{{\inputRV_\txIndex'}^\blocklength}^2
&=
\sum_{\indexAnalogTransmissions=1}^\numAnalogTransmissions
  \euclidNorm{
    \planeMap_{\blocklength_\indexAnalogTransmissions}(\inputRV_{\txIndex, \blocklength_1 + \dots + \blocklength_{\indexAnalogTransmissions-1} - \indexAnalogTransmissions + 2}, \dots, \inputRV_{\txIndex, \blocklength_1 + \dots + \blocklength_{\indexAnalogTransmissions} - \indexAnalogTransmissions})
  }^2
\\
&=
\sum_{\indexAnalogTransmissions=1}^\numAnalogTransmissions
  \euclidNorm{
    (\inputRV_{\txIndex, \blocklength_1 + \dots + \blocklength_{\indexAnalogTransmissions-1} - \indexAnalogTransmissions + 2}, \dots, \inputRV_{\txIndex, \blocklength_1 + \dots + \blocklength_{\indexAnalogTransmissions} - \indexAnalogTransmissions})
  }^2
\\
&=
\euclidNorm{\inputRV_\txIndex^{\blocklength-\numAnalogTransmissions}}^2.
\end{align*}
Similarly, to prove item \ref{item:zerosum-peakpower}, we apply Lemma~\ref{lemma:orthogonal}-\ref{item:orthogonal-max} and get
\begin{align*}
\maximumNorm{{\inputRV'}^\blocklength}
&=
\max_{\indexAnalogTransmissions\in\naturalsInitialSection{\numAnalogTransmissions}}
  \maximumNorm{
    \planeMap_{\blocklength_\indexAnalogTransmissions}(\inputRV_{\txIndex, \blocklength_1 + \dots + \blocklength_{\indexAnalogTransmissions-1} - \indexAnalogTransmissions + 2}, \dots, \inputRV_{\txIndex, \blocklength_1 + \dots + \blocklength_{\indexAnalogTransmissions} - \indexAnalogTransmissions})
  }
\\
&\leq
\sqrt{2}/(\sqrt{2}-1)
\max_{\indexAnalogTransmissions\in\naturalsInitialSection{\numAnalogTransmissions}}
  \maximumNorm{
    (\inputRV_{\txIndex, \blocklength_1 + \dots + \blocklength_{\indexAnalogTransmissions-1} - \indexAnalogTransmissions + 2}, \dots, \inputRV_{\txIndex, \blocklength_1 + \dots + \blocklength_{\indexAnalogTransmissions} - \indexAnalogTransmissions})
  }
\\
&=
\sqrt{2}/(\sqrt{2}-1)\maximumNorm{\inputRV_\txIndex^{\blocklength-\numAnalogTransmissions}}.
\end{align*}

We denote the channel output obtained by passing ${\inputRV'}^\blocklength$ through $\compoundChannel{\stateSpaceElement_1}^{\blocklength_1} \otimes \dots \otimes \compoundChannel{\stateSpaceElement_\numAnalogTransmissions}^{\blocklength_\numAnalogTransmissions}$ as ${\outputRV'}^\blocklength$ and define $\outputRV^{\blocklength-\numAnalogTransmissions}$ by
\[
\forall \indexAnalogTransmissions \in \naturalsInitialSection{\numAnalogTransmissions}~
 (\outputRV_{\blocklength_1 + \dots + \blocklength_{\indexAnalogTransmissions-1} - \indexAnalogTransmissions + 2}, \dots, \outputRV_{\blocklength_1 + \dots + \blocklength_{\indexAnalogTransmissions} - \indexAnalogTransmissions})
 :=
 \planeMapInv_{\blocklength_\indexAnalogTransmissions}(\outputRV_{\blocklength_1 + \dots + \blocklength_{\indexAnalogTransmissions-1} + 1}', \dots, \outputRV_{\blocklength_1 + \dots + \blocklength_{\indexAnalogTransmissions}}').
\]
In order to prove the decoding error bound of $(\preproc_1', \dots, \preproc_\txNum', \postproc')$, we denote the additive noise added by the channel $\compoundChannel{\stateSpaceElement_1}^{\blocklength_1} \otimes \dots \otimes \compoundChannel{\stateSpaceElement_\numAnalogTransmissions}^{\blocklength_\numAnalogTransmissions}$ with ${\noiseRV'}^\blocklength$. We define a sequence ${\noiseRV}^{\blocklength-\numAnalogTransmissions}$ of real-valued random variables by
\[
\forall \indexAnalogTransmissions \in \naturalsInitialSection{\numAnalogTransmissions}~
 (\noiseRV_{\blocklength_1 + \dots + \blocklength_{\indexAnalogTransmissions-1} - \indexAnalogTransmissions + 2}, \dots, \noiseRV_{\blocklength_1 + \dots + \blocklength_{\indexAnalogTransmissions} - \indexAnalogTransmissions})
 :=
 \planeMapInv_{\blocklength_\indexAnalogTransmissions}(\noiseRV_{\blocklength_1 + \dots + \blocklength_{\indexAnalogTransmissions-1} + 1}', \dots, \noiseRV_{\blocklength_1 + \dots + \blocklength_{\indexAnalogTransmissions}}')
\]
and note that Lemma~\ref{lemma:orthogonal}-\ref{item:orthogonal-inv} and the fact that ${\noiseRV'}^\blocklength$ is distributed i.i.d. according to $\normaldistribution{0}{\noisestd^2}$ together imply that ${\noiseRV}^{\blocklength-\numAnalogTransmissions}$ is also distributed i.i.d. according to $\normaldistribution{0}{\noisestd^2}$. We thus obtain
\begin{align*}
&\hphantom{{}={}}
(\outputRV_{\blocklength_1 + \dots + \blocklength_{\indexAnalogTransmissions-1} - \indexAnalogTransmissions + 2}, \dots, \outputRV_{\blocklength_1 + \dots + \blocklength_{\indexAnalogTransmissions} - \indexAnalogTransmissions})
\\
\overset{(a)}&{=}
\sum_{\txIndex=1}^\txNum
  \planeMapInv_{\blocklength_\indexAnalogTransmissions}(\inputRV_{\txIndex, \blocklength_1 + \dots + \blocklength_{\indexAnalogTransmissions-1} + 1}', \dots, \inputRV_{\txIndex, \blocklength_1 + \dots + \blocklength_{\indexAnalogTransmissions}}')
+
\planeMapInv_{\blocklength_\indexAnalogTransmissions}(\stateSpaceElement_\indexAnalogTransmissions,\dots,\stateSpaceElement_\indexAnalogTransmissions)
+
\planeMapInv_{\blocklength_\indexAnalogTransmissions}(\noiseRV_{\blocklength_1 + \dots + \blocklength_{\indexAnalogTransmissions-1} + 1}', \dots, \noiseRV_{\blocklength_1 + \dots + \blocklength_{\indexAnalogTransmissions}}')
\\
\overset{(b)}&{=}
\sum_{\txIndex=1}^\txNum
  (\inputRV_{\txIndex, \blocklength_1 + \dots + \blocklength_{\indexAnalogTransmissions-1} - \indexAnalogTransmissions + 2}, \dots, \inputRV_{\txIndex, \blocklength_1 + \dots + \blocklength_{\indexAnalogTransmissions} - \indexAnalogTransmissions})
+
(\noiseRV_{\blocklength_1 + \dots + \blocklength_{\indexAnalogTransmissions-1} - \indexAnalogTransmissions + 2}, \dots, \noiseRV_{\blocklength_1 + \dots + \blocklength_{\indexAnalogTransmissions} - \indexAnalogTransmissions}),
\end{align*}
where (a) is due to the linearity of $\planeMapInv_{\blocklength_\indexAnalogTransmissions}$ and (b) uses Lemma~\ref{lemma:orthogonal}-\ref{item:orthogonal-inv} in the first summand and Lemma~\ref{lemma:orthogonal}-\ref{item:orthogonal-zero} in the second summand.

This means that we have defined a random experiment in which ${\outputRV'}^{\blocklength}$ is obtained by passing ${\inputRV'}^{\blocklength}$ through the channel $\compoundChannel{\stateSpaceElement_1}^{\blocklength_1} \otimes \dots \otimes \compoundChannel{\stateSpaceElement_\numAnalogTransmissions}^{\blocklength_\numAnalogTransmissions}$, and ${\outputRV}^{\blocklength-\numAnalogTransmissions}$ is obtained by passing ${\inputRV}^{\blocklength-\numAnalogTransmissions}$ through the channel $\channelAWGN^{\blocklength-\numAnalogTransmissions}$. With the definition $\postproc'({\outputRV'}^\blocklength):= \postproc(\outputRV^{\blocklength-\numAnalogTransmissions})$, the decoding error events of the codes coincide which means that $(\preproc_1, \dots, \preproc_\txNum, \postproc)$ inherits its average decoding error from $(\preproc_1', \dots, \preproc_\txNum', \postproc')$.
\end{proof}

\subsection{Achievable Rate Region for \gls{socc}}
\label{sec:random-hybrid}
We are now ready to prove the achievability part of Theorem~\ref{theorem:capacity}. For exact details on how this follows from the following lemma, see the end of this subsection.
\begin{lemma}
\label{lemma:digital-mac-with-ota-c}
Let $\powerconstraint_{\analogTxNum+1}, \dots, \powerconstraint_{\txNum}, \amplitudeConstraint_{\analogTxNum+1}, \dots, \amplitudeConstraint_{\txNum}, \analogAmplitudeConstraint \in (0,\infty)$, $\analograte, \analograte' \in (0,1)$ with $\analograte < \analograte'$, $\noisestd \in (0,\infty)$ and $\analogTxNum, \digitalTxNum \in \naturals$ be fixed, let $\blocklength \in \naturals$ be the block length and assume that $\analograte\blocklength \leq \numAnalogTransmissions \leq \analograte'\blocklength$. Let $\blocklength_1, \dots, \blocklength_\numAnalogTransmissions$ be natural numbers with $\blocklength = \blocklength_1 + \dots + \blocklength_\numAnalogTransmissions$. Let $\codebookRate = (\codebookRate_1, \dots, \codebookRate_{\digitalTxNum})$ be such that
\[
\codebookRate' = (\codebookRate_1', \dots, \codebookRate_{\digitalTxNum}') := (\codebookRate_1/(1-\analograte'), \dots, \codebookRate_{\digitalTxNum}/(1-\analograte'))
\]
is an inner point of
\[
\constrainedMacCapacity{\channelAWGN}\left(
  \frac{\powerconstraint_{\analogTxNum+1}}{1-\analograte},
  \dots,
  \frac{\powerconstraint_{\txNum}}{1-\analograte},
  \frac{\amplitudeConstraint_{\analogTxNum+1}}{\frac{\sqrt{2}}{\sqrt{2}-1}},
  \dots,
  \frac{\amplitudeConstraint_{\txNum}}{\frac{\sqrt{2}}{\sqrt{2}-1}}
\right).
\]

For large enough block lengths $\blocklength$, there is a \gls{socc} scheme for the channel $\channelAWGN$ with $\stateSpace_{\txIndex,\indexAnalogTransmissions} = [-1,1]$ for all $\txIndex \in \naturalsInitialSection{\analogTxNum}, \indexAnalogTransmissions \in \naturalsInitialSection{\numAnalogTransmissions}$; $\codebookSize_1 \geq \exp(\blocklength\codebookRate_1), \dots, \codebookSize_{\digitalTxNum} \geq \exp(\blocklength\codebookRate_{\digitalTxNum})$ such that:
\begin{enumerate}
    \item \label{item:digital-mac-with-ota-c-analog-powerconstraint} The analog transmission amplitude constraint \eqref{eq:analog-amplitude-constraint} holds.
    \item \label{item:digital-mac-with-ota-c-digital-powerconstraint} The digital amplitude constraint \eqref{eq:digital-amplitude-constraint} and power constraint \eqref{eq:digital-power-constraint} hold.
    \item \label{item:digital-mac-with-ota-c-error-digital} $\decodingError$ defined in \eqref{eq:reconstruction-error} tends to $0$ uniformly for all possible realizations of $\stateSpaceElement_{1,1}, \dots, \stateSpaceElement_{1,\numAnalogTransmissions}$, $\dots$, $\stateSpaceElement_{\analogTxNum,1}, \dots, \stateSpaceElement_{\analogTxNum,\numAnalogTransmissions}$ as $\blocklength \rightarrow \infty$.
    \item \label{item:digital-mac-with-ota-c-error-analog} For every $\indexAnalogTransmissions \in \naturalsInitialSection{\numAnalogTransmissions}$, conditioned on every possible realization of $\stateSpaceElement_{1,1}, \dots, \stateSpaceElement_{1,\numAnalogTransmissions}$, $\dots$, $\stateSpaceElement_{\analogTxNum,1}, \dots, \stateSpaceElement_{\analogTxNum,\numAnalogTransmissions}$, $\messageRealization_1, \dots, \messageRealization_{\digitalTxNum}$, each $\hat{\objectiveFunction}_\indexAnalogTransmissions$ is a Gaussian approximation of $\stateSpaceElement_{1,\indexAnalogTransmissions} + \dots + \stateSpaceElement_{\txNum,\indexAnalogTransmissions}$ with variance $\noisestd^2/\blocklength_\indexAnalogTransmissions\analogAmplitudeConstraint^2$.
\end{enumerate}
\end{lemma}
\begin{proof}
Fix $\analograte'' > \analograte'$ such that $\codebookRate'' = (\codebookRate_1'', \dots, \codebookRate_{\digitalTxNum}'') := (\codebookRate_1/(1-\analograte''), \dots, \codebookRate_{\digitalTxNum}/(1-\analograte''))$ is an inner point of $\constrainedMacCapacity{\channelAWGN}(\powerconstraint_{\analogTxNum+1}/(1-\analograte), \dots, \powerconstraint_{\txNum}/(1-\analograte), \amplitudeConstraint_{\analogTxNum+1}/(\sqrt{2}/(\sqrt{2}-1)), \dots, \amplitudeConstraint_{\txNum}/(\sqrt{2}/(\sqrt{2}-1))$. For every large enough $\blocklength$ and $\numAnalogTransmissions$, we can by definition of this capacity region fix a code (consisting of encoders $\preproc_1, \dots, \preproc_{\digitalTxNum}$ and a decoder $\postproc$) for rate $\codebookRate''$ and block length $\blocklength-\numAnalogTransmissions$ which satisfies the average power constraints $\powerconstraint_{\analogTxNum+1}/(1-\analograte), \dots, \powerconstraint_\txNum/(1-\analograte)$ and the amplitude constraints $\amplitudeConstraint_{\analogTxNum+1}/(\sqrt{2}/(\sqrt{2}-1)), \dots, \amplitudeConstraint_{\txNum}/(\sqrt{2}/(\sqrt{2}-1))$. Furthermore, the codes can be chosen in such a way that the decoding error vanishes for $\blocklength \rightarrow \infty$. For each $\txIndex$, the codebook size satisfies
\[
\codebookSize_\txIndex
\geq
\exp((\blocklength-\numAnalogTransmissions)\codebookRate_\txIndex'')
\geq
\exp(\blocklength(1-\analograte')\codebookRate_\txIndex/(1-\analograte''))
>
\exp(\blocklength\codebookRate_\txIndex).
\]

We next invoke Lemma~\ref{lemma:zerosum} to obtain modified encoding and decoding procedures $\preproc_1', \dots, \preproc_\txNum', \postproc'$ for block length $\blocklength$. For $\txIndex \in \naturalsInitialSection{\analogTxNum}$, we define
\[
\preprocHybrid_\txIndex(\stateSpaceElement_{\txIndex,1}, \dots, \stateSpaceElement_{\txIndex,\numAnalogTransmissions})
:=
\analogAmplitudeConstraint
(\stateSpaceElement_{\txIndex,1} \onevector_{\blocklength_1},
 \dots,
 \stateSpaceElement_{\txIndex,\numAnalogTransmissions} \onevector_{\blocklength_\numAnalogTransmissions}),
\]
where $\onevector_{\blocklength_\indexAnalogTransmissions}$ denotes the all-ones vector of length $\blocklength_\indexAnalogTransmissions$. For $\txIndex \in \naturalsInitialSection{\digitalTxNum}$, we define
\[
\preprocHybrid_{\analogTxNum+\txIndex}(\messageRealization_\txIndex)
:=
\preproc_\txIndex'(\messageRealization_\txIndex).
\]
Item \ref{item:digital-mac-with-ota-c-analog-powerconstraint} is satisfied since each $\stateSpaceElement_{\txIndex,\indexAnalogTransmissions} \in [-1,1]$. For item \ref{item:digital-mac-with-ota-c-digital-powerconstraint}, we note that the amplitude constraint is clearly satisfied due to Lemma~\ref{lemma:zerosum}-\ref{item:zerosum-peakpower} and the amplitude constraint of $\preproc_1, \dots, \preproc_{\digitalTxNum}$. For the average power constraint, we note that the total power of each digital transmitter's pre-processor output is bounded as
\[
\euclidNorm{\preprocHybrid_{\analogTxNum+\txIndex}(\messageRealization_\txIndex)}^2
=
\euclidNorm{\preproc_\txIndex'(\messageRealization_\txIndex)}^2
\overset{(a)}{=}
\euclidNorm{\preproc_\txIndex(\messageRealization_\txIndex)}^2
\overset{(b)}{\leq}
(\blocklength-\numAnalogTransmissions)\frac{\powerconstraint_{\analogTxNum+\txIndex}}{1-\analograte}
\overset{(c)}{\leq}
\blocklength(1-\analograte)\frac{\powerconstraint_{\analogTxNum+\txIndex}}{1-\analograte}
=
\blocklength\powerconstraint_{\analogTxNum+\txIndex},
\]
where (a) is due to Lemma~\ref{lemma:zerosum}-\ref{item:zerosum-power}, (b) is due to the choice of $\preproc_\txIndex(\messageRealization_\txIndex)$, and (c) is due to $\numAnalogTransmissions \geq \analograte\blocklength$.

Next, we define the digital part of $\postprocHybrid$ as $\hat{\messageRealization} := \postproc'(\outputRV^\blocklength)$. We observe that $\outputRV^\blocklength$ is actually $\preproc_1'(\messageRealization_1), \dots, \preproc_\txNum'(\messageRealization_\txNum)$ passed through the channel
$
\compoundChannel{\stateSpaceElement_1}^{\blocklength_1}
  \otimes
  \dots
  \otimes
  \compoundChannel{\stateSpaceElement_\numAnalogTransmissions}^{\blocklength_\numAnalogTransmissions}
$
with $\stateSpaceElement_\indexAnalogTransmissions := \analogAmplitudeConstraint(\stateSpaceElement_{1,\indexAnalogTransmissions} + \dots + \stateSpaceElement_{\txNum,\indexAnalogTransmissions})$. Therefore, item \ref{item:digital-mac-with-ota-c-error-digital} follows immediately from the error guarantee of $\preproc_1, \dots, \preproc_\txNum, \postproc$ in conjunction with Lemma~\ref{lemma:zerosum}.

Finally, we define the analog part of $\postprocHybrid$ and calculate its distribution as
\begin{align*}
\hat{\objectiveFunction}_\indexAnalogTransmissions
:&=
\frac{1}{\blocklength_\indexAnalogTransmissions\analogAmplitudeConstraint}
(\outputRV_{\blocklength_1 + \dots + \blocklength_{\indexAnalogTransmissions-1} + 1} + \dots + \outputRV_{\blocklength_1 + \dots + \blocklength_\indexAnalogTransmissions})
\\
\overset{(a)}&{=}
\frac{1}{\blocklength_\indexAnalogTransmissions\analogAmplitudeConstraint}
\left(
  \blocklength_\indexAnalogTransmissions\analogAmplitudeConstraint(\stateSpaceElement_{1,\indexAnalogTransmissions} + \dots + \stateSpaceElement_{\txNum,\indexAnalogTransmissions}) + \noiseRV_{\blocklength_1 + \dots + \blocklength_{\indexAnalogTransmissions-1} + 1} + \dots + \noiseRV_{\blocklength_1 + \dots + \blocklength_\indexAnalogTransmissions}
\right)
\\
&=
\stateSpaceElement_{1,\indexAnalogTransmissions} + \dots + \stateSpaceElement_{\txNum,\indexAnalogTransmissions}
+
\frac{1}{\blocklength_\indexAnalogTransmissions\analogAmplitudeConstraint}
\left(
  \noiseRV_{\blocklength_1 + \dots + \blocklength_{\indexAnalogTransmissions-1}}
  +
  \dots
  +
  \noiseRV_{\blocklength_1 + \dots + \blocklength_\indexAnalogTransmissions}
\right),
\end{align*}
where step (a) is due to the definition of $\preprocHybrid_1, \dots, \preprocHybrid_{\analogTxNum}$ and Lemma~\ref{lemma:zerosum}-\ref{item:zerosum}. Item \ref{item:digital-mac-with-ota-c-error-analog} immediately follows, completing the proof of the theorem.
\end{proof}
\begin{remark}
\label{remark:zero-rate}
\emph{(Zero-rate analog transmissions).}
If $\numAnalogTransmissions$ grows sub-linearly in $\blocklength$, i.e., for any $\analograte' > 0$ and sufficiently large $\blocklength$, we have $\numAnalogTransmissions < \analograte' \blocklength$, then Lemma~\ref{lemma:digital-mac-with-ota-c} guarantees that every inner point of
\[
\constrainedMacCapacity{\channelAWGN}\left(\powerconstraint_{\analogTxNum+1}, \dots, \powerconstraint_{\analogTxNum}, \frac{\amplitudeConstraint_{\analogTxNum+1}}{\frac{\sqrt{2}}{\sqrt{2}-1}}, \dots, \frac{\amplitudeConstraint_{\txNum}}{\frac{\sqrt{2}}{\sqrt{2}-1}}\right)
\]
can be achieved for the digital transmissions, hence the analog computations can be carried out at (almost) no cost in terms of digital rates, at least when the digital amplitude constraints are reasonably generous.

One particularly relevant special case of sub-linear growth is when $\numAnalogTransmissions$ is a constant that does not grow with $\blocklength$ at all. This might be an accurate model in a communication system where analog computation occurs very infrequently.
\end{remark}

\begin{proof}[Proof of the achievability part of Theorem~\ref{theorem:capacity}]
We invoke Lemma~\ref{lemma:digital-mac-with-ota-c} with $\analograte$ and $\analograte'$ from the Theorem~\ref{theorem:capacity} statement. Since $\analograte < \analograte'$, we can, for large enough $\blocklength$, pick $\numAnalogTransmissions$ in such a manner that $\analograte\blocklength \leq \numAnalogTransmissions \leq \analograte'\blocklength$. We choose $\blocklength_1, \dots, \blocklength_{\numAnalogTransmissions-1} := 1/\analograte'$ (which is an integer by definition) and $\blocklength_\numAnalogTransmissions := \blocklength - (\numAnalogTransmissions-1)/\analograte' \geq \blocklength - (\analograte'\blocklength - 1)/\analograte'=1/\analograte'$. The sequence of schemes constructed with these choices clearly satisfies items \ref{item:hads-achievable-digital-rates}, \ref{item:hads-achievable-analog-rate}, \ref{item:hads-achievable-constraints}, and \ref{item:hads-achievable-digital-error} of Definition~\ref{def:hads-achievable}. For item \ref{item:hads-achievable-analog-error}, we note that by Lemma~\ref{lemma:digital-mac-with-ota-c}-\ref{item:digital-mac-with-ota-c-error-analog}, for every $\indexAnalogTransmissions$, $\objectiveFunctionEstimate_\indexAnalogTransmissions$ is a Gaussian approximation of the sum function with variance
\[
\converseAnalogErrorVariance
=
\frac{\noisestd^2}{\blocklength_\indexAnalogTransmissions \analogAmplitudeConstraint^2}
\leq
\frac{\analograte' \noisestd^2}{\analogAmplitudeConstraint^2},
\]
from which the \gls{mse} approximation follows by Lemma~\ref{lemma:analog-approximation-implications}-\ref{item:analog-approximation-implication-gauss-mse}.
\end{proof}

\subsection{Simple Numerical Bounds for the Amplitude and Power Constrained Gaussian \gls{mac}}
\label{sec:gaussian-mac}
In this section, we discuss simple inner and outer bounds for the rate region $\constrainedMacCapacity{\channelAWGN}$ that appears in the statements of Theorem~\ref{theorem:capacity} and Lemma~\ref{lemma:digital-mac-with-ota-c}. In~\cite{mamandipoor2014capacity}, the authors propose a method to determine the capacity region of a two-user Gaussian \gls{mac} under amplitude constraints only which is readily extensible to the case of a larger number of transmitters. The result that optimizing input distributions are discrete with a finite number of mass points also holds if there is an additional average power constraint. However, using this fact to determine $\constrainedMacCapacity{\channelAWGN}$ numerically could incur high computational complexity for a large number of transmitters. Moreover, if the amplitude constraint is somewhat generous compared to the average power constraint (which is the regime we focus on in this work), the true boundary of $\constrainedMacCapacity{\channelAWGN}$ is not far away from much easier to obtain bounds, as was observed for the single-user case in~\cite{smith1971information}. In the following, we propose very simple such bounds based on the single-user capacities determined in~\cite{smith1971information} and discuss their tightness in a specific numerical example.

A code of block length $\blocklength$ for a \gls{mac} $\channel$ consists of encoders
\[
\codebook^\blocklength_1: \naturalsInitialSection{\codebookSize_1} \rightarrow \reals^\blocklength,
\dots,
\codebook^\blocklength_\txNum: \naturalsInitialSection{\codebookSize_\txNum} \rightarrow \reals^\blocklength,
\]
and a decoder
\[
\decoder^\blocklength:
\reals^\blocklength
\rightarrow
\naturalsInitialSection{\codebookSize_1} \times \dots \times \naturalsInitialSection{\codebookSize_\txNum}.
\]
The average decoding error over the channel $\channel$ is defined as
\[
\decodingError^\channel(\codebook_1^\blocklength, \dots, \codebook^\blocklength_\txNum, \decoder^\blocklength)
:=
\frac{1}{\codebookSize_1 \cdots \codebookSize_\txNum}
\sum_{\messageRealization_1=1}^{\codebookSize_1}
\dots
\sum_{\messageRealization_\txNum=1}^{\codebookSize_\txNum}
\Probability\left(
  \decoder(\outputRV^\blocklength)
  \neq
  (\messageRealization_1, \dots, \messageRealization_\txNum)
  ~|~
  \inputRV^\blocklength_1 = \codebook_1(\messageRealization_1),
  \dots,
  \inputRV^\blocklength_\txNum = \codebook_\txNum(\messageRealization_\txNum)
\right)
\]
A code has a rate associated with each transmitter $\txIndex \in \naturalsInitialSection{\txNum}$ defined by
\[
\codebookRate_\txIndex
:=
\frac{\log \codebookSize_\txIndex}{\blocklength}.
\]
A rate tuple $(\codebookRate_1, \dots, \codebookRate_\txNum)$ is called \emph{achievable} for $\channel$ under average power constraints $\powerconstraint_1, \dots, \powerconstraint_\txNum$ if for every $\blocklength \in \naturals$, there is a code $\codebook_1^\blocklength, \dots, \codebook^\blocklength_\txNum, \decoder^\blocklength$ of rates at least $(\codebookRate_1, \dots, \codebookRate_\txNum)$ such that
\begin{equation}
\label{eq:achievable}
\lim_{\blocklength \rightarrow \infty}
\decodingError^\channel(\codebook_1^\blocklength, \dots, \codebook^\blocklength_\txNum, \decoder^\blocklength)
=
0
,~~
\forall \blocklength \in \naturals~
\forall \txIndex \in \naturalsInitialSection{\txNum}~
\forall \messageRealization \in \naturalsInitialSection{\codebookSize_\txIndex}~
\frac{1}{\blocklength} \euclidNorm{\codebook^\blocklength_\txIndex(\messageRealization)}^2
\leq
\powerconstraint_\txIndex.
\end{equation}
The capacity region of $\channel$ under average power constraint is defined as
\[
\gls{unconstrainedMacCapacity}
=
\closure\left\{
  (\codebookRate_1, \dots, \codebookRate_\txNum)
  :~
  (\codebookRate_1, \dots, \codebookRate_\txNum)
  \text{ is achievable for $\channel$ under average power constraints }
  \powerconstraint_1, \dots, \powerconstraint_\txNum
\right\}.
\]

Similarly, a rate tuple $(\codebookRate_1, \dots, \codebookRate_\txNum)$ is called \emph{achievable} under average power constraints $\powerconstraint_1, \dots, \powerconstraint_\txNum$ and amplitude constraints $\amplitudeConstraint_1, \dots, \amplitudeConstraint_\txNum$ if for every $\decodingError > 0$ there is a block length $\blocklength$ and a code $\codebook_1^\blocklength, \dots, \codebook^\blocklength_\txNum, \decoder^\blocklength$ of rates $(\codebookRate_1, \dots, \codebookRate_\txNum)$ which satisfies (\ref{eq:achievable}) and
\[
\forall \txIndex \in \naturalsInitialSection{\txNum}~
\forall \messageRealization \in \naturalsInitialSection{\codebookSize_\txIndex}~
\maximumNorm{\codebook^\blocklength_\txIndex(\messageRealization)}
\leq
\amplitudeConstraint_\txIndex.
\]

We recall definition (\ref{eq:constrained-mac-capacity}) of the capacity region of $\channel$ under an average power constraint and an amplitude constraint:
\begin{multline*}
\constrainedMacCapacity{\channel}(
  \powerconstraint_1, \dots, \powerconstraint_\txNum,
  \amplitudeConstraint_1, \dots, \amplitudeConstraint_\txNum
)
=
\closure\big\{
  (\codebookRate_1, \dots, \codebookRate_\txNum)
  :~
  (\codebookRate_1, \dots, \codebookRate_\txNum)
  \text{ is achievable for $\channel$}
  \\
  \text{under average power constraints }
  \powerconstraint_1, \dots, \powerconstraint_\txNum
  \text{ and amplitude constraints }
  \amplitudeConstraint_1, \dots, \amplitudeConstraint_\txNum
\big\}.
\end{multline*}

It is known~\cite[eq. (15.152), (15.153)]{cover2006elements} that under average power constraints the capacity region of $\channelAWGN$ is
\begin{equation}
\label{eq:unconstrained-capacity}
\unconstrainedMacCapacity{\channelAWGN}(\powerconstraint_1, \dots, \powerconstraint_\txNum)
=
\left\{
  (\codebookRate_1, \dots, \codebookRate_\txNum)
  :~
  \forall \txSubset \subseteq \naturalsInitialSection{\txNum}~
  \sum_{\txIndex \in \txSubset} \codebookRate_\txIndex
  \leq
  \unconstrainedGaussianCapacity\left(
    \frac{
      \sum_{\txIndex \in \txSubset} \powerconstraint_\txIndex
    }{
      \noisestd^2
    }
  \right)
\right\}.
\end{equation}
For the average power and amplitude constrained capacity region, a straightforward extension of~\cite[Theorem 8]{han2006information} to the case of more than two transmitters tells us that for any channel $\channel$,
$
\constrainedMacCapacity{\channel}(
  \powerconstraint_1, \dots, \powerconstraint_\txNum,
  \amplitudeConstraint_1, \dots, \amplitudeConstraint_\txNum
)
$
is the convex closure of
\begin{align*}
\Bigg\{
  (\codebookRate_1, \dots, \codebookRate_\txNum)
  :~
  \exists \inputDistribution_1, \dots, \inputDistribution_\txNum~
  &\forall \txSubset \subseteq \naturalsInitialSection{\txNum}~
  \sum_{\txIndex \in \txSubset} \codebookRate_\txIndex
  \leq
  \mutualInformation{\inputDistribution_1, \dots, \inputDistribution_\txNum}{\channel}
    ((\inputRV_\txIndex)_{\txIndex \in \txSubset}; \outputRV | (\inputRV_\txIndex)_{\txIndex \in \naturalsInitialSection{\txNum} \setminus \txSubset})
  ,
  \\
  &\forall \txIndex \in \naturalsInitialSection{\txNum}~
  \Expectation_{\inputDistribution_\txIndex} \big((\inputRV_\txIndex)^2\big) \leq \powerconstraint_\txIndex,
  \inputDistribution_\txIndex(\absolutevalue{\inputRV} > \amplitudeConstraint_\txIndex) = 0
\Bigg\},
\end{align*}
where \gls{mutualInformation} denotes mutual information between inputs and output of the channel $\channel$ under the input distributions \gls{inputDistributions}. For $\txNum=1$, it is shown in \cite{smith1971information} that there is a unique optimal input distribution $\optimalInputDistribution(\powerconstraint,\amplitudeConstraint)$ with second moment at most $\powerconstraint$ and amplitude at most $\amplitudeConstraint$ such that
\[
\constrainedMacCapacity{\channelAWGN}(\powerconstraint,\amplitudeConstraint)
=
\left[0, \mutualInformation{\optimalInputDistribution(\powerconstraint,\amplitudeConstraint)}{\channel_{(1)}}(\inputRV; \outputRV)\right],
\]
where $\channel_{(1)}$ denotes the channel $\channelAWGN$ with $\txNum=1$.

We also define rate regions
\begin{align*}
&\begin{multlined}
\constrainedMacCapacityAchievable{\channel}(\powerconstraint_1, \dots, \powerconstraint_\txNum, \amplitudeConstraint_1, \dots, \amplitudeConstraint_\txNum)
:=\\
\left\{
  (\codebookRate_1, \dots, \codebookRate_\txNum)
  :~
  \forall \txSubset \subseteq \naturalsInitialSection{\txNum}~
  \sum_{\txIndex \in \txSubset} \codebookRate_\txIndex
  \leq
  \mutualInformation{
    \optimalInputDistribution(\powerconstraint_1, \amplitudeConstraint_1),
    \dots,
    \optimalInputDistribution(\powerconstraint_\txNum, \amplitudeConstraint_\txNum)
  }{\channel}
  \left(
      \sum_{\txIndex \in \txSubset} \inputRV_\txIndex;
      \outputRV
      |
      (\inputRV_\txIndex)_{\txIndex \in \naturalsInitialSection{\txNum} \setminus \txSubset}
  \right)
\right\}
\end{multlined}
\\
&\begin{multlined}
\constrainedMacCapacityConverse{\channel}(\powerconstraint_1, \dots, \powerconstraint_\txNum, \amplitudeConstraint_1, \dots, \amplitudeConstraint_\txNum)
:=\\
\left\{
  (\codebookRate_1, \dots, \codebookRate_\txNum)
  :~
  \forall \txIndex \in \naturalsInitialSection{\txNum}~
  \codebookRate_\txIndex
  \leq
  \mutualInformation{
    \optimalInputDistribution(\powerconstraint_\txIndex, \amplitudeConstraint_\txIndex)
  }{\channel_{(1)}}
  \left(
      \inputRV;
      \outputRV
  \right)
\right\}
\end{multlined}
\end{align*}
It is clear from the definition of $\constrainedMacCapacity{\channel}(\powerconstraint_1, \dots, \powerconstraint_\txNum, \amplitudeConstraint_1, \dots, \amplitudeConstraint_\txNum)$ that for every inner point $(\codebookRate_1, \dots, \codebookRate_\txNum)$ of the region $\constrainedMacCapacity{\channelAWGN}(\powerconstraint_1, \dots, \powerconstraint_\txNum, \amplitudeConstraint_1, \dots, \amplitudeConstraint_\txNum)$ and every $\txIndex \in \naturalsInitialSection{\txNum}$, there are probability distributions $\inputDistribution_1, \dots, \inputDistribution_\txNum$ such that
\begin{align*}
\codebookRate_\txIndex
&\leq
\mutualInformation{
  \inputDistribution_1, \dots, \inputDistribution_\txNum
}{
  \channelAWGN
}(
  \inputRV_\txIndex; \outputRV
  |
  \inputRV_1, \dots, \inputRV_{\txIndex-1},
  \inputRV_{\txIndex+1}, \dots, \inputRV_\txNum
)
\\
\overset{(a)}&{=}
\mutualInformation{
  \inputDistribution_1, \dots, \inputDistribution_\txNum
}{
  \channelAWGN
}(
  \inputRV_\txIndex;
  \outputRV
  -
  \inputRV_1 - \dots - \inputRV_{\txIndex-1} -
  \inputRV_{\txIndex+1} - \dots - \inputRV_\txNum
)
\\
&=
\mutualInformation{
  \inputDistribution_\txIndex
}{
  \channel_{(1)}
}(
  \inputRV;
  \outputRV
)
\\
&\leq
\mutualInformation{
  \optimalInputDistribution(\powerconstraint_\txIndex, \amplitudeConstraint_\txIndex)
}{
  \channel_{(1)}
}(
  \inputRV;
  \outputRV
),
\end{align*}
where (a) is due to the additivity of $\channelAWGN$. This yields
\begin{multline}
\label{eq:constrained-capacity-region-bounds}
\constrainedMacCapacityAchievable{\channelAWGN}(\powerconstraint_1, \dots, \powerconstraint_\txNum, \amplitudeConstraint_1, \dots, \amplitudeConstraint_\txNum)
\subseteq
\constrainedMacCapacity{\channelAWGN}(
  \powerconstraint_1, \dots, \powerconstraint_\txNum,
  \amplitudeConstraint_1, \dots, \amplitudeConstraint_\txNum
)
\\
\subseteq
\unconstrainedMacCapacity{\channelAWGN}(\powerconstraint_1, \dots, \powerconstraint_\txNum)
\cap
\constrainedMacCapacityConverse{\channelAWGN}(\powerconstraint_1, \dots, \powerconstraint_\txNum, \amplitudeConstraint_1, \dots, \amplitudeConstraint_\txNum).
\end{multline}

With the representation (\ref{eq:unconstrained-capacity}), it is easy to determine $\unconstrainedMacCapacity{\channelAWGN}(\powerconstraint_1, \dots, \powerconstraint_\txNum)$ for any given set of parameters, and with the methods from~\cite{smith1971information}, it is possible to numerically determine $\constrainedMacCapacityAchievable{\channelAWGN}(\powerconstraint_1, \dots, \powerconstraint_\txNum, \amplitudeConstraint_1, \dots, \amplitudeConstraint_\txNum)$ and $\constrainedMacCapacityConverse{\channelAWGN}(\powerconstraint_1, \dots, \powerconstraint_\txNum, \amplitudeConstraint_1, \dots, \amplitudeConstraint_\txNum)$.

\begin{figure}
\centering
\begin{tikzpicture}[scale=7,spy using outlines={circle, magnification=10, size=3cm, connect spies}]
\drawaxes{.7}{.7}
\drawxtick{.5}{$0.5$}
\drawytick{.5}{$0.5$}

%Unconstrained Gaussian region
\draw[densely dotted] (.4068,0) -- (.4068,.3735) -- (.4040,.3763);
\draw (.4040,.3763) -- (.1653,.6150);
\draw[densely dotted] (.1653,.6150) -- (.1531,.6272) -- (0,.6272) node[at end, above right]{$\unconstrainedMacCapacity{\channelAWGN}(\powerconstraint_1,\powerconstraint_2)$};

%Individual rate constraints amplitude constrained
\draw[red] (.4040,0) -- (.4040,.3725);
\draw[red,densely dotted] (.4040,.3725) -- (.4040,.6150) node[at end, right] {$\constrainedMacCapacityConverse{\channelAWGN}(\powerconstraint_1,\powerconstraint_2,\amplitudeConstraint_1,\amplitudeConstraint_2)$} -- (.1615,.6150);
\draw[red] (.1615,.6150) -- (0,.6150);

%Achievable sumrate bound amplitude constrained
\draw[blue] (.4040,.3725) -- (.1615,.6150) node[at start,left, align = center]{$\constrainedMacCapacityAchievable{\channelAWGN}$\\$(\powerconstraint_1,\powerconstraint_2,\amplitudeConstraint_1,\amplitudeConstraint_2)$};

%Magnify parts
\spy[size=2cm] on (2.828,2.6341) in node[right] at (.5,.3763);
\end{tikzpicture}
\caption{Example of the amplitude-constrained rate region for $\txNum=2, \powerconstraint_1 = \qty{1}{\dB}, \powerconstraint_2 = \qty{4}{\dB}, \noisestd^2 = 1, \amplitudeConstraint_1 = 2\sqrt{\powerconstraint_1}, \amplitudeConstraint_2 = 2\sqrt{\powerconstraint_2}$. The solid lines outline $\constrainedMacCapacity{\channelAWGN}(\powerconstraint_1,\powerconstraint_2,\amplitudeConstraint_1,\amplitudeConstraint_2)$, where in the region of the sum rate constraint, the boundary of the capacity region is somewhere between the two solid lines.}
\label{fig:constrained-capacity-region-example}
\end{figure}

It has been noted in~\cite{smith1971information} that even for moderately generous amplitude constraints, the constrained and unconstrained capacity bounds are very close in the single-user case. Our numerical evaluations suggest that this phenomenon is even more pronounced for sum rate constraints, making the inner and outer bounds in (\ref{eq:constrained-capacity-region-bounds}) quite tight in practice. In Fig.~\ref{fig:constrained-capacity-region-example}, we show a numerical example for the case of two users.

\section{Converse}
\label{sec:converse}
In this section, we prove Lemma~\ref{lemma:converse} which implies the converse part (second inclusion) of Theorem~\ref{theorem:capacity} via Corollary~\ref{cor:converse} proven at the end of the section. To make the notation in this section more compact, we use $\analogTxSet := \naturalsInitialSection{\analogTxNum}$ to denote the set of analog transmitters and $\digitalTxSet := \{\analogTxNum+1, \dots,\txNum\}$ to denote the set of digital transmitters. If $\generalRV^\generalNatural = (\generalRV_1, \dots, \generalRV_\generalNatural)$ is a tuple of random variables and $\generalSubset \subseteq \naturalsInitialSection{\generalNatural}$, we use the shorthand $\generalRV_{\generalSubset} := (\generalRV_\generalNatural)_{\generalNatural \in \generalSubset}$ to denote the sub-tuple indexed by $\generalSubset$ and we use $\sum \generalRV_{\generalSubset} := \sum_{\generalNatural \in \generalSubset} \generalRV_\generalNatural$ to denote the sum of all elements of the tuple $\generalRV_{\generalSubset}$. Moreover, whenever the set $\naturalsInitialSection{\generalNatural}$ is clear from context, we use $\complement{\generalSubset} := \naturalsInitialSection{\generalNatural} \setminus \generalSubset$ to denote the complement of $\generalSubset$. In the system model we use, the digital transmitters are numbered starting from $\analogTxNum+1$ while some quantities related to the digital transmitters are numbered starting from $1$. For simplicity, we use notations such as $\messageRV_{\analogTxNum+\txIndex}$ and $\messageRV_\txIndex$ or $\codebookRate_{\analogTxNum+\txIndex}$ and $\codebookRate_\txIndex$ interchangeably for $\txIndex \in \naturalsInitialSection{\digitalTxNum}$ wherever the meaning is sufficiently clear from context.

\begin{lemma}
\label{lemma:converse}
Let $\powerconstraint_{\analogTxNum+1}, \dots, \powerconstraint_{\txNum}, \amplitudeConstraint_{\analogTxNum+1}, \dots, \amplitudeConstraint_{\txNum}, \analogAmplitudeConstraint, \codebookRate_1, \dots, \codebookRate_{\digitalTxNum}, \converseAnalogErrorVariance \in (0,\infty)$, $\analograte \in (0,1)$, $\noisestd \in (0,\infty)$ and $\analogTxNum, \digitalTxNum \in \naturals$ be fixed. Assume that for every $\blocklength \in \naturals$, there are $\numAnalogTransmissions \geq \analograte\blocklength, \codebookSize_1 \geq \exp(\blocklength\codebookRate_1), \dots, \codebookSize_{\digitalTxNum} \geq \exp(\blocklength\codebookRate_{\digitalTxNum})$, encoders $\preprocHybrid_1: [-1,1]^\numAnalogTransmissions \rightarrow \reals^\blocklength, \dots, \preprocHybrid_{\analogTxNum}: [-1,1]^\numAnalogTransmissions \rightarrow \reals^\blocklength, \preprocHybrid_{\analogTxNum+1}: \naturalsInitialSection{\codebookSize_1} \rightarrow \reals^\blocklength, \dots, \preprocHybrid_{\txNum}: \naturalsInitialSection{\codebookSize_{\digitalTxNum}} \rightarrow \reals^\blocklength$ and a decoder $\postprocHybrid: \reals^\blocklength \rightarrow \naturalsInitialSection{\codebookSize_1} \times \dots \times \naturalsInitialSection{\codebookSize_{\digitalTxNum}} \times \reals^\numAnalogTransmissions$ such that items \ref{item:digital-mac-with-ota-c-analog-powerconstraint}, \ref{item:digital-mac-with-ota-c-digital-powerconstraint}, and \ref{item:digital-mac-with-ota-c-error-digital} of Lemma~\ref{lemma:digital-mac-with-ota-c} are all satisfied, and for every $\indexAnalogTransmissions \in \naturalsInitialSection{\numAnalogTransmissions}$, $\objectiveFunctionEstimate_\indexAnalogTransmissions$ approximates $\objectiveFunction: (\stateSpaceElement_1, \dots, \stateSpaceElement_{\analogTxNum}) \mapsto \stateSpaceElement_1 + \dots + \stateSpaceElement_{\analogTxNum}$ with \gls{mse} $\converseAnalogErrorVariance$, i.e., we have, conditioned on any possible realization of $\stateSpaceElement_{1,1}, \dots, \stateSpaceElement_{1,\numAnalogTransmissions}$, $\dots$, $\stateSpaceElement_{\analogTxNum,1}, \dots, \stateSpaceElement_{\analogTxNum,\numAnalogTransmissions}$, $\messageRealization_1, \dots, \messageRealization_{\digitalTxNum}$, that
\begin{equation}
\label{eq:converse-mse}
\Expectation\left(
  \left(
    \stateSpaceElement_{1,\indexAnalogTransmissions} + \dots + \stateSpaceElement_{\txNum,\indexAnalogTransmissions}
    -
    \hat{\objectiveFunction}_\indexAnalogTransmissions
  \right)^2
\right)
\leq
\converseAnalogErrorVariance.
\end{equation}
Let $\analogTxSubset \subseteq \analogTxSet$. Then, there is a probability distribution $\inputDistribution_{\timeSharingRV, \inputRV_1, \dots, \inputRV_{\txNum}} = \inputDistribution_\timeSharingRV
\inputDistribution_{\inputRV_1, \dots, \inputRV_{\analogTxNum} | \timeSharingRV} \inputDistribution_{\inputRV_{\analogTxNum+1} | \timeSharingRV} \dots \inputDistribution_{\inputRV_{\txNum} | \timeSharingRV}$ for a discrete random variable $\timeSharingRV$ and the channel inputs $\inputRV_1, \dots, \inputRV_{\txNum}$ such that
\begin{align*}
&\Expectation\big((\inputRV_{\analogTxNum+1})^2\big) \leq \powerconstraint_{\analogTxNum+1}, \dots \Expectation\big((\inputRV_{\txNum})^2\big) \leq \powerconstraint_{\txNum},
\\
&\Probability(\absolutevalue{\inputRV_1} > \analogAmplitudeConstraint) = \dots = \Probability(\absolutevalue{\inputRV_{\analogTxNum}} > \analogAmplitudeConstraint)
= \Probability(\absolutevalue{\inputRV_{\analogTxNum+1}} > \amplitudeConstraint_{\analogTxNum+1}) = \dots = \Probability(\absolutevalue{\inputRV_{\txNum}} > \amplitudeConstraint_{\txNum})
= 0
\end{align*}
and for all $\digitalTxSubset \subseteq \digitalTxSet$, we have
\begin{equation}
\label{eq:converse-rate-bound-mse}
\sum\codebookRate_{\digitalTxSubset}
\leq
\mutualInformationRVCond{
  \inputRV_{\analogTxSubset},
  \inputRV_{\digitalTxSubset}
}{\outputRV}{
  \inputRV_{\analogTxSubsetComplement},
  \inputRV_{\digitalTxSubsetComplement},
  \timeSharingRV
}
-
\positivePart{
  \frac{\analograte}{2}\log\frac{2\cardinality{\analogTxSubset}^2}{\pi\eulersNumber\converseAnalogErrorVariance}
}.
\end{equation}
\end{lemma}

\begin{remark}
We note that in this converse, the probability distribution of the random variables in (\ref{eq:converse-rate-bound-mse}) depends on $\analogTxSubset$. However, this is sufficient here as in our proof of Corollary~\ref{cor:converse}, we need to invoke Lemma~\ref{lemma:converse} only for one specific choice of $\analogTxSubset$. For details, see the proof of Corollary~\ref{cor:converse}.
\end{remark}

\begin{proof}
We fix a block length $\blocklength$ and the associated encoders $\preprocHybrid_1, \dots, \preprocHybrid_{\txNum}$ and decoder $\postprocHybrid$ that exist according to the assumption of the theorem, as well as $\analogTxSubset \subseteq \analogTxSet$ and $\digitalTxSubset \subseteq \digitalTxSet$. For the digital messages $(\messageRV_1, \dots, \messageRV_{\digitalTxNum})$, we assume uniform distribution in $\naturalsInitialSection{\codebookSize_1} \times \dots \times \naturalsInitialSection{\codebookSize_{\digitalTxNum}}$ (which in particular means that they are independent). According to our assumptions, the error guarantees hold for every possible realization of the analog values $(\stateSpaceElement_{1,1}, \dots, \stateSpaceElement_{1,\numAnalogTransmissions}, \dots, \stateSpaceElement_{\analogTxNum,1}, \dots, \stateSpaceElement_{\analogTxNum,\numAnalogTransmissions})$. In particular, they hold under any possible random distribution of these values. We will treat the analog values as random variables that are independent of the digital messages, with a distribution to be specified later. For now, we only make the following restrictions on their probability distribution:
\begin{align}
\label{eq:analog-as-zero}
\forall \indexAnalogTransmissions \in \analogTransmissionsIndexSet~ \forall \txIndex \in \analogTxSubsetComplement:~
 \stateSpaceElement_{\txIndex, \indexAnalogTransmissions} = 0 \text{ almost surely,}
\\
\label{eq:analog-independence}
\sum \stateSpaceElement_{\analogTxSubset, 1},
\dots,
\sum \stateSpaceElement_{\analogTxSubset,\numAnalogTransmissions}
\text{ are independent.}
\end{align}
Let $\timeSharingRV$ be a random variable, uniformly distributed in $\naturalsInitialSection{\blocklength}$ and independent of all random variables defined thus far. Denote $\inputRV_\txIndex := \inputRV_{\txIndex, \timeSharingRV}$ and $\outputRV := \outputRV_\timeSharingRV$. We bound the conditional mutual information of the inputs and outputs from both sides. From above, we use the bound
\begin{align}
\nonumber
\mutualInformationRVCond{
  \inputRV_{\analogTxSubset}^\blocklength,
  \inputRV_{\digitalTxSubset}^\blocklength
}{\outputRV^\blocklength}{
  \inputRV_{\analogTxSubsetComplement}^\blocklength,
  \inputRV_{\digitalTxSubsetComplement}^\blocklength
}
\overset{(a)}&{=}
\sum_{\blIndex=1}^\blocklength
\mutualInformationRVCond{
  \inputRV_{\analogTxSubset}^\blocklength,
  \inputRV_{\digitalTxSubset}^\blocklength
}{\outputRV_\blIndex}{
  \inputRV_{\analogTxSubsetComplement}^\blocklength,
  \inputRV_{\digitalTxSubsetComplement}^\blocklength,
  \outputRV^{\blIndex-1}
}
\\
\nonumber
&\leq
\sum_{\blIndex=1}^\blocklength
\mutualInformationRVCond{
  \inputRV_{\analogTxSubset}^\blocklength,
  \inputRV_{\digitalTxSubset}^\blocklength,
  \inputRV_{\analogTxSubsetComplement,\complement{\{\blIndex\}}},
  \inputRV_{\digitalTxSubsetComplement,\complement{\{\blIndex\}}},
  \outputRV^{\blIndex-1}
}{\outputRV_\blIndex}{
  \inputRV_{\analogTxSubsetComplement}^\blocklength,
  \inputRV_{\digitalTxSubsetComplement}^\blocklength,
  \outputRV^{\blIndex-1}
}
\\
\nonumber
\overset{(b)}&{\leq}
\sum_{\blIndex=1}^\blocklength
\mutualInformationRVCond{
  \inputRV_{\analogTxSubset}^\blocklength,
  \inputRV_{\digitalTxSubset}^\blocklength,
  \inputRV_{\analogTxSubsetComplement,\complement{\{\blIndex\}}},
  \inputRV_{\digitalTxSubsetComplement,\complement{\{\blIndex\}}},
  \outputRV^{\blIndex-1}
}{\outputRV_\blIndex}{
  \inputRV_{\analogTxSubsetComplement,\blIndex},
  \inputRV_{\digitalTxSubsetComplement,\blIndex}
}
\\
\nonumber
\overset{(a)}&{=}
\sum_{\blIndex=1}^\blocklength
\mutualInformationRVCond{
  \inputRV_{\analogTxSubset,\blIndex},
  \inputRV_{\digitalTxSubset,\blIndex}
}{\outputRV_\blIndex}{
  \inputRV_{\analogTxSubsetComplement,\blIndex},
  \inputRV_{\digitalTxSubsetComplement,\blIndex}
}
+
\sum_{\blIndex=1}^\blocklength
\mutualInformationRVCond{
  \inputRV_{\analogTxSet,\complement{\{\blIndex\}}},
  \inputRV_{\digitalTxSet,\complement{\{\blIndex\}}},
  \outputRV^{\blIndex-1}
}{\outputRV_\blIndex}{
  \inputRV_{\analogTxSet,\blIndex},
  \inputRV_{\digitalTxSet,\blIndex}
}
\\
\nonumber
\overset{(c)}&{=}
\sum_{\blIndex=1}^\blocklength
\mutualInformationRVCond{
  \inputRV_{\analogTxSubset,\blIndex},
  \inputRV_{\digitalTxSubset,\blIndex}
}{\outputRV_\blIndex}{
  \inputRV_{\analogTxSubsetComplement,\blIndex},
  \inputRV_{\digitalTxSubsetComplement,\blIndex}
}
\\
\nonumber
\overset{(d)}&{=}
\blocklength
\sum_{\blIndex=1}^\blocklength
\frac{1}{\blocklength}
\mutualInformationRVCond{
  \inputRV_{\analogTxSubset},
  \inputRV_{\digitalTxSubset}
}{\outputRV}{
  \inputRV_{\analogTxSubsetComplement},
  \inputRV_{\digitalTxSubsetComplement},
  \timeSharingRV=\blIndex
}
\\
\label{eq:converse-bound-above}
&=
\blocklength
\mutualInformationRVCond{
  \inputRV_{\analogTxSubset},
  \inputRV_{\digitalTxSubset}
}{\outputRV}{
  \inputRV_{\analogTxSubsetComplement},
  \inputRV_{\digitalTxSubsetComplement},
  \timeSharingRV
},
\end{align}
where both steps labeled (a) are due to the chain rule for mutual information, (b) is because the removed conditions are conditionally independent of $\outputRV_\blIndex$ given the first argument of the information, (c) follows from the memoryless property of the channel, and (d) from the independence of $\timeSharingRV$ of the other random variables. For the following, we slightly extend our convention about indexing tuples with sets and use $\stateSpaceElement_{\analogTxSubset,\analogTransmissionsIndexSet} := (\stateSpaceElement_{\txIndex,\indexAnalogTransmissions})_{\txIndex\in\analogTxSubset,\indexAnalogTransmissions\in\analogTransmissionsIndexSet}$. From below, we use the bound
\begin{align}
\nonumber
\mutualInformationRVCond{
  \inputRV_{\analogTxSubset}^\blocklength,
  \inputRV_{\digitalTxSubset}^\blocklength
}{\outputRV^\blocklength}{
  \inputRV_{\analogTxSubsetComplement}^\blocklength,
  \inputRV_{\digitalTxSubsetComplement}^\blocklength
}
\overset{(a)}&{\geq}
\mutualInformationRVCond{
  \stateSpaceElement_{\analogTxSubset,\analogTransmissionsIndexSet},
  \messageRV_{\digitalTxSubset}
}{
  \objectiveFunctionEstimate_{\analogTransmissionsIndexSet},
  \hat{\messageRV}_{\digitalTxSubset}
}{
  \inputRV_{\analogTxSubsetComplement}^\blocklength,
  \inputRV_{\digitalTxSubsetComplement}^\blocklength
}
\\
\nonumber
\overset{(b)}&{\geq}
\mutualInformationRV{
  \stateSpaceElement_{\analogTxSubset,\analogTransmissionsIndexSet},
  \messageRV_{\digitalTxSubset}
}{
  \objectiveFunctionEstimate_{\analogTransmissionsIndexSet},
  \hat{\messageRV}_{\digitalTxSubset}
}
\\
\nonumber
\overset{(c)}&{=}
\mutualInformationRV{
  \messageRV_{\digitalTxSubset}
}{
  \objectiveFunctionEstimate_{\analogTransmissionsIndexSet},
  \hat{\messageRV}_{\digitalTxSubset}
}
+
\mutualInformationRVCond{
  \stateSpaceElement_{\analogTxSubset,\analogTransmissionsIndexSet}
}{
  \objectiveFunctionEstimate_{\analogTransmissionsIndexSet},
  \hat{\messageRV}_{\digitalTxSubset}
}{
  \messageRV_{\digitalTxSubset}
}
\\
\nonumber
&\geq
\mutualInformationRV{
  \messageRV_{\digitalTxSubset}
}{
  \hat{\messageRV}_{\digitalTxSubset}
}
+
\mutualInformationRVCond{
  \stateSpaceElement_{\analogTxSubset,\analogTransmissionsIndexSet}
}{
  \objectiveFunctionEstimate_{\analogTransmissionsIndexSet}
}{
  \messageRV_{\digitalTxSubset}
}
\\
\label{eq:converse-bound-below}
\overset{(d)}&{\geq}
\mutualInformationRV{
  \messageRV_{\digitalTxSubset}
}{
  \hat{\messageRV}_{\digitalTxSubset}
}
+
\mutualInformationRV{
  \stateSpaceElement_{\analogTxSubset,\analogTransmissionsIndexSet}
}{
  \objectiveFunctionEstimate_{\analogTransmissionsIndexSet}
},
\end{align}
where (a) is due to the data processing inequality. (b) holds because $\inputRV_{\analogTxSubsetComplement}^\blocklength$ is independent of $\stateSpaceElement_{\analogTxSubset,\analogTransmissionsIndexSet}, \messageRV_{\digitalTxSubset}$ due to (\ref{eq:analog-as-zero}), and $\inputRV_{\digitalTxSubsetComplement}^\blocklength$ is also independent of $\stateSpaceElement_{\analogTxSubset,\analogTransmissionsIndexSet}, \messageRV_{\digitalTxSubset}$ due to the independence assumption of digital messages. (c) is an application of the chain rule, and (d) holds because $\messageRV_{\digitalTxSubset}$ is independent of $\stateSpaceElement_{\analogTxSubset,\analogTransmissionsIndexSet}$. We next further bound these terms separately. By Fano's inequality (e.g.,~\cite[Section 2.1]{elgamal2011network}), we have
\begin{equation}
\label{eq:converse-bound-digital}
\mutualInformationRV{
  \messageRV_{\digitalTxSubset}
}{
  \hat{\messageRV}_{\digitalTxSubset}
}
=
\entropyRV{\messageRV_{\digitalTxSubset}}
-
\entropyRVCond{\messageRV_{\digitalTxSubset}}{\hat{\messageRV}_{\digitalTxSubset}}
\geq
\entropyRV{\messageRV_{\digitalTxSubset}}
-
1
-
\decodingError
\blocklength\sum\codebookRate_{\digitalTxSubset}
\geq
\blocklength
(1-\decodingError)
\sum\codebookRate_{\digitalTxSubset}
-
1,
\end{equation}
where $\entropyRV{\cdot}$ denotes Shannon entropy and $\decodingError$ denotes the average digital decoding error which tends to $0$ since Lemma~\ref{lemma:digital-mac-with-ota-c}-\ref{item:digital-mac-with-ota-c-error-digital} is satisfied by assumption. For the analog part, we obtain
\begin{align}
\nonumber
\mutualInformationRV{
  \stateSpaceElement_{\analogTxSubset,\analogTransmissionsIndexSet}
}{
  \objectiveFunctionEstimate_{\analogTransmissionsIndexSet}
}
\overset{(a)}&{\geq}
\mutualInformationRV{
  \sum \stateSpaceElement_{\analogTxSubset, 1},
  \dots,
  \sum \stateSpaceElement_{\analogTxSubset,\numAnalogTransmissions}
}{
  \objectiveFunctionEstimate_{\analogTransmissionsIndexSet}
}
\\
\nonumber
\overset{(b)}&{=}
\sum_{\indexAnalogTransmissions=1}^\numAnalogTransmissions
\mutualInformationRVCond{
  \sum \stateSpaceElement_{\analogTxSubset, \indexAnalogTransmissions}
}{
  \objectiveFunctionEstimate_{\analogTransmissionsIndexSet}
}{
  \sum \stateSpaceElement_{\analogTxSubset, 1},
  \dots,
  \sum \stateSpaceElement_{\analogTxSubset,\indexAnalogTransmissions-1}
}
\\
\nonumber
\overset{(\ref{eq:analog-independence})}&{\geq}
\sum_{\indexAnalogTransmissions=1}^\numAnalogTransmissions
\mutualInformationRV{
  \sum \stateSpaceElement_{\analogTxSubset, \indexAnalogTransmissions}
}{
  \objectiveFunctionEstimate_\indexAnalogTransmissions
}
\\
\label{eq:converse-bound-analog}
&=
\sum_{\indexAnalogTransmissions=1}^\numAnalogTransmissions
\mutualInformationRV{
  \frac{\sum \stateSpaceElement_{\analogTxSubset,\indexAnalogTransmissions}}{\sqrt{\converseAnalogErrorVariance}}
}{
  \frac{\objectiveFunctionEstimate_\indexAnalogTransmissions}{\sqrt{\converseAnalogErrorVariance}}
}
\end{align}
where (a) is by replacing the first argument of the information with a function of it and (b) is due to the chain rule. We fix $\indexAnalogTransmissions$ for now and use the shorthand notation $\generalRV := \sum \stateSpaceElement_{\analogTxSubset,\indexAnalogTransmissions}/\sqrt{\converseAnalogErrorVariance}$ and $\hat{\generalRV} := \objectiveFunctionEstimate_\indexAnalogTransmissions/\sqrt{\converseAnalogErrorVariance}$. Moreover, we define a new random variable $\generalRVTwo \sim \normaldistribution{0}{\disturbancestd^2}$ with some $\disturbancestd > 0$. We fix the distributions of $\stateSpaceElement_{1,\indexAnalogTransmissions}, \dots, \stateSpaceElement_{\analogTxNum,\indexAnalogTransmissions}$ such that for all $\txIndex,\txIndex' \in \analogTxSubset$, we have $\stateSpaceElement_{\txIndex,\indexAnalogTransmissions} = \stateSpaceElement_{\txIndex',\indexAnalogTransmissions}$ almost surely and $\stateSpaceElement_{\txIndex,\indexAnalogTransmissions}$ follows a uniform distribution on $[-1,1]$. Together with (\ref{eq:analog-as-zero}) and (\ref{eq:converse-mse}), this means that $\generalRV$ follows a uniform distribution on $[-\cardinality{\analogTxSubset}/\sqrt{\converseAnalogErrorVariance},\cardinality{\analogTxSubset}/\sqrt{\converseAnalogErrorVariance}]$ and
\begin{equation}
\label{eq:converse-mse-normalized}
\Expectation\left(
  \left(
    \generalRV
    -
    \hat{\generalRV}
  \right)^2
\right)
\leq
1.
\end{equation}
With these definitions and choices, we calculate
\begin{align}
\nonumber
\mutualInformationRV{
  \frac{\sum \stateSpaceElement_{\analogTxSubset,\indexAnalogTransmissions}}{\sqrt{\converseAnalogErrorVariance}}
}{
  \frac{\objectiveFunctionEstimate_\indexAnalogTransmissions}{\sqrt{\converseAnalogErrorVariance}}
}
&=
\mutualInformationRV{\generalRV}{\hat{\generalRV}}
\\
\nonumber
\overset{(a)}&{\geq}
\mutualInformationRV{\generalRV}{\hat{\generalRV}+\generalRVTwo}
\\
\nonumber
\overset{(b)}&{=}
\diffEntropyRV{\generalRV}
-
\diffEntropyRVCond{\generalRV}{\hat{\generalRV}+\generalRVTwo}
\\
\nonumber
\overset{(c)}&{=}
\diffEntropyRV{\generalRV}
-
\diffEntropyRVCond{\generalRV - \hat{\generalRV} - \generalRVTwo}{\hat{\generalRV}+\generalRVTwo}
\\
\nonumber
\overset{(d)}&{\geq}
\diffEntropyRV{\generalRV}
-
\diffEntropyRV{\generalRV - \hat{\generalRV} - \generalRVTwo}
\\
\nonumber
\overset{(e)}&{\geq}
\diffEntropyRV{\generalRV}
-
\frac{1}{2}\log\left(2\pi\eulersNumber(1+\disturbancestd)^2\right)
\\
\label{eq:converse-information-bound-general}
\overset{(f)}&{=}
\frac{1}{2}\log\frac{2\cardinality{\analogTxSubset}^2}{\pi\eulersNumber(1+\disturbancestd)^2\converseAnalogErrorVariance},
\end{align}
where (a) is due to the data processing inequality. For (b), we only need to ensure that the differential entropies that appear are well-defined. To this end, we note that $\generalRV$ has a density with respect to the Lebesgue measure, implying that the differential entropy $\diffEntropyRV{\generalRV}$ exists. Since $\generalRVTwo$ has a density with respect to the Lebesgue measure and is independent of $\generalRV$ and $\hat{\generalRV}$, by Lemma~\ref{lemma:convolution-density} in Appendix~\ref{appendix:technical-lemmas-converse}, $\hat{\generalRV} + \generalRVTwo$ also has a density, as does $\hat{\generalRV} + \generalRVTwo$ conditioned on any realization of $\generalRV$. By the Bayes rule, this implies that $\generalRV$ conditioned on any realization of $\hat{\generalRV} + \generalRVTwo$ has a density and consequently, $\diffEntropyRVCond{\generalRV}{\hat{\generalRV} + \generalRVTwo}$ exists as well. (c) is due to the invariance of differential entropy to shifts. (d) is valid because conditioning does not increase entropy. To argue (e), we use (\ref{eq:converse-mse-normalized}) to calculate
\[
\Expectation\left(
  \left(
    \generalRV - \hat{\generalRV} - \generalRVTwo
  \right)^2
\right)
=
\Expectation\left(
  \left(
    \generalRV - \hat{\generalRV}
  \right)^2
\right)
-
2\Expectation\left(
  \generalRVTwo(\generalRV - \hat{\generalRV})
\right)
+
\Expectation\left(
  \generalRVTwo^2
\right)
\leq
1 + \disturbancestd^2
\]
and then use the well-known fact that differential entropy under a power constraint is maximized by a centered Gaussian distribution (see, e.g., \cite[Section 2.2]{elgamal2011network}). Finally, for (f), we use the uniform distribution of $\generalRV$ which implies $\diffEntropyRV{\generalRV} = \log(2\cardinality{\analogTxSubset}/\sqrt{\converseAnalogErrorVariance})$.

The bound (\ref{eq:converse-information-bound-general}) is valid for every $\disturbancestd > 0$ and is independent of $\indexAnalogTransmissions$. Letting $\disturbancestd \rightarrow 0$ and substituting the resulting bound in (\ref{eq:converse-bound-analog}), we obtain
\begin{equation*}
\mutualInformationRV{
  \stateSpaceElement_{\analogTxSubset,\analogTransmissionsIndexSet}
}{
  \objectiveFunctionEstimate_{\analogTransmissionsIndexSet}
}
\geq
\frac{\numAnalogTransmissions}{2}\log\frac{2\cardinality{\analogTxSubset}^2}{\pi\eulersNumber\converseAnalogErrorVariance}.
\end{equation*}
Noting
$
\mutualInformationRV{
  \stateSpaceElement_{\analogTxSubset,\analogTransmissionsIndexSet}
}{
  \objectiveFunctionEstimate_{\analogTransmissionsIndexSet}
}
\geq
0
$
and combining this with (\ref{eq:converse-bound-above}), (\ref{eq:converse-bound-below}), and (\ref{eq:converse-bound-digital}) yields
\[
\blocklength
\mutualInformationRVCond{
  \inputRV_{\analogTxSubset},
  \inputRV_{\digitalTxSubset}
}{\outputRV}{
  \inputRV_{\analogTxSubsetComplement},
  \inputRV_{\digitalTxSubsetComplement},
  \timeSharingRV
}
\geq
\blocklength
(1-\decodingError)
\sum\codebookRate_{\digitalTxSubset}
-
1
+
\positivePart{
  \frac{\numAnalogTransmissions}{2}\log\frac{2\cardinality{\analogTxSubset}^2}{\pi\eulersNumber\converseAnalogErrorVariance}
}.
\]
We divide this inequality by $\blocklength$ and then let $\blocklength \rightarrow \infty$, arriving at (\ref{eq:converse-rate-bound-mse}).
\end{proof}

\begin{cor}
\label{cor:converse}
Lemma~\ref{lemma:converse} holds with its conclusion replaced by
\begin{equation}
\label{eq:converse-rate-bound-cor}
\sum\codebookRate_{\digitalTxSubset}
\leq
\min_{\txIndex \in \{0,\dots,\analogTxNum\}}\left(
  \unconstrainedGaussianCapacity\left(
    \frac{\sum \powerconstraint_{\digitalTxSubset} + \txIndex^2\analogAmplitudeConstraint^2}
        {\noisestd^2}
  \right)
  -
  \positivePart{
    \frac{\analograte}{2}\log\frac{2\txIndex^2}{\pi\eulersNumber\converseAnalogErrorVariance}
  }
\right).
\end{equation}
\end{cor}
\begin{proof}
We fix an arbitrary $\digitalTxSubset$, a $\txIndex$ that realizes the minimum in (\ref{eq:converse-rate-bound-cor}), and some $\analogTxSubset \subseteq \analogTxSet$ with $\cardinality{\analogTxSubset}=\txIndex$. With these choices, we invoke Lemma~\ref{lemma:converse} and decompose the information term that appears in (\ref{eq:converse-rate-bound-mse}) as
\begin{align}
\nonumber
\mutualInformationRVCond{
  \inputRV_{\analogTxSubset},
  \inputRV_{\digitalTxSubset}
}{\outputRV}{
  \inputRV_{\analogTxSubsetComplement},
  \inputRV_{\digitalTxSubsetComplement},
  \timeSharingRV
}
\overset{(a)}&{=}
\mutualInformationRVCond{
  \inputRV_{\analogTxSubset}
}{\outputRV}{
  \inputRV_{\analogTxSubsetComplement},
  \inputRV_{\digitalTxSubsetComplement},
  \timeSharingRV
}
+
\mutualInformationRVCond{
  \inputRV_{\digitalTxSubset}
}{\outputRV}{
  \inputRV_{\analogTxSet},
  \inputRV_{\digitalTxSubsetComplement},
  \timeSharingRV
}
\\
\label{eq:cor-converse-decomposition}
\overset{(b)}&{=}
\mutualInformationRVCond{
  \sum \inputRV_{\analogTxSubset}
}{\outputRV}{
  \inputRV_{\analogTxSubsetComplement},
  \inputRV_{\digitalTxSubsetComplement},
  \timeSharingRV
}
+
\mutualInformationRVCond{
  \sum \inputRV_{\digitalTxSubset}
}{\outputRV}{
  \inputRV_{\analogTxSet},
  \inputRV_{\digitalTxSubsetComplement},
  \timeSharingRV
},
\end{align}
where (a) is due to the chain rule for mutual information and for (b), we apply Lemma~\ref{lemma:information-of-sum} in Appendix~\ref{appendix:technical-lemmas-converse} in each summand. For each $\txIndex \in \digitalTxSubset$, the power constraint $\Expectation\big((\inputRV_\txIndex)^2\big) \leq \powerconstraint_\txIndex$ from the statement of Lemma~\ref{lemma:converse} implies a variance constraint $\Variance(\inputRV_\txIndex) \leq \powerconstraint_\txIndex$. Due to the independence of the variables in $\digitalTxSubset$, this means that the sum obeys a variance constraint $\Variance(\sum \inputRV_{\digitalTxSubset}) \leq \sum \powerconstraint_{\digitalTxSubset}$. We can obtain a simple upper bound for the second summand in (\ref{eq:cor-converse-decomposition}) by disregarding the amplitude constraint. In this case it is known that the information is maximized by Gaussian input distributions which yields
\begin{equation}
\label{eq:cor-converse-digital-bound}
\mutualInformationRVCond{
  \sum \inputRV_{\digitalTxSubset}
}{\outputRV}{
  \inputRV_{\analogTxSet},
  \inputRV_{\digitalTxSubsetComplement},
  \timeSharingRV
}
\leq
\unconstrainedGaussianCapacity\left(\sum \powerconstraint_{\digitalTxSubset}/\noisestd^2\right).
\end{equation}
For the first summand in (\ref{eq:cor-converse-decomposition}), we do not necessarily have independence of the variables in $\analogTxSubset$, but from Lemma~\ref{lemma:converse}, we know that for every $\txIndex \in \analogTxSubset$, we have $\absolutevalue{\inputRV_\txIndex} \leq \analogAmplitudeConstraint$ almost surely. From this, we can conclude $\absolutevalue{\sum\inputRV_{\analogTxSubset}} \leq \cardinality{\analogTxSubset}\analogAmplitudeConstraint$ and therefore $\Expectation\big((\sum\inputRV_{\analogTxSubset})^2\big) \leq \cardinality{\analogTxSubset}^2 \analogAmplitudeConstraint^2$. Again disregarding the amplitude constraint and only considering the power constraint, we obtain
\begin{equation}
\label{eq:cor-converse-analog-bound}
\mutualInformationRVCond{
  \sum \inputRV_{\analogTxSubset}
}{\outputRV}{
  \inputRV_{\analogTxSubsetComplement},
  \inputRV_{\digitalTxSubsetComplement},
  \timeSharingRV
}
\leq
\unconstrainedGaussianCapacity\left(
  \frac{\cardinality{\analogTxSubset}^2 \analogAmplitudeConstraint^2}{\sum \powerconstraint_{\digitalTxSubset} + \noisestd^2}
\right).
\end{equation}
Substituting (\ref{eq:cor-converse-digital-bound}) and (\ref{eq:cor-converse-analog-bound}) into (\ref{eq:cor-converse-decomposition}) yields
\[
\mutualInformationRVCond{
  \inputRV_{\analogTxSubset},
  \inputRV_{\digitalTxSubset}
}{\outputRV}{
  \inputRV_{\analogTxSubsetComplement},
  \inputRV_{\digitalTxSubsetComplement},
  \timeSharingRV
}
\leq
\unconstrainedGaussianCapacity\left(\frac{\sum \powerconstraint_{\digitalTxSubset}}{\noisestd^2}\right)
+
\unconstrainedGaussianCapacity\left(
  \frac{\cardinality{\analogTxSubset}^2 \analogAmplitudeConstraint^2}{\sum \powerconstraint_{\digitalTxSubset} + \noisestd^2}
\right)
=
\unconstrainedGaussianCapacity\left(
  \frac{\cardinality{\analogTxSubset}^2 \analogAmplitudeConstraint^2 + \sum \powerconstraint_{\digitalTxSubset}}{\noisestd^2}
\right).
\]
If we substitute this in (\ref{eq:converse-rate-bound-mse}), the corollary follows.
\end{proof}

\section{\gls{socc} over fading channels for General Functions in $\Fmon$}
\label{sec:fmon}

For simplicity, we have restricted the analog computations in Lemma~\ref{lemma:digital-mac-with-ota-c} to the case where a sum of values in $[-1,1]$ is computed and the channel is pure \gls{awgn} without fading. In this section, we extend this to fading channels. The pre-processors we use in the extended scheme depend on the fading coefficients. This means that this scheme is only applicable when channel state information is available at the transmitter. We also extend the class of functions that can be computed to $\Fmon$ from~\cite[Definition 3]{frey2021over-tsp}. It is important to note that besides linear functions such as weighted sums (which are for instance important in federated learning applications), $\Fmon$ also contains nonlinear functions such as $p$-norms for $p \geq 1$. For convenience, we recall the definition here.
\begin{definition}
\label{def:Fmon}
\emph{(\cite[Definition 3, equations (9)-(11)]{frey2021over-tsp}).}
A measurable function $\objectiveFunction: \stateSpace_1 \times \dots \times \stateSpace_\txNum \to \reals$ is said to belong to $\Fmon$ if there exist bounded and measurable functions
$(\nomPre{\txIndex}{})_{\txIndex \in \naturalsInitialSection{\txNum}}$, a measurable set $\nomIntermediateDomain \subseteq \reals$ with the property $\nomPre{1}{}(\stateSpace_1)+\dots + \nomPre{\txNum}{}(\stateSpace_\txNum) \subseteq \nomIntermediateDomain$, a measurable function $\nomPost{}:\nomIntermediateDomain \rightarrow \reals$ such that
for all $(\stateSpaceElement_1, \dots , \stateSpaceElement_\txNum) \in \stateSpace_1 \times \dots \times \stateSpace_\txNum$, we have
\begin{equation*}
\objectiveFunction(\stateSpaceElement_1, \dots , \stateSpaceElement_\txNum) = \nomPost{}\left(  \sum_{\txIndex=1}^\txNum \nomPre{\txIndex}{}(\stateSpaceElement_\txIndex)     \right ),
\end{equation*}
and there is a strictly increasing function $\nomInc{} : [0, \infty) \to [0, \infty)$  with $\nomInc{}(0)=0$ and
\begin{equation}\label{eq:monotone-domination}
\absolutevalue{
  \nomPost{}(\generalReal) - \nomPost{}(\generalRealTwo)
}
\leq
\nomInc{}\left(
  \absolutevalue{\generalReal-\generalRealTwo}
\right)
\end{equation}
for all $\generalReal,\generalRealTwo \in \nomIntermediateDomain$. We call the function $\nomInc{}$ an \emph{increment majorant} of $\objectiveFunction$ and $\nomPre{1}{}, \dots, \nomPre{\txNum}{}, \nomPost{}$ with the properties above an \emph{$\Fmon$-nomographic representation} of $\objectiveFunction$. For $\objectiveFunction \in \Fmon$ with a fixed $\Fmon$-nomographic representation, we also define the following quantities:
\begin{equation}
\label{eq:nom-quantities}
\nomMin{\txIndex}{}
:=
\inf_{\stateSpaceElement \in \stateSpace_\txIndex} \nomPre{\txIndex}{}(\stateSpaceElement),
~~
\nomMax{\txIndex}{}
:=
\sup_{\stateSpaceElement \in \stateSpace_\txIndex} \nomPre{\txIndex}{}(\stateSpaceElement),
~~
\nomMaxSpread{\objectiveFunction}
:=
\max_{\txIndex=1}^\txNum
  (\nomMax{\txIndex}{}-\nomMin{\txIndex}{}).
\end{equation}
\end{definition}
In~\cite{frey2021over-tsp}, there are a few examples of functions that are contained in $\Fmon$, along with some discussion of their relevance for practical systems.

In this section, we consider the complex fading channel $\channelFading$ given by
\[
\outputRV
=
\fading_1 \inputRV_1
+
\dots
+
\fading_{\txNum} \inputRV_{\txNum}
+
\noiseRV
\]
where the channel inputs $\inputRV_1, \dots, \inputRV_{\txNum}$ are $\complexNumbers$-valued, and $\noiseRV$ follows a complex normal distribution with mean $0$ and variance $\noisestd^2$ per complex dimension. We assume that each transmitter $\txIndex$ knows the fading coefficient $\fading_\txIndex \in \complexNumbers \setminus \{0\}$. Moreover, we assume amplitude constraints $\maximumNorm{\inputRV_1^\blocklength} \leq \amplitudeConstraint_1, \dots, \maximumNorm{\inputRV_{\txNum}^\blocklength} \leq \amplitudeConstraint_{\txNum}$ for all transmitters and digital power constraints $\euclidNorm{\inputRV_{\analogTxNum+1}^\blocklength}^2 \leq \powerconstraint_{\analogTxNum+1}, \dots \euclidNorm{\inputRV_{\txNum}^\blocklength}^2 \leq \powerconstraint_{\txNum}$. Define
\begin{equation}
\label{eq:cor-analog-amplitude-constraint}
\analogAmplitudeConstraint
:=
\min_{\txIndex = 1, \dots, \analogTxNum}
  \absolutevalue{\fading_\txIndex}\amplitudeConstraint_\txIndex.
\end{equation}
Then we have the following corollary to Lemma~\ref{lemma:digital-mac-with-ota-c}:
\begin{cor}
\label{cor:nom}
Let $\objectiveFunction^{(1)}, \dots, \objectiveFunction^{(\numAnalogTransmissions)}$ be a tuple of functions in $\Fmon$, and for every $\indexAnalogTransmissions$, fix an $\Fmon$-nomographic representation $\nomPre{1}{\indexAnalogTransmissions}, \dots, \nomPre{\txNum}{\indexAnalogTransmissions}, \nomPost{\indexAnalogTransmissions}$ of $\objectiveFunction^{(\indexAnalogTransmissions)}$ along with an increment majorant $\nomInc{\indexAnalogTransmissions}$. We use $\nomMin{\txIndex}{\indexAnalogTransmissions}$, $\nomMax{\txIndex}{\indexAnalogTransmissions}$, and $\nomMaxSpread{\objectiveFunction^{(\indexAnalogTransmissions)}}$ to denote the quantities defined in (\ref{eq:nom-quantities}).

Let $\analograte, \analograte' \in (0,1)$ with $\analograte < \analograte'$, $\noisestd \in (0,\infty)$ and $\analogTxNum, \digitalTxNum \in \naturals$ be fixed, let $\blocklength \in \naturals$ be the block length and assume that $\analograte\blocklength \leq \numAnalogTransmissions \leq \analograte'\blocklength$. Let $\blocklength_1, \dots, \blocklength_\numAnalogTransmissions$ be natural numbers with $\blocklength = \blocklength_1 + \dots + \blocklength_\numAnalogTransmissions$. Let $\codebookRate = (\codebookRate_1, \dots, \codebookRate_{\digitalTxNum})$ be such that
\[
\codebookRate' = (\codebookRate_1', \dots, \codebookRate_{\digitalTxNum}') := (\codebookRate_1/2(1-\analograte'), \dots, \codebookRate_{\digitalTxNum}/2(1-\analograte'))
\]
is an inner point of
\begin{multline*}
\constrainedMacCapacity{\channelAWGN}\left(
  \frac{\absolutevalue{\fading_{\analogTxNum+1}}^2\powerconstraint_{\analogTxNum+1}}{1-\analograte},
  \dots,
  \frac{\absolutevalue{\fading_{\txNum}}^2\powerconstraint_{\txNum}}{1-\analograte},
  \frac{(\sqrt{2}-1)\absolutevalue{\fading_{\analogTxNum+1}}\amplitudeConstraint_{\analogTxNum+1}}{2},
  \dots,
  \frac{(\sqrt{2}-1)\absolutevalue{\fading_{\txNum}}\amplitudeConstraint_{\txNum}}{2}
\right)
\\
\supseteq
\constrainedMacCapacity{\channelAWGN}\left(
  \frac{\absolutevalue{\fading_{\analogTxNum+1}}^2 \powerconstraint_{\analogTxNum+1}}{1-\analograte},
  \dots,
  \frac{\absolutevalue{\fading_{\txNum}}^2 \powerconstraint_{\txNum}}{1-\analograte},
  \frac{\absolutevalue{\fading_{\analogTxNum+1}}\amplitudeConstraint_{\analogTxNum+1}}{4.83},
  \dots,
  \frac{\absolutevalue{\fading_{\txNum}}\amplitudeConstraint_{\txNum}}{4.83}
\right).
\end{multline*}
For large enough block lengths $\blocklength$, there is a \gls{socc} scheme consisting of pre-processors $\preprocHybrid_1, \dots, \preprocHybrid_\txNum$ and a post-processor $\postprocHybrid$ with $\stateSpace_{\txIndex,\indexAnalogTransmissions} = [\nomMin{\txIndex}{},\nomMax{\txIndex}{}]$ for all $\txIndex \in \naturalsInitialSection{\txNum}, \indexAnalogTransmissions \in \naturalsInitialSection{\numAnalogTransmissions}$; $\codebookSize_1 \geq \exp(\blocklength\codebookRate_1), \dots, \codebookSize_{\digitalTxNum} \geq \exp(\blocklength\codebookRate_{\digitalTxNum})$ such that:
\begin{enumerate}
    \item \label{item:cor-digital-mac-with-ota-c-analog-powerconstraint} The following slight generalization of \eqref{eq:analog-amplitude-constraint} holds: For all $\txIndex \in \naturalsInitialSection{\analogTxNum}$, we have
    $
    \maximumNorm{
      \preprocHybrid_\txIndex(
        \stateSpaceElement_{\txIndex,1}, \dots, \stateSpaceElement_{\txIndex,\numAnalogTransmissions}
      )
    }
    \leq
    \amplitudeConstraint_\txIndex.
    $
    \item \label{item:cor-digital-mac-with-ota-c-digital-powerconstraint} The digital amplitude constraint \eqref{eq:digital-amplitude-constraint} and power constraint \eqref{eq:digital-power-constraint} hold.
    \item \label{item:cor-digital-mac-with-ota-c-error-digital} $\decodingError$ defined in \eqref{eq:reconstruction-error} tends to $0$ uniformly for all possible realizations of $\stateSpaceElement_{1,1}, \dots, \stateSpaceElement_{1,\numAnalogTransmissions}$, $\dots$, $\stateSpaceElement_{\analogTxNum,1}, \dots, \stateSpaceElement_{\analogTxNum,\numAnalogTransmissions}$ as $\blocklength \rightarrow \infty$.
    \item \label{item:cor-digital-mac-with-ota-c-error-analog} For every $\tail\in(0,\infty)$, $\indexAnalogTransmissions \in \naturalsInitialSection{\numAnalogTransmissions}$, conditioned on every possible realization of $\stateSpaceElement_{1,1}, \dots, \stateSpaceElement_{1,\numAnalogTransmissions}$, $\dots$, $\stateSpaceElement_{\analogTxNum,1}, \dots, \stateSpaceElement_{\analogTxNum,\numAnalogTransmissions}$, $\messageRealization_1, \dots, \messageRealization_{\digitalTxNum}$, we have that $\objectiveFunctionEstimate_\indexAnalogTransmissions$ is an
    $(
      \tail,
      \exp(
        -\generalReal
      )/\sqrt{\pi\generalReal}
    )$-approximation of $\objectiveFunction^{(\indexAnalogTransmissions)}(\stateSpaceElement_{1,\indexAnalogTransmissions}, \dots, \stateSpaceElement_{\txNum,\indexAnalogTransmissions})$,
    where
    \begin{equation}
    \label{eq:nom-tail-short}
    \generalReal
    :=
    \frac{2\left({\nomInc{\indexAnalogTransmissions}}^{-1}(\tail)\right)^2}
        {\nomMaxSpread{\objectiveFunction^{(\indexAnalogTransmissions)}}^2}
    \cdot
    \frac{\analogAmplitudeConstraint^2}{\noisestd^2}
    \cdot
    \blocklength_\indexAnalogTransmissions
    .
    \end{equation}
\end{enumerate}
\end{cor}
\begin{proof}
We define, for $\txIndex\in\naturalsInitialSection{\analogTxNum}$ and $\indexAnalogTransmissions\in\naturalsInitialSection{\numAnalogTransmissions}$,
\begin{equation}
\label{eq:nom-preproc}
\stateSpaceElement_{\txIndex,\indexAnalogTransmissions}'
:=
2
\frac{
  \nomPre{\txIndex}{\indexAnalogTransmissions}(\stateSpaceElement_{\txIndex,\indexAnalogTransmissions}) - \nomMin{\txIndex}{\indexAnalogTransmissions}
}{
  \nomMaxSpread{\objectiveFunction^{(\indexAnalogTransmissions)}}
}
-
1.
\end{equation}
Then clearly, $\stateSpaceElement_{\txIndex,\indexAnalogTransmissions}' \in [-1,1]$ for any choice $\stateSpaceElement_{\txIndex,\indexAnalogTransmissions} \in \stateSpace_{\txIndex,\indexAnalogTransmissions}$. We invoke Lemma~\ref{lemma:digital-mac-with-ota-c} to obtain, for some $\codebookSize_1 \geq \exp(\blocklength\codebookRate_1), \dots, \codebookSize_{\digitalTxNum} \geq \exp(\blocklength\codebookRate)$ and $\objectiveFunctionDomain_{\txIndex,\indexAnalogTransmissions} \in [-1,1]$ for all $\txIndex \in \naturalsInitialSection{\analogTxNum}, \indexAnalogTransmissions \in \naturalsInitialSection{\numAnalogTransmissions}$, a \gls{socc} scheme for the channel $\channelAWGN$ and block length $2\blocklength$. This scheme consists of pre-processors $\preprocHybrid_1', \dots, \preprocHybrid_\txNum'$ and a post-processor $\postprocHybrid'$. The pre-processors obey the analog amplitude constraint $\analogAmplitudeConstraint/\sqrt{2}$, the digital amplitude constraints $\absolutevalue{\fading_{\analogTxNum+1}}\amplitudeConstraint_{\analogTxNum+1}/\sqrt{2}, \dots, \absolutevalue{\fading_{\txNum}}\amplitudeConstraint_{\txNum}/\sqrt{2}$, and the digital power constraints $\absolutevalue{\fading_{\analogTxNum+1}}^2 \powerconstraint_{\analogTxNum+1}, \dots, \absolutevalue{\fading_{\txNum}}^2 \powerconstraint_{\txNum}$.

Let
\begin{equation}
\label{eq:cor-original-preproc}
\inputRV_\txIndex'^\blocklength
:=
\begin{cases}
\preprocHybrid_\txIndex'(\stateSpaceElement_{\txIndex,1}', \dots, \stateSpaceElement_{\txIndex,\numAnalogTransmissions}'), &\txIndex \in \naturalsInitialSection{\analogTxNum},\\
\preprocHybrid_\txIndex'(\messageRealization), &\txIndex \in \{\analogTxNum+1, \dots, \txNum\}.
\end{cases}
\end{equation}
Then define $\preprocHybrid_\txIndex(\stateSpaceElement_{\txIndex,1}, \dots, \stateSpaceElement_{\txIndex,\numAnalogTransmissions}) := \inputRV_\txIndex^\blocklength$ for $\txIndex \in \naturalsInitialSection{\analogTxNum}$ and $\preprocHybrid_\txIndex(\messageRealization) := \inputRV_\txIndex^\blocklength$ for $\txIndex \in \{\analogTxNum+1, \dots, \txNum\}$, where
\begin{equation}
\label{eq:real-complex}
\inputRV_{\txIndex,\blIndex}
:=
\frac{1}{\fading_\txIndex}\left(
  \inputRV_{\txIndex,2\blIndex-1}'
  +
  \imaginaryUnit \cdot
  \inputRV_{\txIndex,2\blIndex}'
\right),
\end{equation}
with the imaginary unit denoted as $\imaginaryUnit$.

In order to argue that item \ref{item:cor-digital-mac-with-ota-c-analog-powerconstraint} holds, we note that for $\txIndex \in \naturalsInitialSection{\analogTxNum}$,
\begin{align*}
\maximumNorm{
  \preprocHybrid_\txIndex(
    \stateSpaceElement_{\txIndex,1}, \dots, \stateSpaceElement_{\txIndex,\numAnalogTransmissions}
  )
}
&=
\max_{\blIndex \in \naturalsInitialSection{\blocklength}} \absolutevalue{\inputRV_{\txIndex,\blIndex}}
\\
\overset{\eqref{eq:real-complex}}&{=}
\max_{\blIndex \in \naturalsInitialSection{\blocklength}}
  \frac{\sqrt{(\inputRV_{\txIndex,2\blIndex-1}')^2 + (\inputRV_{\txIndex,2\blIndex}')^2}}
       {\absolutevalue{\fading_\txIndex}}
\\
&\leq
\frac{\sqrt{2}\max_{\blIndex \in \naturalsInitialSection{2\blocklength}} \absolutevalue{\inputRV_{\txIndex,\blIndex}'}}
     {\absolutevalue{\fading_\txIndex}}
\\
\overset{\eqref{eq:cor-original-preproc}}&{=}
\frac{\sqrt{2}
        \maximumNorm{
          \preprocHybrid_\txIndex'(
            \stateSpaceElement_{\txIndex,1}', \dots, \stateSpaceElement_{\txIndex,\numAnalogTransmissions}'
          )
        }
     }
     {\absolutevalue{\fading_\txIndex}}
\\
\overset{(a)}&{\leq}
\frac{\analogAmplitudeConstraint}
     {\absolutevalue{\fading_\txIndex}}
\\
\overset{\eqref{eq:cor-analog-amplitude-constraint}}&{\leq}
\amplitudeConstraint_\txIndex,
\end{align*}
where step (a) is by the analog amplitude constraint $\analogAmplitudeConstraint/\sqrt{2}$ which $\preprocHybrid_1', \dots, \preprocHybrid_{\analogTxNum}'$ satisfy via Lemma~\ref{lemma:digital-mac-with-ota-c}-\ref{item:digital-mac-with-ota-c-analog-powerconstraint}.

For the amplitude constraints in item \ref{item:cor-digital-mac-with-ota-c-digital-powerconstraint}, we argue in an analog way that for $\txIndex \in \{\analogTxNum+1, \txNum\}$,
\[
\maximumNorm{
  \preprocHybrid_\txIndex(
    \messageRealization
  )
}
\overset{(a)}{\leq}
\frac{\sqrt{2}
        \maximumNorm{
          \preprocHybrid_\txIndex'(
            \messageRealization'
          )
        }
     }
     {\absolutevalue{\fading_\txIndex}}
\overset{(b)}{\leq}
\amplitudeConstraint_\txIndex,
\]
where step (a) is the same as the corresponding steps in the calculation for item \ref{item:cor-digital-mac-with-ota-c-analog-powerconstraint} and (b) is via Lemma~\ref{lemma:digital-mac-with-ota-c}-\ref{item:digital-mac-with-ota-c-digital-powerconstraint} because each $\preprocHybrid_\txIndex'$ satisfies the amplitude constraint $\absolutevalue{\fading_\txIndex}\amplitudeConstraint_\txIndex/\sqrt{2}$.

For the power constraints in item \ref{item:cor-digital-mac-with-ota-c-digital-powerconstraint}, we calculate
\[
\euclidNorm{
  \preprocHybrid_\txIndex(
    \messageRealization
  )
}^2
=
\sum_{\blIndex=1}^\blocklength
  \absolutevalue{\inputRV_{\txIndex,\blIndex}}^2
\overset{\eqref{eq:real-complex}}{=}
\frac{1}{\absolutevalue{\fading_\txIndex}^2}
\sum_{\blIndex=1}^{2\blocklength}
  (\inputRV_{\txIndex,\blIndex}')^2
\overset{(a)}{\leq}
\powerconstraint_\txIndex,
\]
where step (a) is via Lemma~\ref{lemma:digital-mac-with-ota-c}-\ref{item:digital-mac-with-ota-c-digital-powerconstraint} because each $\preprocHybrid_\txIndex'$ satisfies the power constraint $\absolutevalue{\fading_\txIndex}^2\powerconstraint_\txIndex$.

At the receiver, we first determine $(\hat{\messageRV}_1, \dots, \hat{\messageRV}_{\digitalTxNum}, \objectiveFunctionEstimate_{(1)}', \dots, \objectiveFunctionEstimate_{(\numAnalogTransmissions)}') := \postprocHybrid'(\outputRV^\blocklength)$. The post-processor output is then $\postprocHybrid(\outputRV^\blocklength) := (\hat{\messageRV}_1, \dots, \hat{\messageRV}_{\digitalTxNum}, \objectiveFunctionEstimate_{(1)}, \dots, \objectiveFunctionEstimate_{(\numAnalogTransmissions)})$, where
\begin{equation}
\label{eq:nom-postproc}
\objectiveFunctionEstimate_{(\indexAnalogTransmissions)}
:=
\nomPost{\indexAnalogTransmissions}\left(
  \left(
    \objectiveFunctionEstimate_{(\indexAnalogTransmissions)}'
    +
    \analogTxNum
  \right)
  \cdot
  \frac{\nomMaxSpread{\objectiveFunction^{(\indexAnalogTransmissions)}}}
       {2}
  +
  \sum_{\txIndex=1}^{\analogTxNum}
    \nomMin{\txIndex}{\indexAnalogTransmissions}
\right).
\end{equation}

Item \ref{item:cor-digital-mac-with-ota-c-error-digital} follows from Lemma~\ref{lemma:digital-mac-with-ota-c}-\ref{item:digital-mac-with-ota-c-error-digital} since due to the channel inversion in \eqref{eq:real-complex}, each use of the channel $\channel_\mathrm{fading}$ corresponds to two uses of the channel $\channelAWGN$. For item \ref{item:cor-digital-mac-with-ota-c-error-analog}, we first note that by Lemma~\ref{lemma:digital-mac-with-ota-c}-\ref{item:digital-mac-with-ota-c-error-analog}, each $\objectiveFunctionEstimate_{(\indexAnalogTransmissions)}'$ is a Gaussian approximation of $\stateSpaceElement_{1,\indexAnalogTransmissions}' + \dots + \stateSpaceElement_{\txNum,\indexAnalogTransmissions}'$ with variance $\noisestd^2/\blocklength_\indexAnalogTransmissions\analogAmplitudeConstraint^2$ (note that we substitute $\blocklength_\indexAnalogTransmissions' := 2\blocklength_\indexAnalogTransmissions$ and amplitude constraint $\analogAmplitudeConstraint' := \analogAmplitudeConstraint/\sqrt{2}$ into the expression given in Lemma~\ref{lemma:digital-mac-with-ota-c}, with the resulting factors of $2$ canceling in the expression). In other words,  we can write, for $\indexAnalogTransmissions\in\naturalsInitialSection{\numAnalogTransmissions}$,
\begin{equation*}
\objectiveFunctionEstimate_{(\indexAnalogTransmissions)}'
=
\stateSpaceElement_{1,\indexAnalogTransmissions}'
+
\dots
+
\stateSpaceElement_{\analogTxNum,\indexAnalogTransmissions}'
+
\effectiveNoise^{(\indexAnalogTransmissions)},
\end{equation*}
where $\effectiveNoise^{(\indexAnalogTransmissions)}$ is distributed as $\normaldistribution{0}{\noisestd^2/\blocklength_\indexAnalogTransmissions \analogAmplitudeConstraint^2}$. So we can calculate
\begin{align*}
\objectiveFunctionEstimate_{(\indexAnalogTransmissions)}
\overset{(\ref{eq:nom-postproc})}&{=}
\nomPost{\indexAnalogTransmissions}\left(
  \left(
    \sum_{\txIndex=1}^{\analogTxNum}
      \stateSpaceElement_{\txIndex,\indexAnalogTransmissions}'
    +
    \effectiveNoise^{(\indexAnalogTransmissions)}
    +
    \analogTxNum
  \right)
  \cdot
  \frac{\nomMaxSpread{\objectiveFunction^{(\indexAnalogTransmissions)}}}
       {2}
  +
  \sum_{\txIndex=1}^{\analogTxNum}
    \nomMin{\txIndex}{\indexAnalogTransmissions}
\right)
\\
\overset{(\ref{eq:nom-preproc})}&{=}
\nomPost{\indexAnalogTransmissions}\left(
  \left(
    2
    \sum_{\txIndex=1}^{\analogTxNum}
      \frac{
        \nomPre{\txIndex}{\indexAnalogTransmissions}(\stateSpaceElement_{\txIndex,\indexAnalogTransmissions}) - \nomMin{\txIndex}{\indexAnalogTransmissions}
      }{
        \nomMaxSpread{\objectiveFunction^{(\indexAnalogTransmissions)}}
      }
    +
    \effectiveNoise^{(\indexAnalogTransmissions)}
  \right)
  \cdot
  \frac{\nomMaxSpread{\objectiveFunction^{(\indexAnalogTransmissions)}}}
       {2}
  +
  \sum_{\txIndex=1}^{\analogTxNum}
    \nomMin{\txIndex}{\indexAnalogTransmissions}
\right)
\\
&=
\nomPost{\indexAnalogTransmissions}\left(
  \sum_{\txIndex=1}^{\analogTxNum}
    \nomPre{\txIndex}{\indexAnalogTransmissions}(\stateSpaceElement_{\txIndex,\indexAnalogTransmissions})
  +
  \effectiveNoise^{(\indexAnalogTransmissions)}
  \cdot
  \frac{\nomMaxSpread{\objectiveFunction^{(\indexAnalogTransmissions)}}}
       {2}
\right).
\end{align*}
This yields
\begin{align*}
\Probability\left(
  \absolutevalue{
    \objectiveFunction^{(\indexAnalogTransmissions)}(\stateSpaceElement_{1,\indexAnalogTransmissions}, \dots, \stateSpaceElement_{\analogTxNum,\indexAnalogTransmissions})
    -
    \objectiveFunctionEstimate_{(\indexAnalogTransmissions)}
  }
  \geq
  \tail
\right)
&=
\Probability\left(
  \absolutevalue{
    \nomPost{\indexAnalogTransmissions}\left(
      \sum_{\txIndex=1}^{\analogTxNum}
        \nomPre{\txIndex}{\indexAnalogTransmissions}(\stateSpaceElement_{\txIndex,\indexAnalogTransmissions})
    \right)
    -
    \nomPost{\indexAnalogTransmissions}\left(
      \sum_{\txIndex=1}^{\analogTxNum}
        \nomPre{\txIndex}{\indexAnalogTransmissions}(\stateSpaceElement_{\txIndex,\indexAnalogTransmissions})
      +
      \effectiveNoise^{(\indexAnalogTransmissions)}
      \cdot
      \frac{\nomMaxSpread{\objectiveFunction^{(\indexAnalogTransmissions)}}}
           {2}
    \right)
  }
  \geq
  \tail
\right)
\\
\overset{(\ref{eq:monotone-domination})}&{\leq}
\Probability\left(
  \nomInc{\indexAnalogTransmissions}\left(
    \absolutevalue{
      \effectiveNoise^{(\indexAnalogTransmissions)}
      \cdot
      \frac{\nomMaxSpread{\objectiveFunction^{(\indexAnalogTransmissions)}}}
           {2}
    }
  \right)
  \geq
  \tail
\right)
\\
&=
\Probability\left(
  \absolutevalue{
    \frac{\effectiveNoise^{(\indexAnalogTransmissions)}\analogAmplitudeConstraint\sqrt{\blocklength_\indexAnalogTransmissions}}
         {\noisestd}
  }
  \geq
  \frac{2{\nomInc{\indexAnalogTransmissions}}^{-1}(\tail)\analogAmplitudeConstraint\sqrt{\blocklength_\indexAnalogTransmissions}}
       {\nomMaxSpread{\objectiveFunction^{(\indexAnalogTransmissions)}} \noisestd}
\right)
\\
\overset{(\ref{eq:nom-tail-short})}&{=}
\Probability\left(
  \absolutevalue{
    \frac{\effectiveNoise^{(\indexAnalogTransmissions)}\analogAmplitudeConstraint\sqrt{\blocklength_\indexAnalogTransmissions}}
         {\noisestd}
  }
  \geq
  \sqrt{2\generalReal}
\right)
.
\end{align*}
Since $\effectiveNoise^{(\indexAnalogTransmissions)}\analogAmplitudeConstraint\sqrt{\blocklength_\indexAnalogTransmissions}/\noisestd$ is standard Gaussian, we can use the tail bound~\cite[Proposition 2.1.2]{vershynin2018high} to finish the proof of the corollary.
\end{proof}

\section{Simulation Results}
\label{sec:simulations}
\begin{figure}
\centering
\begin{subfigure}[t]{.45\textwidth}
\begin{tikzpicture}
\begin{semilogyaxis}[
  xlabel={noise power in $\unit{\dB}$},
  ylabel={digital \acrshort{ber}},
  ymin=1e-6,
  ymax=1,
  legend pos=south east
]
  \addplot[black,solid,mark=triangle*,mark options=solid] table[x=noise_power_dB, y=digital_ber, col sep=comma] {pgf_data_middleton.csv};
  \addplot[black,dashed,mark=diamond*,mark options=solid] table[x=noise_power_dB, y=digital_ber, col sep=comma] {pgf_data_middleton_moderate_short_blocks.csv};
  \addplot[black,dashed,mark=*,mark options=solid] table[x=noise_power_dB, y=digital_ber, col sep=comma] {pgf_data_short_blocks.csv};
  \addplot[black,solid,mark=diamond*,mark options=solid] table[x=noise_power_dB, y=digital_ber, col sep=comma] {pgf_data_middleton_moderate.csv};
  \addplot[black,solid,mark=*,mark options=solid] table[x=noise_power_dB, y=digital_ber, col sep=comma] {pgf_data.csv};
  \addplot[black,densely dotted] coordinates {(-9.999984753320783,1e-8) (-9.999984753320783,1)};
  \legend{{M $0.1$ $2000$}, {M $1.5$ $324$}, {G $324$}, {M $1.5$ $2000$}, {G $2000$}, achievable};
\end{semilogyaxis}
\end{tikzpicture}
\caption{Digital \acrshort{ber} for varying levels of channel noise. ``G'' stands for Gaussian noise, ``M $1.5$'' for Middleton Class A noise with impulsive index and Gaussian-to-impulsive power ratio of $1.5$, and ``M 0.1''  for Middleton Class A noise with impulsive index and Gaussian-to-impulsive power ratio of $0.1$. The last number indicates how many digital bits were encoded per transmission, which at the parameters of the simulated scenario corresponds to a block length of $740$ complex channel uses at $2000$ bits and $120$ complex channel uses at $324$ bits. The vertical dotted line indicates the noise level below which the scheme is asymptotically achievable for block lengths tending to $\infty$ according to Theorem~\ref{theorem:capacity}.}
\label{fig:simulation-ber}
\end{subfigure}
\hspace{.02\textwidth}
\begin{subfigure}[t]{.45\textwidth}
\begin{tikzpicture}
\begin{semilogyaxis}[
  xlabel={noise power in $\unit{\dB}$},
  ylabel={analog \acrshort{mse}},
  legend style={at={(0.5,0.5)}}
]
  \addplot[black,solid,mark=*,mark options=solid] table[x=noise_power_dB, y=analog_mse, col sep=comma] {pgf_data_middleton.csv};
  \addplot[black,dashed,mark=square*,mark options=solid] table[x=noise_power_dB, y=analog_mse, col sep=comma] {pgf_data_middleton_norm.csv};
  \legend{sum,norm}
\end{semilogyaxis}
\end{tikzpicture}
\caption{Analog \gls{mse} for varying levels of channel noise. It should be kept in mind that the analog transmitters have a very low transmission power of $\qty{-10}{\dB}$. The analog \glspl{mse} were determined for the same parameter sets as in Fig.~\ref{fig:simulation-ber} and confirmed to be equal for all evaluated block lengths and noise distributions. Therefore, only one line is shown per \gls{ota} computed function.}
\label{fig:simulation-mse}
\end{subfigure}
\caption{Digital and analog error performance for the simulated \gls{socc} scenario. Every data point shown is averaged over the transmission of approximately $2 \cdot 10^9$ digital bits.}
\label{fig:simulation-errors}
\end{figure}

\begin{figure}
\centering
\begin{tikzpicture}
\begin{axis} [
  yticklabel style={/pgf/number format/.cd,fixed,precision=2},
  xlabel={maximum amplitude ratio},
  ylabel={relative frequency},
  ymin=0,
  ymax=.22,
]
  \addplot[ybar interval,mark=no,fill=black] table[x=bin_left, y=density, col sep=comma] {pgf_data_histogram.csv};
\end{axis}
\end{tikzpicture}
\caption{Histogram showing the ratio between the maximal amplitude of the encoded and modulated digital code words before and after Lemma~\ref{lemma:zerosum} processing. This includes all code words across all parameter sets contained in Fig.~\ref{fig:simulation-errors}. The histogram shows that the maximal amplitude of none of the code words is increased by a factor of more than $1.6$ in the simulated scenario. Note that Lemma~\ref{lemma:zerosum} guarantees that the processing does not affect the average power of the code words.}
\label{fig:simulation-amplitude}
\end{figure}

The \gls{socc} scheme proposed in Lemma~\ref{lemma:digital-mac-with-ota-c} uses a random coding argument to prove the reliability of the digital transmissions. This means that reliable decoding is possible only asymptotically for block lengths tending to infinity. A straightforward implementation of this procedure (as is normal for random codes) would entail exponential computational complexity at codebook generation, encoding, and decoding. Also, an exponential amount of memory that would be needed to store the random codebook. However, the technique we use in the proof of Lemma~\ref{lemma:digital-mac-with-ota-c} consists of a constructive modification by additional pre- and post-processing as per Lemma~\ref{lemma:zerosum} which consists of relatively low-dimensional linear maps. Hence, these operations do not present a high computational burden in addition to the digital encoding and decoding. In the proof of Lemma~\ref{lemma:digital-mac-with-ota-c}, this modification is applied to a standard random codebook to prove the achievability bound for \gls{socc}, but it is agnostic to the digital coding technique used and can also be applied to computationally feasible codes.

To illustrate this point and explore the practicality of \gls{socc} schemes that are generated in this way, we have implemented such a \gls{socc} scheme based on an \gls{ldpc} coding and modulation scheme that is used in \acrshort{5g}. The Python code used for the simulations is available as an electronic supplement with this paper. In our example, we use a system model with $\analogTxNum=10$ analog transmitters with a transmission power of $\qty{-10}{\dB}$, $\digitalTxNum=1$ digital transmitter\footnote{We choose the case of a single transmitter for simplicity, but the \gls{socc} scheme can be used with any practical multiple-access scheme in the same way, as detailed in the proof of Lemma~\ref{lemma:zerosum}.} with a transmission power of $\qty{0}{\dB}$, and a varying noise power. For digital encoding, we use the \texttt{py3gpp} toolbox to \gls{ldpc}-encode the digital message at a rate of $769/1024$ and modulate the resulting code word as complex channel symbols with a 16QAM scheme, yielding an overall code rate of about $2.08$ nats per complex channel symbol. The input values held at the analog transmitters are drawn in such a way that the sums of all $10$ values at each computation are uniformly distributed in $[-10,10]$, and the analog rate is set at $0.1$. This corresponds to one computation per $10$ real channel uses or to one computation per $5$ complex channel uses. The scheme modifications proposed in Lemma~\ref{lemma:zerosum} are designed for real signals and channels. Therefore, the signals are mapped from their complex representations to real vectors and converted back to complex signals afterwards. This is done in the same way as in the proof of Corollary~\ref{cor:nom}. Disregarding amplitude constraints,\footnote{This means in particular that we cannot expect any scheme that is constructed as proposed in Lemma~\ref{lemma:zerosum} to be reliable at noise powers greater than $\qty{-10.0}{\dB}$.} the scheme with these parameters is achievable for noise powers below $\qty{-10.0}{\dB}$ according to Theorem~\ref{theorem:capacity}. We simulate the \gls{socc} scheme both for the computation of sums and for $2$-norms (which is a nonlinear function).

In Fig.~\ref{fig:simulation-ber}, we show the digital \gls{ber} for two different block lengths and various distributions of the additive channel noise. Besides Gaussian noise, we also simulate Middleton Class A noise, which has been found in real-world measurements to fairly accurately model various different common types of non-Gaussian interference in wireless communications. Examples of this are interfering radio communications, electromagnetic emissions from power lines, and interference from cosmic radiation~\cite[Table 2.1]{middleton1993elements}. We compare the results for Gaussian distributions with a Middleton Class A distribution with an impulsive index and Gaussian-to-impulsive power ratio of $1.5$, representing a moderate departure from the Gaussianity assumption, and with a Middleton Class A distribution where these two non-Gaussianity parameters have been set to $0.1$, representing a more strongly non-Gaussian noise.\footnote{The two non-Gaussianity parameters of the Middleton Class A distribution range in $(0,\infty)$ and the Middleton Class A noise approaches a Gaussian distribution whenever at least one of these parameters tends to $\infty$.} As can be seen in the plots, the \gls{socc} scheme can be practically implemented at moderate block lengths with a performance that is only about $\qty{1}{\dB}$ away from the theoretical achievability bound, and the performance degrades only moderately when we use extremely short block lengths and/or the channel noise departs from the assumption of Gaussianity. The error curves are the same regardless of whether sums or $2$-norms are computed in the analog part of the system (as can be expected from the theoretical results).

In Fig.~\ref{fig:simulation-mse}, we show the \gls{mse} of the two \gls{ota} computed functions. The error values are the same for all parameter sets evaluated in Fig.~\ref{fig:simulation-ber}, so we show only one curve for each \gls{ota} computed function. It should be kept in mind that in our example scenario, we have analog transmitters of a very low transmission power. Hence, it can be expected that if the analog transmission power increases, the error will reduce in the same manner as for a corresponding decrease in noise power.

In Fig.~\ref{fig:simulation-amplitude}, we show how the additional processing affects the maximum amplitude of the digital code words. Compared to standard \gls{ldpc} coding and 16QAM modulation, our proposed \gls{socc} scheme increases the maximal amplitude of code words by not more than a factor of approximately $1.6$, which is much better performance in this regard than the theoretical upper bound of $\sqrt{2}/(\sqrt{2}-1)\approx3.42$ according to Lemma~\ref{lemma:zerosum}-\ref{item:zerosum-peakpower}, indicating that our analysis may be a little bit too pessimistic. We remark, however, that it can be seen in the proof of Lemma~\ref{lemma:orthogonal} that this factor increases as the analog computation rate $\analograte$ approaches $0$ and that this factor is lower whenever $1/\analograte$ is close to a power of $2$. In our example, we have $\analograte^{-1} = 10$, representing neither the best nor the worst possible case. For any given value of $\analograte^{-1}$, $\maximumNorm{\planeMap_{\analograte^{-1}}}$ (defined in the proof of Lemma~\ref{lemma:orthogonal}) can easily be evaluated numerically. In the simulated scenario, we have $\maximumNorm{\planeMap_{10}} \approx 1.75$, indicating that this provides a reasonably sharp upper bound for the factor of maximum amplitude increase.

We conclude from our numerical evaluations that the \gls{socc} scheme we propose in this paper is very promising in terms of its practical applicability. Still, we leave many important questions in this regard open for further research, such as the effects of fading with imperfect channel state information and the presence of interfering wireless communications, as well as the importance of perfect time synchronization between the digital and analog transmitters.

\section{Concluding Remarks}
\label{sec:analograte-mse-relationship}
The communication scheme we propose to show achievability in this paper only works if the relationship \eqref{eq:analograte-mse-relationship} between $\analograte$ and $\converseAnalogErrorVariance$ is satisfied. This means that if we want to increase the analog rate $\analograte$ of the scheme, we need to allow for a proportional increase in the error $\converseAnalogErrorVariance$ of the analog \gls{otac} function values. At the same time, the increase in analog rate by \eqref{eq:inner-bound} means that the digital rates will typically decrease as well. While it is a property we can expect from any reasonable tradeoff that if we want to increase one performance parameter (here $\analograte$), other performance parameters will necessarily have to deteriorate, \eqref{eq:analograte-mse-relationship} represents a somewhat more restrictive constraint than is typical: Because it mandates a direct relationship between $\analograte$ and $\converseAnalogErrorVariance$, it is not possible to, e.g., sacrifice some additional digital rate in order to be able to decrease $\converseAnalogErrorVariance$ while keeping $\analograte$ constant. From this point of view, \eqref{eq:analograte-mse-relationship} represents a restriction on system design in addition to what one would ordinarily expect. In the remainder of this section, we study the impact of this restriction and possible mitigation. In Section~\ref{sec:analograte-mse-relationship-application}, we take a closer look at an example communication system which uses \gls{otac} for a federated learning application and discuss how this restriction impacts the feasibility of using \gls{otac} in this case. In Section~\ref{sec:analograte-mse-relationship-extension}, we discuss ideas of how this restriction could be lifted or at least mitigated in further research.

\subsection{Application to Federated Learning}
\label{sec:analograte-mse-relationship-application}
Most envisioned applications of analog \gls{otac} are designed to deal with communication noise of different strengths, as long as it is possible to a limited extent to control this noise. It is not commonly required to design the noise strength to exact specifications. Examples for communication systems in which this is the case include proposed systems for federated learning~\cite{ang2020robust,wei2022federated} as well as widespread applications in communication systems that deal with problems in control, such as average consensus in networks~\cite{rajagopal2010network} and distributed convex optimization~\cite{agrawal2024distributed}.

In the following, we take a closer look at the federated learning example. The federated learning protocols proposed in~\cite{ang2020robust,wei2022federated} lead to convergence under the following conditions:
\begin{itemize}
  \item For the \emph{expectation-based model},~\cite[Proposition 2 and Remark 3]{ang2020robust} state that the overall error converges to a local optimum with rate $\landauO(1/\flRounds)$ (where $\flRounds$ is the number of communication rounds) as long as the variance of the effective noise adheres to a (constant) upper bound stated in ~\cite[Remark 2]{ang2020robust}.
  \item For the \emph{model transmission scheme},~\cite[Theorem 1]{wei2022federated} states that the overall error converges to a local optimum with rate $\landauO(1/\flRounds)$ as long as the variance of the effective noise scales as $\landauO(1/\flRounds^2)$.
  \item For the \emph{model differential transmission scheme},~\cite[Theorem 3]{wei2022federated} states that the overall error converges to a local optimum with rate $\landauO(1/\flRounds)$ as long as the variance of the effective noise scales as $\landauO(1)$ (i.e., it is allowed to remain constant and does not need to converge to $0$).
\end{itemize}
All of these results have in common that they require the effective noise variance only to satisfy some bound (which for one of the schemes becomes lower for later rounds of the scheme). They can, however, accommodate a variety of different noise powers without it directly affecting the convergence rate. They also have in common that, while the authors do provide more detailed expressions, the main takeaway emphasized in the results is the convergence rate in $\landauO$ notation. In light of our results in Theorem~\ref{theorem:capacity},  it is always possible to integrate such schemes into a wireless network with concurrent digital communications using \gls{socc}. Specifically, as long as the corresponding value for $\converseAnalogErrorVariance$ according to \eqref{eq:analograte-mse-relationship} is within the fairly large region of effective noise powers in which convergence is guaranteed, the federated learning scheme can simply be run until the convergence needed by the application is achieved. This means that there is also significant design freedom in choosing the analog rate based on the availability of channel resources that are not being used for digital communication. Moreover, we point out that the authors of~\cite{wei2022federated} propose adjusting the variance of the effective noise with power allocation at the analog transmitters over the subsequent communication rounds, which in the case of our \gls{socc} scheme would vary the analog error without impacting on the communication rates at all. This is therefore an approach that would work in a similar way as it does in~\cite{wei2022federated} when \gls{socc} is used to carry out \gls{otac}.

\subsection{Extension in Further Research}
\label{sec:analograte-mse-relationship-extension}
In the previous subsection, we have shown for an example that implementing \gls{socc} in a communication system with \gls{otac} will not in every case hinge on an ability to trade off analog error for performance parameters of the system other than the analog rate. In this subsection, we discuss some possible directions for future research.

To a certain extent, the relationship \eqref{eq:analograte-mse-relationship} is inherent to the proposed communication scheme and cannot easily be fully lifted. However, we do have some ideas on how it could be mitigated as a part of future research. The basic idea is that the \gls{socc} scheme we propose is necessarily simplified compared to what would eventually be used in a real-world communication system. This simplification is necessary for the scheme to be amenable to theoretical analysis. It is very likely that research which intends to bring \gls{socc} closer to real-world applications needs to make a number of modifications and extensions. We expect that they will be necessary to address challenges and complications faced in specific implementation scenarios. A mitigation of the strictness of \eqref{eq:analograte-mse-relationship} is something that could naturally be achieved as part of an extension that also addresses a different challenge of real-world applications.

The idea we intend to outline here is specific to a scenario such as in Section~\ref{sec:fmon} where the analog signals are subject to fading and the analog transmitters may be subject to different amplitude constraints. In this case, our scheme adjusts the effective transmit amplitude of all analog transmitters as shown in \eqref{eq:cor-analog-amplitude-constraint} to match that of the weakest user in the system. The cost of this in terms of performance could be massive if the product of amplitude constraint and channel strength $\absolutevalue{\fading_{\txIndex}}\amplitudeConstraint_{\txIndex}$ vary a lot among the analog users. One possible extension to our \gls{socc} scheme would be to select a number of users with the lowest values of $\absolutevalue{\fading_{\txIndex}}\amplitudeConstraint_{\txIndex}$ and have these users transmit a source and channel encoded version of their analog values during the digital part of the communication. The receiver would adapt to this by decoding the digitally transmitted values and adjust the computed sum value in post-processing. This would have the following effects on the performance parameters of the communication system:
\begin{itemize}
  \item The remaining analog users could transmit at a higher power and make fuller use of their amplitude allowance, decreasing the analog error $\converseAnalogErrorVariance$ (notice that \eqref{eq:analograte-mse-relationship} contains not just $\analograte$ but also $\analogAmplitudeConstraint$ which would increase according to \eqref{eq:cor-analog-amplitude-constraint}).
  \item The analog rate $\analograte$ used in the communication system would not change.
  \item Since some of the analog users now participate in the digital part of the communication system, the sum rate available for the remaining digital communications would decrease.
\end{itemize}
This means we have achieved the tradeoff which according to our example in the opening paragraph of Section~\ref{sec:analograte-mse-relationship} is not possible to achieve in the unmodified version of the \gls{socc} scheme.

We conclude this section with the remark that this idea is just a simple example of how the strictness of \eqref{eq:analograte-mse-relationship} can be mitigated while addressing different practical limitations of the \gls{socc} scheme. Another example is to use the scheduling scheme proposed in~\cite{nazer2024computation} to mitigate the requirement of channel state information at the transmitters that is needed for the transmitter pre-processing operation as seen in \eqref{eq:real-complex}. The authors of~\cite{nazer2024computation} propose to relax this requirement by scheduling groups of analog transmissions of users at a time that have similar channel coefficients. It is then only necessary to feed these selection patterns back to the transmitters which presumably have a significantly lower rate than the full channel state information would have. This can be seen as a more refined version of the extension described in the preceding paragraphs. This extension also aims at making the users that remain analog transmitters more similar in terms of channel gains. Similarly, this more refined version would allow to modify the tradeoff between analog rates, digital rates, and analog errors. For instance, it would be possible to allocate different analog rates to different groups of users. This would result in an average analog rate that can be prescribed by a design choice. If a reduction in error at the cost of digital sum rate is intended, it is also possible to have some of the smaller groups of analog transmitters participate in the digital part of the communication scheme.

\appendix
\subsection{Table of Symbols}
\label{appendix:symbols}

\renewcommand{\glossarysection}[2][]{}
\glsfindwidesttoplevelname
\renewcommand{\glstreenamefmt}{}
\printnoidxglossary[sort=use, type=main, style=alttree]

\subsection{Proof of Lemma~\ref{lemma:orthogonal}}
\label{appendix:lemma-orthogonal}
In this proof, we use the following conventions that are common in linear algebra: We identify linear maps between Euclidean spaces with matrices, use $\innerproduct{\cdot}{\cdot}$ to denote the standard scalar product, and write the all-ones vector of length $\generalNatural$ as $\onevector_\generalNatural$. For the purposes of matrix-vector and vector-vector multiplication, we consider vectors to be column vectors unless transposition is indicated, and we sometimes identify tuples in information-theoretic notation such as $\generalReal^\generalNatural$ with column vectors. We use $\euclidNorm{\cdot}$ and $\maximumNorm{\cdot}$ to denote the operator norms induced by the standard Euclidean norm and maximum norm, respectively. Outside of this proof, however (including the lemma statement), we stick closely to information-theoretic notational conventions and treat $\planeMap_\generalNatural$ and $\planeMapInv_\generalNatural$ as maps between sets of tuples of real numbers, and we do not identify them with matrices or distinguish between row and column vectors.

For $\generalNatural = 1$, we use the convention that $\reals^0 := \{0\}$ denotes a zero-dimensional vector space. We note that there is only one linear map $\reals \rightarrow \reals^0$ and $\reals^0 \rightarrow \reals$ and that these maps satisfy the lemma statement. Therefore, we assume $\generalNatural \geq 2$ from now on.
For $\generalNatural = 2$, we define
\begin{equation}
\label{eq:planemap-def-base}
\planeMap_2
:=
\begin{pmatrix}
2^{-\frac{1}{2}} \\
- 2^{-\frac{1}{2}}
\end{pmatrix}.
\end{equation}

For $\generalNatural > 2$, we define $\planeMap_\generalNatural$ recursively. For even $\generalNatural$, we let
\begin{equation}
\label{eq:planemap-def-even}
\planeMap_\generalNatural
:=
\begin{pmatrix}
 &                                       & &                                       & & 1/\sqrt{\generalNatural} \\
 & \planeMap_{\frac{\generalNatural}{2}} & & 0                                     & & \vdots \\
 &                                       & &                                       & & 1/\sqrt{\generalNatural} \\
 &                                       & &                                       & & -1/\sqrt{\generalNatural} \\
 & 0                                     & & \planeMap_{\frac{\generalNatural}{2}} & & \vdots \\
 &                                       & &                                       & & -1/\sqrt{\generalNatural}
\end{pmatrix}.
\end{equation}
For odd $\generalNatural$, we let
\begin{equation}
\label{eq:planemap-def-odd}
\planeMap_\generalNatural
:=
\begin{pmatrix}
 &                               & & 1/\sqrt{\generalNatural^2- \generalNatural} \\
 & \planeMap_{\generalNatural-1} & & \vdots \\
 &                               & & 1/\sqrt{\generalNatural^2- \generalNatural} \\
 & 0                             & & - \sqrt{(\generalNatural-1)/\generalNatural}
\end{pmatrix}.
\end{equation}
In the following, we use
\[
\planeMap_\generalNatural
=
\begin{pmatrix}
\planeMapElement^{(\generalNatural)}_{1,1}               & \cdots & \planeMapElement^{(\generalNatural)}_{1,\generalNatural-1} \\
\vdots                                                   &        & \vdots \\
\planeMapElement^{(\generalNatural)}_{\generalNatural,1} & \cdots & \planeMapElement^{(\generalNatural)}_{\generalNatural,\generalNatural-1}
\end{pmatrix}
\]
and $\basisTwo^{(\generalNatural)}_{\generalIndex} := \transpose{\left(\planeMapElement^{(\generalNatural)}_{1,\generalIndex}, \dots, \planeMapElement^{(\generalNatural)}_{\generalNatural,\generalIndex}\right)}$ to denote the columns of $\planeMap_\generalNatural$.

We first show by induction that $\basisTwo^{(\generalNatural)}_1, \dots, \basisTwo^{(\generalNatural)}_{\generalNatural-1}$ forms an orthonormal basis of the hyperplane
\[
\plane_\generalNatural := \left\{
  \generalReal^\generalNatural \in \reals^\generalNatural:~
  \innerproduct{\generalReal^\generalNatural}{\onevector_\generalNatural} = 0
\right\}.
\]
It is obvious from (\ref{eq:planemap-def-base}) that $\basisTwo^{(2)}_1$ forms an orthonormal basis of the one-dimensional subspace $\plane_2$ since it has length $1$ and is clearly an element of $\plane_2$. For $\generalNatural > 2$, we observe that due to the definitions (\ref{eq:planemap-def-even}) and (\ref{eq:planemap-def-odd}), $\basisTwo^{(\generalNatural)}_{\generalNatural-1}$ (the last column of $\planeMap_\generalNatural$) clearly has length $1$, it is an element of $\plane_\generalNatural$ since its elements sum to $0$, and it is orthogonal to all the other columns since their non-zero entries sum to $0$ (by induction hypothesis, these sub-vectors are elements of $\plane_{\generalNatural/2}$ respectively $\plane_{\generalNatural-1}$) and are aligned with entries in $\basisTwo^{(\generalNatural)}_{\generalNatural-1}$ which are all equal. Similarly, it is clear from (\ref{eq:planemap-def-even}), (\ref{eq:planemap-def-odd}), and the induction hypothesis that $\basisTwo^{(\generalNatural)}_{1}, \dots, \basisTwo^{(\generalNatural)}_{\generalNatural-2}$ are also of unit length, mutually orthogonal, and elements of $\plane_{\generalNatural}$. Since $\plane_{\generalNatural}$ is of dimension $\generalNatural - 1$, this concludes the proof that $\basisTwo^{(\generalNatural)}_1, \dots, \basisTwo^{(\generalNatural)}_{\generalNatural-1}$ forms an orthonormal basis of $\plane_\generalNatural$.

We write
\[
\planeMap_\generalNatural
=
\sum_{\generalIndex=1}^{\generalNatural-1}
  \basisTwo^{(\generalNatural)}_\generalIndex \transpose{\left(\basisOne^{(\generalNatural-1)}_\generalIndex\right)},
\]
where $\basisOne^{(\generalNatural-1)}_1, \dots, \basisOne^{(\generalNatural-1)}_{\generalNatural-1}$ denotes the standard basis of $\reals^{\generalNatural-1}$. To conclude the proof of item \ref{item:orthogonal-inv}, we define $\planeMapInv_\generalNatural$ as the matrix transpose of $\planeMap_\generalNatural$ and use the orthonormality of the basis systems to argue
\[
\transpose{\planeMap_\generalNatural} \planeMap_\generalNatural
=
\left(
  \sum_{\generalIndex=1}^{\generalNatural-1}
    \basisOne^{(\generalNatural-1)}_\generalIndex \transpose{\left(\basisTwo^{(\generalNatural)}_\generalIndex\right)}
\right)
\left(
  \sum_{\generalIndex=1}^{\generalNatural-1}
    \basisTwo^{(\generalNatural)}_\generalIndex \transpose{\left(\basisOne^{(\generalNatural-1)}_\generalIndex\right)}
\right)
=
\sum_{\generalIndex=1}^{\generalNatural-1}
  \basisOne^{(\generalNatural-1)}_\generalIndex \transpose{\left(\basisOne^{(\generalNatural-1)}_\generalIndex\right)}
=
\identityMapping{\reals^{\generalNatural-1}}.
\]
For item \ref{item:orthogonal-plane}, we note that
\[
\planeMap_\generalNatural \generalReal^{\generalNatural-1}
=
\left(
\sum_{\generalIndex=1}^{\generalNatural-1}
  \basisTwo^{(\generalNatural)}_\generalIndex \transpose{\left(\basisOne^{(\generalNatural-1)}_\generalIndex\right)}
\right)
\generalReal^{\generalNatural-1}
=
\sum_{\generalIndex=1}^{\generalNatural-1}
  \basisTwo^{(\generalNatural)}_\generalIndex \innerproduct{\basisOne^{(\generalNatural-1)}_\generalIndex}{\generalReal^{\generalNatural-1}}
\in
\plane_\generalNatural,
\]
so the statement holds by definition of $\plane_\generalNatural$. For item \ref{item:orthogonal-zero}, we calculate
\[
\planeMapInv_\generalNatural \onevector_\generalNatural
=
\left(
  \sum_{\generalIndex=1}^{\generalNatural-1}
    \basisOne^{(\generalNatural-1)}_\generalIndex \transpose{\left(\basisTwo^{(\generalNatural)}_\generalIndex\right)}
\right)
\onevector_\generalNatural
=
\sum_{\generalIndex=1}^{\generalNatural-1}
  \basisOne^{(\generalNatural-1)}_\generalIndex \innerproduct{\basisTwo^{(\generalNatural)}_\generalIndex}{\onevector_\generalNatural}
=
0,
\]
where the last step follows because $\basisTwo^{(\generalNatural)}_1, \dots, \basisTwo^{(\generalNatural)}_{\generalNatural-1}$ are elements of $\plane_\generalNatural$. For item \ref{item:orthogonal-norm}, we use item \ref{item:orthogonal-inv} and observe
\[
\euclidNorm{\planeMap_\generalNatural \generalReal^{\generalNatural-1}}^2
=
\innerproduct{\planeMap_\generalNatural \generalReal^{\generalNatural-1}}{\planeMap_\generalNatural \generalReal^{\generalNatural-1}}
=
\innerproduct{\transpose{\planeMap_\generalNatural}\planeMap_\generalNatural \generalReal^{\generalNatural-1}}{\generalReal^{\generalNatural-1}}
=
\euclidNorm{\generalReal^{\generalNatural-1}}^2.
\]

For item \ref{item:orthogonal-max}, we use the facts that

\begin{align*}
\maximumNorm{\planeMap_{\generalNatural}\generalReal^{\generalNatural-1}}
&\leq
\maximumNorm{\planeMap_{\generalNatural}} \maximumNorm{\generalReal^{\generalNatural-1}}
\\
\maximumNorm{\planeMap_{\generalNatural}}
&=
\max_{1\leq \generalIndex \leq n}
  \sum_{j=1}^{\generalNatural-1}
    \absolutevalue{\planeMapElement^{(\generalNatural)}_{\generalIndex , j}}
\end{align*}
for all $\generalReal^{\generalNatural-1}\in \reals^{\generalNatural-1}$,
and show by induction that if $\generalNatural=2^\generalNaturalTwo + \generalNaturalTwo'$ with $\generalNaturalTwo' \in \{0, \dots, 2^\generalNaturalTwo-1\}$, we have, for all $\generalIndex \in \{1, \dots, \generalNatural\}$,
\begin{equation}
\label{eq:orthogonal-max-induction}
\absolutevalue{\planeMapElement^{(\generalNatural)}_{\generalIndex,1}}
+
\cdots
+
\absolutevalue{\planeMapElement^{(\generalNatural)}_{\generalIndex,\generalNatural-1}}
\leq
\frac{1}{\sqrt{2}^1}
+
\cdots
+
\frac{1}{\sqrt{2}^\generalNaturalTwo}
+
\frac{\generalNaturalTwo'}{2^\generalNaturalTwo}.
\end{equation}

This is stronger than item \ref{item:orthogonal-max} since $\generalNaturalTwo'/2^\generalNaturalTwo < 1$ and due to the known convergence behavior of the geometric series,
\[
\frac{1}{\sqrt{2}^1}
+
\cdots
+
\frac{1}{\sqrt{2}^\generalNaturalTwo}
<
\sum_{\generalNaturalTwo=1}^\infty
  \frac{1}{\sqrt{2}^\generalNaturalTwo}
=
\frac{1}{\sqrt{2}-1}.
\]
For $\generalNatural=2=2^1+0$, it is immediately clear from (\ref{eq:planemap-def-base}) that (\ref{eq:orthogonal-max-induction}) holds. For even $\generalNatural > 2$, we write $\generalNatural = 2^\generalNaturalTwo + \generalNaturalTwo'$ and $\generalNatural/2 = 2^{\generalNaturalTwo-1} + \generalNaturalTwo'/2$. By induction hypothesis and (\ref{eq:planemap-def-even}), for $\generalIndex \in \{1, \dots, \generalNatural/2\}$, we have
\[
\absolutevalue{\planeMapElement^{(\generalNatural)}_{\generalIndex,1}}
+
\cdots
+
\absolutevalue{\planeMapElement^{(\generalNatural)}_{\generalIndex,\generalNatural/2-1}}
\leq
\frac{1}{\sqrt{2}^1}
+
\cdots
+
\frac{1}{\sqrt{2}^{\generalNaturalTwo-1}}
+
\frac{\generalNaturalTwo'/2}{2^{\generalNaturalTwo-1}}
\]
and
\[
\planeMapElement^{(\generalNatural)}_{\generalIndex,\generalNatural/2}
=
\cdots
=
\planeMapElement^{(\generalNatural)}_{\generalIndex,\generalNatural-2}
=
0,
\]
and for $\generalIndex \in \{\generalNatural/2+1, \dots, \generalNatural\}$, we have
\[
\planeMapElement^{(\generalNatural)}_{\generalIndex,1}
=
\cdots
=
\planeMapElement^{(\generalNatural)}_{\generalIndex,\generalNatural/2-1}
=
0
\]
and
\[
\absolutevalue{\planeMapElement^{(\generalNatural)}_{\generalIndex,\generalNatural/2}}
+
\cdots
+
\absolutevalue{\planeMapElement^{(\generalNatural)}_{\generalIndex,\generalNatural-2}}
\leq
\frac{1}{\sqrt{2}}
+
\cdots
+
\frac{1}{\sqrt{2}^{\generalNaturalTwo-1}}
+
\frac{\generalNaturalTwo'/2}{2^{\generalNaturalTwo-1}}.
\]
In both cases, $\absolutevalue{\planeMapElement^{(\generalNatural)}_{\generalIndex,\generalNatural-1}} = 1/\sqrt{\generalNatural} \leq 1/\sqrt{2}^\generalNaturalTwo$, and hence,
\[
\absolutevalue{\planeMapElement^{(\generalNatural)}_{\generalIndex,1}}
+
\cdots
+
\absolutevalue{\planeMapElement^{(\generalNatural)}_{\generalIndex,\generalNatural-1}}
\leq
\frac{1}{\sqrt{2}^1}
+
\cdots
+
\frac{1}{\sqrt{2}^{\generalNaturalTwo-1}}
+
\frac{\generalNaturalTwo'/2}{2^{\generalNaturalTwo-1}}
+
\frac{1}{\sqrt{2}^\generalNaturalTwo}
=
\frac{1}{\sqrt{2}^1}
+
\cdots
+
\frac{1}{\sqrt{2}^{\generalNaturalTwo}}
+
\frac{\generalNaturalTwo'}{2^{\generalNaturalTwo}},
\]
so we have shown (\ref{eq:orthogonal-max-induction}). For odd $\generalNatural = 2^\generalNaturalTwo + \generalNaturalTwo'$, we conclude from the induction hypothesis and (\ref{eq:planemap-def-odd}) that for $\generalIndex \in \{1, \dots, \generalNatural-1\}$,
\[
\absolutevalue{\planeMapElement^{(\generalNatural)}_{\generalIndex,1}}
+
\cdots
+
\absolutevalue{\planeMapElement^{(\generalNatural)}_{\generalIndex,\generalNatural-1}}
\leq
\frac{1}{\sqrt{2}^1}
+
\cdots
+
\frac{1}{\sqrt{2}^\generalNaturalTwo}
+
\frac{\generalNaturalTwo'-1}{2^\generalNaturalTwo}
+
\frac{1}{\sqrt{(\generalNatural-1)^2}}
\leq
\frac{1}{\sqrt{2}^1}
+
\cdots
+
\frac{1}{\sqrt{2}^\generalNaturalTwo}
+
\frac{\generalNaturalTwo'}{2^\generalNaturalTwo}.
\]
For the last row sum, we note that
\[
\absolutevalue{\planeMapElement^{(\generalNatural)}_{\generalNatural,1}}
+
\cdots
+
\absolutevalue{\planeMapElement^{(\generalNatural)}_{\generalNatural,\generalNatural-1}}
=
\absolutevalue{\planeMapElement^{(\generalNatural)}_{\generalNatural,\generalNatural-1}}
=
\sqrt{\frac{\generalNatural-1}{\generalNatural}}
\leq
1
\]
and since we have dealt with the base case $\generalNatural=2$ separately, we either have $\generalNaturalTwo \geq 2$ or $\generalNaturalTwo = \generalNaturalTwo' = 1$. In both cases the right hand side of (\ref{eq:orthogonal-max-induction}) is clearly greater than $1$. Therefore, we have shown (\ref{eq:orthogonal-max-induction}).

\subsection{Technical Lemmas Used in the Proof of Lemma~\ref{lemma:converse}}
\label{appendix:technical-lemmas-converse}
In this appendix, we state two technical lemmas that are necessary to prove the converse result of Lemma~\ref{lemma:converse}.

\begin{lemma}
\label{lemma:convolution-density}
Let $\generalRV, \generalRVTwo$ be independent and real-valued random variables. Suppose that $\generalRVTwo$ is absolutely continuous with respect to the Lebesgue measure. Then $\generalRV + \generalRVTwo$ is absolutely continuous with respect to the Lebesgue measure.
\end{lemma}
\begin{proof}
This is an immediate consequence of \cite[Theorem 2.1.11]{durrett2010probability}.
\end{proof}

\begin{lemma}
\label{lemma:information-of-sum}
Let $\generalRV_1, \dots, \generalRV_\generalNatural, \generalRVTwo$ be random variables, and let $\generalRV := \generalRV_1 + \dots + \generalRV_\generalNatural + \generalRVTwo + \generalReal$, where $\generalReal \in
\reals$ is deterministic. Assume that the tuple $(\generalRV_1, \dots, \generalRV_\generalNatural)$ is stochastically independent of $\generalRVTwo$ and that $\generalRVTwo$ is absolutely continuous with respect to the Lebesgue measure. Then
\[
\mutualInformationRV{\generalRV_1, \dots, \generalRV_\generalNatural}{\generalRV}
=
\mutualInformationRV{\generalRV_1 + \dots + \generalRV_\generalNatural}{\generalRV}.
\]
\end{lemma}
\begin{proof}
First note that by Lemma~\ref{lemma:convolution-density}, $\generalRV$ is absolutely continuous with respect to the Lebesgue measure. Hence, the differential entropies $\diffEntropyRV{\generalRV}$ and $\diffEntropyRV{\generalRVTwo}$ exist and due to the equalities proven below, also the conditional differential entropies that appear in this proof. We have
\[
\mutualInformationRV{\generalRV_1, \dots, \generalRV_\generalNatural}{\generalRV}
=
\diffEntropyRV{\generalRV}
-
\diffEntropyRVCond{\generalRV}{\generalRV_1, \dots, \generalRV_\generalNatural}
\overset{(a)}{=}
\diffEntropyRV{\generalRV}
-
\diffEntropyRVCond{\generalRVTwo}{\generalRV_1, \dots, \generalRV_\generalNatural}
\overset{(b)}{=}
\diffEntropyRV{\generalRV}
-
\diffEntropyRV{\generalRVTwo},
\]
where (a) is because under the condition $(\generalRV_1, \dots, \generalRV_\generalNatural)$, we have that $\generalRVTwo$ is a shifted version of $\generalRV$, and (b) holds due to the independence assumption. Clearly, this derivation also holds with $(\generalRV_1, \dots, \generalRV_\generalNatural)$ replaced by $\generalRV_1 + \dots + \generalRV_\generalNatural$, obtaining equality with the same term $\diffEntropyRV{\generalRV}-\diffEntropyRV{\generalRVTwo}$ on the right hand side.
\end{proof}

\bibliographystyle{unsrt}
\bibliography{hybrid-analog-digital-references}

\begin{thebibliography}{10}

\bibitem{gastpar2008uncoded}
Michael Gastpar.
\newblock Uncoded transmission is exactly optimal for a simple {G}aussian
  ``sensor'' network.
\newblock {\em IEEE Transactions on Information Theory}, 54(11):5247--5251,
  2008.

\bibitem{nazer2007computation}
Bobak Nazer and Michael Gastpar.
\newblock Computation over multiple-access channels.
\newblock {\em IEEE Transactions on Information Theory}, 53(10):3498--3516,
  2007.

\bibitem{gastpar2003source}
Michael Gastpar and Martin Vetterli.
\newblock Source-channel communication in sensor networks.
\newblock In Feng Zhao and Leonidas Guibas, editors, {\em Information
  Processing in Sensor Networks}, pages 162--177. Berlin Heidelberg, Germany,
  2003.

\bibitem{goldenbaum2009function}
Mario Goldenbaum, Sławomir Stańczak, and Michał Kaliszan.
\newblock On function computation via wireless sensor multiple-access channels.
\newblock In {\em 2009 IEEE Wireless Communications and Networking Conference},
  pages 1--6. IEEE, 2009.

\bibitem{goldenbaum2013harnessing}
Mario Goldenbaum, Holger Boche, and S{\l}awomir Sta{\'n}czak.
\newblock Harnessing interference for analog function computation in wireless
  sensor networks.
\newblock {\em IEEE Transactions on Signal Processing}, 61(20):4893--4906,
  2013.

\bibitem{goldenbaum2014nomographic}
Mario Goldenbaum, Holger Boche, and Slawomir Sta{\'n}czak.
\newblock Nomographic functions: Efficient computation in clustered {G}aussian
  sensor networks.
\newblock {\em IEEE Transactions on Wireless Communications}, 14(4):2093--2105,
  2014.

\bibitem{nazer2011compute}
Bobak Nazer and Michael Gastpar.
\newblock Compute-and-forward: Harnessing interference through structured
  codes.
\newblock {\em IEEE Transactions on Information Theory}, 57(10):6463--6486,
  2011.

\bibitem{cai2018modulation}
Songfu Cai and Vincent K.~N. Lau.
\newblock Modulation-free {M2M} communications for mission-critical
  applications.
\newblock {\em IEEE Transactions on Signal and Information Processing over
  Networks}, 4(2):248--263, 2018.

\bibitem{amiri2020machine}
Mohammad~Mohammadi Amiri and Deniz G{\"u}nd{\"u}z.
\newblock Machine learning at the wireless edge: Distributed stochastic
  gradient descent over-the-air.
\newblock {\em IEEE Transactions on Signal Processing}, 68:2155--2169, 2020.

\bibitem{wang2022over}
Zhibin Wang, Yapeng Zhao, Yong Zhou, Yuanming Shi, Chunxiao Jiang, and Khaled~B
  Letaief.
\newblock Over-the-air computation: Foundations, technologies, and
  applications.
\newblock {\em IEEE Internet of Things Journal}, 11(14):24634--24658, 2024.

\bibitem{sahin2023survey}
Alphan {\c{S}}ahin and Rui Yang.
\newblock A survey on over-the-air computation.
\newblock {\em IEEE Communications Surveys \& Tutorials}, 2023.

\bibitem{qi2020integration}
Qiao Qi, Xiaoming Chen, Caijun Zhong, and Zhaoyang Zhang.
\newblock Integration of energy, computation and communication in 6{G} cellular
  internet of things.
\newblock {\em IEEE Communications Letters}, 24(6):1333--1337, 2020.

\bibitem{qi2021integrated}
Qiao Qi, Xiaoming Chen, Caijun Zhong, and Zhaoyang Zhang.
\newblock Integrated sensing, computation and communication in {B5G} cellular
  internet of things.
\newblock {\em IEEE Transactions on Wireless Communications}, 20(1):332--344,
  2021.

\bibitem{du2021interference}
Yuhang Du, Lukuan Xing, Yong Zhou, and Yuanming Shi.
\newblock Interference management for over-the-air computation and cellular
  coexistence systems.
\newblock In {\em 2021 IEEE Globecom Workshops (GC Wkshps)}, pages 1--5. IEEE,
  2021.

\bibitem{qi2022integrating}
Qiao Qi, Xiaoming Chen, Ata Khalili, Caijun Zhong, Zhaoyang Zhang, and Derrick
  Wing~Kwan Ng.
\newblock Integrating sensing, computing, and communication in 6{G} wireless
  networks: Design and optimization.
\newblock {\em IEEE Transactions on Communications}, 70(9):6212--6227, 2022.

\bibitem{ni2022integrating}
Wanli Ni, Yuanwei Liu, Zhaohui Yang, Hui Tian, and Xuemin Shen.
\newblock Integrating over-the-air federated learning and non-orthogonal
  multiple access: What role can {RIS} play?
\newblock {\em IEEE Transactions on Wireless Communications},
  21(12):10083--10099, 2022.

\bibitem{mittal2002hybrid}
Udar Mittal and Nam Phamdo.
\newblock Hybrid digital-analog ({HDA}) joint source-channel codes for
  broadcasting and robust communications.
\newblock {\em IEEE Transactions on Information Theory}, 48(5):1082--1102,
  2002.

\bibitem{shamai1998systematic}
Shlomo Shamai, Sergio Verd{\'u}, and Ram Zamir.
\newblock Systematic lossy source/channel coding.
\newblock {\em IEEE Transactions on Information Theory}, 44(2):564--579, 1998.

\bibitem{yao2009hybrid}
Sha Yao and Mikael Skoglund.
\newblock Hybrid digital-analog relaying for cooperative transmission over slow
  fading channels.
\newblock {\em IEEE Transactions on Information Theory}, 55(3):944--951, 2009.

\bibitem{lapidoth2010sending}
Amos Lapidoth and Stephan Tinguely.
\newblock Sending a bivariate {G}aussian over a {G}aussian {MAC}.
\newblock {\em IEEE Transactions on Information Theory}, 56(6):2714--2752,
  2010.

\bibitem{minero2015unified}
Paolo Minero, Sung~Hoon Lim, and Young-Han Kim.
\newblock A unified approach to hybrid coding.
\newblock {\em IEEE Transactions on Information Theory}, 61(4):1509--1523,
  2015.

\bibitem{frey2021over-tsp}
Matthias Frey, Igor Bjelaković, and Sławomir Stańczak.
\newblock Over-the-air computation in correlated channels.
\newblock {\em IEEE Transactions on Signal Processing}, 69:5739--5755, 2021.

\bibitem{vershynin2018high}
Roman Vershynin.
\newblock {\em High-Dimensional Probability: An Introduction with Applications
  in Data Science}, volume~47.
\newblock Cambridge University Press, Cambridge, 2018.

\bibitem{billingsley1995probability}
Patrick Billingsley.
\newblock {\em Probability and Measure}.
\newblock Wiley, Hoboken, New Jersey, 3rd edition, 1995.

\bibitem{mamandipoor2014capacity}
Babak Mamandipoor, Kamyar Moshksar, and Amir~K. Khandani.
\newblock Capacity-achieving distributions in gaussian multiple access channel
  with peak power constraints.
\newblock {\em IEEE Transactions on Information Theory}, 60(10):6080--6092,
  2014.

\bibitem{smith1971information}
Joel~G. Smith.
\newblock The information capacity of amplitude-and variance-constrained scalar
  {G}aussian channels.
\newblock {\em Information and Control}, 18(3):203--219, 1971.

\bibitem{cover2006elements}
Thomas~M. Cover and Joy~A. Thomas.
\newblock {\em Elements of Information Theory}.
\newblock John Wiley\& Sons, Hoboken, New Jersey, 2nd edition, 2006.

\bibitem{han2006information}
Te~Sun Han.
\newblock An information-spectrum approach to capacity theorems for the general
  multiple-access channel.
\newblock {\em IEEE Transactions on Information Theory}, 44(7):2773--2795,
  2006.

\bibitem{elgamal2011network}
Abbas El~Gamal and Young-Han Kim.
\newblock {\em Network Information Theory}.
\newblock Cambridge University Press, Cambridge, UK, 2011.

\bibitem{middleton1993elements}
D.~Middleton and A.~D. Spaulding.
\newblock Elements of weak signal detection in non-{G}aussian noise
  environments.
\newblock In V.~Poor and J.~B. Thomas, editors, {\em Advances in Statistical
  Signal Processing}, volume~2, pages 137--215. JAI Press, 1993.

\bibitem{ang2020robust}
Fan Ang, Li~Chen, Nan Zhao, Yunfei Chen, Weidong Wang, and F~Richard Yu.
\newblock Robust federated learning with noisy communication.
\newblock {\em IEEE Transactions on Communications}, 68(6):3452--3464, 2020.

\bibitem{wei2022federated}
Xizixiang Wei and Cong Shen.
\newblock Federated learning over noisy channels: Convergence analysis and
  design examples.
\newblock {\em IEEE Transactions on Cognitive Communications and Networking},
  8(2):1253--1268, 2022.

\bibitem{rajagopal2010network}
Ram Rajagopal and Martin~J Wainwright.
\newblock Network-based consensus averaging with general noisy channels.
\newblock {\em IEEE Transactions on Signal Processing}, 59(1):373--385, 2010.

\bibitem{agrawal2024distributed}
Navneet Agrawal, Renato~LG Cavalcante, Masahiro Yukawa, and S{\l}awomir
  Sta{\'n}czak.
\newblock Distributed convex optimization ``over-the-air'' in dynamic
  environments.
\newblock {\em IEEE Transactions on Signal and Information Processing over
  Networks}, 10:610--625, 2024.

\bibitem{nazer2024computation}
Bobak Nazer and Krishna Narayanan.
\newblock Computation selection: Scheduling users to enable over-the-air
  federated learning.
\newblock In {\em 2024 IEEE International Symposium on Information Theory
  (ISIT)}, pages 1635--1640. IEEE, 2024.

\bibitem{durrett2010probability}
Rick Durrett.
\newblock {\em Probability: Theory and Examples}.
\newblock Cambridge Series in Statistical and Probabilistic Mathematics.
  Cambridge University Press, Cambridge, UK, 4th edition, 2010.

\end{thebibliography}
\end{document}